\documentclass[10pt,a4paper]{article}
\usepackage[cp1251]{inputenc}
    \usepackage[english]{babel}
    \usepackage{amsmath,amssymb,amsthm}
    \usepackage{algorithm}
\usepackage[noend]{algpseudocode}
\usepackage[table]{xcolor}
\usepackage{tabularray}
\usepackage{tabu}
\usepackage{bm}
\usepackage{tikz}
\usepackage{booktabs}
\usepackage{ifthen}
\usepackage[pdftex, bookmarks=true]{hyperref}
\usepackage{threeparttable}
\usetikzlibrary{positioning,shapes.gates.logic.US}
\usepackage{mathtools}

\usetikzlibrary{arrows}
\usetikzlibrary{graphs}

\usepackage{tikz-qtree}
\usetikzlibrary{math}
\usetikzlibrary{calc}
\usepackage{xypic}
\usetikzlibrary{shapes}
\usetikzlibrary{arrows.meta}
\DeclareMathOperator{\Clo}{Clo}

\DeclareMathOperator{\Pol}{Pol}

\DeclareMathOperator{\Image}{Im}
\DeclareMathOperator{\proj}{pr}

\DeclareMathOperator{\squeezer}{Squeezer}

\newcommand{\permd}[1]{
\langle{#1}\rangle
}

\newcommand{\absorbseq}{\trianglelefteq}

\DeclareMathOperator{\centers}{\mathcal C}

\DeclareMathOperator{\Sg}{Sg}

\DeclareMathOperator{\CSP}{CSP}

\DeclareMathOperator{\PCSP}{PCSP}

\DeclareMathOperator{\trace}{tr}

\DeclareMathOperator{\arity}{ar}
  \DeclareMathOperator\Perm{Perm}

  \DeclareMathOperator\FiniteCSP{FiniteCSP}
  
\newcommand{\TL}{
\mathcal{L}
}

\newcommand{\TD}{
\mathcal{D}
}

\newcommand{\TPC}{
\mathcal{PC}
}
\newcommand{\TS}{
\mathcal{S}
}
\newcommand{\TBA}{
\mathcal{BA}
}

\newcommand{\TC}{
\mathcal{C}
}

\newcommand{\SArc}{
\mathrm{SArc}
}

\newcommand{\SBLPAIP}{
\mathrm{S(BLP+AIP)}
}

\long\def\rr#1{\bgroup\color{red} #1\egroup}

\long\def\bb#1{\bgroup\color{blue!30!black} #1\egroup}

\newcommand{\gxy}[2]{
({\xs+(#1+(0.5*(#2)))*\celllength},{\ys+(\s*(#2))*\celllength})
}

\theoremstyle{plain}
\newtheorem{thm}{Theorem}[section]
\newtheorem{claim}{Claim}[section]

\newtheorem{question}{Open question}

\newtheorem{lem}[thm]{Lemma}
\newtheorem{remark}{Remark}[section]
\newtheorem{obs}{Observation}[section]

\newtheorem{cor}[thm]{Corollary}

\newcommand{\densegrid}[1]{
\foreach \x in {0,...,\xsize}
{
    \draw [#1] \gxy{\x}{0}--\gxy{\x}{\ysize};
}
\foreach \y in {0,...,\ysize}
{
    \draw [#1] \gxy{0}{\y}--\gxy{\xsize}{\y};
}
\ifthenelse{\lengthtest{\xsize pt < \ysize pt}}{
    \foreach \y in {1,...,\xsize}
    {
        \draw [#1] \gxy{0}{\y}--\gxy{\y}{0};
    }
    \foreach \z in {\xsize,...,\ysize}
    {
        \pgfmathsetmacro{\y}{\z-\xsize}
        \draw [#1] \gxy{0}{\z}--\gxy{\xsize}{\y};
    }
    \foreach \x in {1,...,\xsize}
    {
        \pgfmathsetmacro{\y}{\ysize-\xsize+\x}    
        \draw [#1] \gxy{\x}{\ysize}--\gxy{\xsize}{\y};
    }
}{    
    \foreach \x in {1,...,\ysize}
    {
        \draw [#1] \gxy{\x}{0}--\gxy{0}{\x};
    }
    \foreach \z in {\ysize,...,\xsize}
    {
        \pgfmathsetmacro{\x}{\z-\ysize}
        \draw [#1] \gxy{\z}{0}--\gxy{\x}{\ysize};
    }
    \foreach \y in {1,...,\ysize}
    {
        \pgfmathsetmacro{\x}{\xsize-\ysize+\y}    
        \draw [#1] \gxy{\x}{\ysize}--\gxy{\xsize}{\y};
    }
}
}

\newcommand{\threegrid}{
\foreach \x in {-10,...,\xsize}
{
    \foreach \y in {0,...,\ysize}
    {
        \pgfmathsetmacro{\lx}{\x +0.5*\y+0.0001+0.5}
        \pgfmathsetmacro{\rx}{\x +0.5*\y+1-0.0001-0.5}
        \ifthenelse{\lengthtest{\lx pt > 0 pt} \AND \lengthtest{\rx pt< \xsize pt}}{
            \draw [line width=1pt,black] \gxy{\x}{\y}--\gxy{\x+1}{\y};
        }{}
        \pgfmathsetmacro{\lx}{\x +0.5*\y+0.0001+0.5}
        \pgfmathsetmacro{\rx}{\x +0.5*\y+0.5-0.0001-0.5}        
        \ifthenelse{\lengthtest{\lx pt > 0 pt} \AND \lengthtest{\rx pt< \xsize pt}
        \AND \lengthtest{\y pt< \ysize pt}}{
            \draw [line width=1pt,black] \gxy{\x}{\y}--\gxy{\x}{\y+1};
        }{}
        \pgfmathsetmacro{\lx}{\x +0.5*\y-0.5+0.0001+0.5}
        \pgfmathsetmacro{\rx}{\x +0.5*\y-0.0001-0.5}        
        \ifthenelse{\lengthtest{\lx pt > 0 pt} \AND \lengthtest{\rx pt< \xsize pt}
        \AND \lengthtest{\y pt< \ysize pt}}{
            \draw [line width=1pt,black] \gxy{\x}{\y}--\gxy{\x-1}{\y+1};
        }{}        
    }
}
}

\newcommand{\drawE}[5]{
        \pgfmathsetmacro{\xonel}{#1+#2*0.5+0.5}
        \pgfmathsetmacro{\xtwol}{#3+#4*0.5+0.5}        
        \pgfmathsetmacro{\xoner}{#1+#2*0.5-1}
        \pgfmathsetmacro{\xtwor}{#3+#4*0.5-1}
        \pgfmathsetmacro{\yoner}{#4-1}        
        \pgfmathsetmacro{\ytwor}{#4-1}
        \ifthenelse{\lengthtest{#2 pt > -0.01 pt} \AND \lengthtest{#4 pt > -0.01 pt} \AND
        \lengthtest{\yoner pt < \ysize pt} \AND \lengthtest{\ytwor pt < \ysize pt} \AND
        \lengthtest{\xonel pt > -0.001 pt} \AND \lengthtest{\xoner pt< \xsize pt}
        \AND \lengthtest{\xtwol pt > -0.001 pt} \AND \lengthtest{\xtwor pt< \xsize pt}}{
            \draw [#5] \gxy{#1}{#2}--\gxy{#3}{#4};
        }{}
}

\newcommand{\SinglArcCons}{\mathrm{SinglArcCons}}
\newcommand{\Singl}{\mathrm{Singl}}
\newcommand{\CSingl}{\mathrm{CSingl}}
\newcommand{\SinglAlgorithm}{\mathrm{SinglAlg}}

\newcommand{\Algorithm}{\mathrm{Alg}}

\newcommand{\BLP}{\mathrm{BLP}}
\newcommand{\AIP}{\mathrm{AIP}}
\newcommand{\LIN}[2]{{#1}\text{-}\mathrm{LIN}(#2)}
\newcommand{\boldLIN}[2]{\mathbf{#1}\text{\textbf{-LIN}}\mathbf{(#2)}}

\newcommand{\ArcCons}{\mathrm{AC}}
\newcommand{\ArcConsistency}{\mbox{\textsc{AC}}}
\newcommand{\CSinglAlg}{\mbox{\textsc{CSinglAlg}}}
\newcommand{\SinglAlg}{\mbox{\textsc{SinglAlg}}}
\newcommand{\Alg}{\mbox{\textsc{Alg}}}

\newcommand{\ReduceDomain}{\mbox{\textsc{ReduceDomain}}}
\newcommand{\ChangeConstraint}{\mbox{\textsc{ChangeConstraint}}}

\newtheorem*{THMSinglArcCharacterizationTHM}{Theorem~\ref{THMSinglArcCharacterization}}

\newtheorem*{THMSinglBLPAIPTHM}{Theorem~\ref{THMSinglBLPAIP}}

\newtheorem*{THMMyTemporalClassificationTHM}{Theorem~\ref{THMMyTemporalClassification}}

\newtheorem*{THMMainSmallerThanEightTHM}{Theorem~\ref{THMMainSmallerThanEight}}

\newtheorem*{LEMRelationsForD4LEM}{Lemma~\ref{LEMRelationsForD4}}

\newtheorem*{LEMNoPaletteSymmetricDFourLEM}{Lemma~\ref{LEMNoPaletteSymmetricDFour}}

\newtheorem*{THMExistenceWBSForSmallAlgebrasTHM}{Theorem~\ref{THMExistenceWBSForSmallAlgebras}}

\title{Singleton algorithms for the Constraint Satisfaction Problem}
\author{Dmitriy Zhuk\thanks{The author is funded by 
the Czech Science Foundation project 25-16324S 
and the European Union (ERC, POCOCOP, 101071674). Views and opinions expressed are however those of the author(s) only and do not necessarily reflect those of the European Union or the European Research Council Executive Agency. Neither the European Union nor the granting authority can be held responsible for them.}
}

\begin{document}

\maketitle

\begin{abstract}
A natural strengthening of an algorithm for the (promise) constraint satisfaction problem is its singleton version: 
we first fix a variable
to an element from its domain, 
then run the algorithm, and remove the element from the domain if the answer is 
negative.
Using the Hales-Jewett theorem, we characterize the power of the singleton versions of standard universal 
algorithms for the (promise) CSP over a fixed template in terms of the existence of  
polymorphisms with certain symmetries, 
which we call palette symmetric polymorphisms. 
By proving the existence of such polymorphisms 
we establish that the singleton version of the BLP+AIP algorithm solves 
all (multi-sorted) tractable CSPs over domains of size at most 7. 
We further show that already for domain size 8 there exists a relational structure arising from the dihedral group 
$\mathbf D_4$
that does not admit palette symmetric polymorphisms and cannot be solved by singleton BLP+AIP. 
By providing concrete CSP templates, we illustrate the limitations of linear programming,
 the power of the singleton versions, 
 and 
the elegance of palette symmetric polymorphisms.
Among tractable temporal templates, we exhibit a structure demonstrating that finiteness is crucial for the Hales-Jewett argument; nevertheless, by introducing generalized palette polymorphisms we establish tractability for each such template.
\end{abstract}




\section{Introduction}

\subsection{Constraint Satisfaction Problem}

The \emph{Constraint Satisfaction Problem (CSP)} is the problem of deciding whether there is an assignment to a set of variables
subject to some specified constraints. 
Formally, the \emph{Constraint Satisfaction Problem} is defined as a triple $\langle X , D , \mathbf{C} \rangle$,
where
\begin{itemize}
\item
$X$ is a finite set of variables,
\item
$D$ is a mapping that assigns a set $D_{x}$ to every variable $x\in X$, 
called \emph{the domain of $x$}; 
\label{CSPDefinition}
\item
$\mathbf{C}$ is a set of constraints 
$C = R(x_1,\dots,x_m)$,  where $m\in\mathbb N$, $x_{1},\dots,x_{m}\in X$,
and
$R\subseteq D_{x_1}\times\dots\times D_{x_m}$, called \emph{constraint relation}.
\end{itemize}
By $\FiniteCSP$ we denote the decision problem whose input is 
a triple $\langle X, D,\mathbf C\rangle$, 
where all the domains are finite 
and constraint relations are given by listing the tuples.

In general, the problem $\FiniteCSP$ is NP-complete,  and to get 
tractable cases we usually restrict the set of allowed constraints. 
For a finite relational structure $\mathbb A$, also called 
\emph{the constraint language} or \emph{template},
by $\CSP(\mathbb A)$ we denote the restricted version of 
$\FiniteCSP$ where each domain is $A$ and the constraint relations are from $\mathbb A$.
Alternatively, the input 
$\langle X , D , \mathbf{C} \rangle$ of $\CSP(\mathbb A)$
can be viewed as a relational structure $\mathbb X$ in the same signature $\sigma$ as $\mathbb A$
on the domain $X$ such that  
$(x_1,\dots,x_m)\in R^{\mathbb X}$ for $R\in \sigma$ if and only if 
$R(x_1,\dots,x_m)$ is a constraint from $\mathbf C$.
Then $\CSP(\mathbb A)$ is equivalent to checking whether there exists 
a homomorphism $\mathbb X\to \mathbb A$.
Nevertheless, we prefer not to fix the target structure $\mathbb A$ and work with a triple $\langle X , D , \mathbf{C} \rangle$ as
standard algorithms for the CSP change the constraint relations 
and even domains of the variables.
We have the following characterization of the complexity of $\CSP(\mathbb A)$ known as the CSP Dichotomy Conjecture.

\begin{thm}[\cite{zhuk2020proof,ZhukFVConjecture,BulatovFVConjecture,BulatovProofCSP}]\label{mainthm}
Suppose $\mathbb A$ is a finite relational structure.
Then 
$\CSP(\mathbb A)$ is solvable in polynomial time if 
there exists a WNU polymorphism of $\mathbb A$; 
$\CSP(\mathbb A)$ is NP-complete otherwise.
\end{thm}

The NP-hardness in the above theorem follows from 
\cite{bulatov2001algebraic,CSPconjecture}
and \cite{MarotiMcKenzie}. 
The essential part of each proof
of the CSP Dichotomy Conjecture is 
a polynomial algorithm that works for all tractable cases.
Both polynomial algorithms \cite{zhuk2020proof,ZhukFVConjecture,BulatovFVConjecture,BulatovProofCSP} we know are not universal in the sense that they work only for a fixed structure $\mathbb A$, and one of the main open questions in the area is the existence of 
a universal algorithm.

A natural generalization of the CSP is \emph{the Promise Constraint Satisfaction Problem (PCSP)}.
For two relational structures $\mathbb A$ and $\mathbb B$ in the same signature  such that 
$\mathbb A\to\mathbb B$, by $\PCSP(\mathbb A,\mathbb B)$ we denote the following decision problem. 
Given a structure $\mathbb X$ in the same signature. 
Distinguish between the case when $\mathbb X\to \mathbb A$ 
and the case when $\mathbb X\not\to \mathbb B$.
Thus, we have an additional promise about the input that one of these two conditions holds.
The complexity of $\PCSP(\mathbb A,\mathbb B)$  is widely open even for relational structures $\mathbb A$ and $\mathbb B$ on a 2-element set \cite{krokhin2022invitation,PCSPBible,austrin2026usefulness}.

\subsection{Universal algorithms and their singleton versions}

We would like to find an algorithm that is universal not only in the sense that it works for any constraint language $\mathbb A$, but also in the following sense. 
An algorithm $\mathfrak A$ is called \emph{universal} if:
\begin{enumerate}
    \item[(1)] $\mathfrak A$ returns ``Yes'' or ``No'' on  any instance $\mathcal I$ of $\FiniteCSP$;
    \item[(2)] $\mathfrak A$ runs in polynomial time;
    \item[(3)] if $\mathfrak A$ returns ``No'' on an instance $\mathcal I$,  then 
    $\mathcal I$ has no solutions.
\end{enumerate}

Thus, the No-answer of a universal algorithm is always correct, 
while the Yes-answer may be wrong.
For example, the trivial algorithm that always returns ``Yes'' satisfies the above conditions.
Since $\FiniteCSP$ is NP-complete, 
we cannot expect to find a universal algorithm that 
always returns the correct answer.
Our aim, therefore, is to find a universal algorithm 
that works correctly on the largest possible 
class of inputs from $\FiniteCSP$.
This class can be specified either by restricting the 
constraint language or by imposing an additional promise on the input.
We envision the following major steps toward the most powerful universal algorithm.

\begin{itemize}
    \item \textbf{CSP over a fixed template.} The first goal is 
an algorithm that works correctly 
on all instances of $\CSP(\mathbb A)$, whenever 
$\CSP(\mathbb A)$ is solvable in polynomial time (by a non-universal algorithm).
\item \textbf{Perfect Matching Problem.}
In addition, we expect the algorithm to solve the perfect matching problem correctly.
This problem can be viewed as a CSP on $\{0,1\}$
over the relations 1IN2, 1IN3, 1IN4, ...
with the additional structural restriction that every variable appears exactly twice. We highlight this problem because it is a particularly elegant tractable CSP for which all known algorithms are non-universal and designed specifically for this problem.

    \item \textbf{PCSP over a fixed template.} The ultimate goal is an algorithm that solves all instances of
    $\PCSP(\mathbb A;\mathbb B)$ that are solvable in polynomial time.
    That is, the algorithm should return ``No'' on any instance $\mathbb X$ of $\CSP(\mathbb A)$ whenever $\mathbb X\not\to\mathbb B$. 
\end{itemize}

In the paper we consider classic universal algorithms
such as the arc-consistency (AC) algorithm, 
Basic Linear Programming (BLP), 
Affine Integer Programming (AIP), 
and their combinations (see Section \ref{SECTIONUniversalAlgorithms}).

The main role in this paper is played by the singleton versions of  the algorithms. Such algorithms also appear in the literature under the name ``cohomological'' \cite{CohomologicalAlgorithms,lichter2024limitations}, 
but we prefer the term singleton, since it appeared earlier \cite{singletonACFirstPaper,SingletonConsistencies} and, in our opinion, better reflects the underlying idea.
For a universal algorithm $\mathfrak A$, by  
$\CSingl\mathfrak A$ we denote the algorithm that 
gradually reduces constraint relations in the following way.
It fixes a constraint to some tuple from the constraint relation, then runs the algorithm $\mathfrak A$, and
removes the tuple from the constraint if the answer is negative. 
It repeats this for every constraint and every tuple 
until no more tuples can be removed. 
If some constraint relation becomes empty, it returns ``No''.
Otherwise, it returns ``Yes''.

A slight modification of this algorithm 
is $\Singl\mathfrak A$, 
where instead of fixing a constraint to a tuple, it fixes a variable to 
a value and removes the value from the domain if the 
answer of $\mathfrak A$ is negative.
For most of the algorithms, 
$\Singl\mathfrak A$ is not stronger than 
$\CSingl\mathfrak A$, 
but this is not obvious for $\Singl\AIP$, as 
$\Singl$-algorithm changes all the constraints containing 
a variable $x$ when 
removing an element from the domain of $x$ (see Remark \ref{REMAEKCSinglvsSingl}).
The precise definitions 
are given in 
Section \ref{SUBSECTIONPaletteOperationDefinition}.

There is another natural way to strengthen a universal algorithm. Instead of arc consistency, one may consider higher levels of local consistency, such as $k$-consistency \cite{BartoKozikLocalConsistency}. One can also combine this with AIP by first enforcing some level of consistency and then applying the AIP relaxation to partial solutions \cite{DalmauOprsalConjecture}. Finally, one can further strengthen such algorithms by considering the singleton (cohomological) version \cite{CohomologicalAlgorithms}.
A major open question in the area is whether some combination of a sufficient level of consistency and a sufficient level of linear relaxation can solve all tractable CSPs. Studying the limitations of such algorithms is currently a very active topic \cite{DalmauOprsalConjecture,ciardo2023hierarchies,lichter2024limitations,ChanHawTOFoolHierarchies}. Only recently, the seemingly simple conjecture of Dalmau and Opr\v sal \cite{DalmauOprsalConjecture} that an affine relaxation above $k$-consistency can solve all tractable CSPs was disproved \cite{lichter2024limitations}. 
Several recent papers demonstrate how such algorithms can be fooled by random instances and regular graphs \cite{ChanHawTOFoolHierarchies,ChanHawTOFoolHierarchiesTwo,conneryd2026lower,lichter2024limitations}. Since higher levels and hierarchies are outside the scope of the present paper, we do not define them here and instead refer the reader to the above references.


\subsection{Main results}\label{SUBSECTIONMainResults}

As shown in \cite{PCSPBible},
the polymorphisms $\Pol(\mathbb A,\mathbb B)$ not only characterize 
the complexity of $\PCSP(\mathbb A,\mathbb B)$ 
but also determine which 
algorithms can solve it. 
For instance, the arc-consistency algorithm solves $\PCSP(\mathbb A,\mathbb B)$ if and only if $\Pol(\mathbb A,\mathbb B)$ contains totally symmetric functions of all arities. 
Similarly, BLP and AIP are characterized 
by symmetric and alternating polymorphisms, respectively.

We characterize the power of the singleton versions of such algorithms in terms of the existence of sufficiently symmetric polymorphisms in $\Pol(\mathbb A,\mathbb B)$. 
This characterization arises from minions that capture the algorithms; these minions are quite straightforward to describe, but deriving complexity results directly from them is rather difficult.
We observe that, after mapping such a minion into $\Pol(\mathbb A,\mathbb B)$, 
when $\mathbb A$ and $\mathbb B$ are finite, the set $\Pol(\mathbb A,\mathbb B)$ may contain more symmetric functions than the original minion.
This is due to the fact that many objects must be collapsed when passing from a locally infinite minion to a locally finite one.
When the obtained symmetric polymorphisms are sufficiently strong to round the relaxed solution, this leads to a characterization.
The local finiteness of $\Pol(\mathbb A,\mathbb B)$ 
is the key property underlying our approach.
To the best of our knowledge, this is the first result that uses 
it
to obtain a characterization or derive complexity results.
We believe that this idea has the potential to become a general framework for studying the power of algorithms.

Our classification of singleton algorithms is collected in Table \ref{TableCharacterization}, 
which constitutes the first main result of the paper.
For every algorithm or combination of algorithms from 
Section \ref{SECTIONUniversalAlgorithms},
we provide a criterion for 
solving $\PCSP(\mathbb A,\mathbb B)$.
This criterion is formulated in terms of 
palette symmetric polymorphisms, which are formally defined in Section \ref{SUBSECTIONPaletteOperationDefinition}. 
Similar polymorphisms first appeared in  \cite{ArcConsAndFriends} in the context of the singleton arc-consistency algorithm.
To illustrate the idea of such polymorphisms, consider
the singleton version of BLP. 
As shown in Table \ref{TableCharacterization},
this algorithm solves $\PCSP(\mathbb A,\mathbb B)$ if and only if 
$\Pol(\mathbb A,\mathbb B)$ contains an  
$(\underbrace{\overline{m},\dots,\overline{m}}_{n})$-palette symmetric function for all $m$ and $n$.
Such a function has arity $m\cdot n$, with its arguments divided into $n$ blocks of $m$ variables each. 
We require the function to be symmetric within blocks (permutations of variables inside a block do not change the result),
but only on tuples that satisfy the palette property. 
A tuple is a palette tuple if any element that appears in it must completely fill one of the overlined blocks.
For example,  
$(\overline 3,\overline 3,\overline 3)$-palette symmetric function $f$ must satisfy the identity 
$$f(\underbrace{x,x,x},\underbrace{x,y,x},\underbrace{y,y,y}) = f(\underbrace{x,x,x},\underbrace{y,x,x},\underbrace{y,y,y})$$
since the first and third blocks are overlined and completely filled with $x$ and $y$, respectively. A detailed discussion on how to go from minions to palette symmetric polymorphisms 
is given in 
Section \ref{SUBSECTIONPathFromMinionToPalette}.

Among the possible combinations of algorithms, the singleton version of the BLP+AIP is particularly powerful.
The second main result of the paper is the following theorem.

\begin{thm}\label{THMMainSmallerThanEight}
Suppose $\mathbb A$ is a finite relational structure 
having a WNU polymorphism
such that 
$|\proj_{i}(R^{\mathbb A})|<8$ for every $R$ and  $i\in [\arity(R)]$.
Then $\Singl(\BLP+\AIP)$ solves $\CSP(\mathbb A)$.
\end{thm}

\tabulinesep=0.9mm
\begin{table}[!ht]
\caption{Characterization of the algorithms}
\begin{tabu}{|[2pt]c|p{12cm}|[2pt]}
\tabucline[2pt]{-}
\cellcolor{blue!25}Algorithms $\mathfrak A$&  $\mathfrak A$ solves $\PCSP(\mathbb A,\mathbb B)$ if and only if for all $m,\ell,n\in\mathbb N$ the set $\Pol(\mathbb A,\mathbb B)$ contains \\ \tabucline[2pt]{-}
$\ArcConsistency$&$m$-ary totally symmetric function\\ \hline
$\BLP$&$m$-ary symmetric function\\ \hline
$\AIP$&$(2\ell+1)$-ary alternating function\\ \hline
$\ArcConsistency\wedge\AIP$&$(m,2n+1)$-block function whose first block is totally symmetric and the second block is alternating\\ \hline
$\BLP\wedge\AIP$&$(m,2\ell+1)$-block function whose first block is symmetric and  second block is alternating\\ \hline
$\ArcConsistency+\AIP$&
$(2\ell+1)$-ary weakly alternating function
    \vspace{-0.7cm}\\ \hline
$\BLP+\AIP$&$(m+1,m)$-block symmetric function\\ \tabucline[2pt]{-}
\cellcolor{blue!25}
\begin{tabular}[t]{@{}c@{}}$\Singl\ArcConsistency$ \\ $\CSingl\ArcConsistency$\end{tabular}%
&$(\underbrace{\overline{m},\overline{m},\dots,\overline{m}}_{n})$-palette totally symmetric function\\ \hline
\begin{tabular}[t]{@{}c@{}}$\Singl\BLP$ \\ $\CSingl\BLP$\end{tabular}&$(\underbrace{\overline{m},\overline{m},\dots,\overline{m}}_{n})$-palette symmetric function\\ \hline
$\CSingl\AIP$
&$(\underbrace{\overline{2\ell+1},\overline{2\ell+1},\dots,\overline{2\ell+1}}_{n})$-palette alternating function\\ \hline
$\CSingl(\ArcConsistency\wedge\AIP)$
\vspace{-0.0cm}&$(\underbrace{\overline{m,2\ell+1},\overline{m,2\ell+1},\dots,\overline{m,2\ell+1}}_{2n})$-palette function whose odd blocks are totally symmetric and even blocks are alternating  \\ \hline
\begin{tabular}[t]{@{}c@{}}
$\Singl(\ArcConsistency+\AIP)$ \\ $\CSingl(\ArcConsistency+\AIP)$ \end{tabular}\vspace{-0.0cm}&
$(\underbrace{\overline{2\ell+1},\overline{2\ell+1},\dots,\overline{2\ell+1}}_{n})$-palette weakly alternating function  \\ \hline
\begin{tabular}[t]{@{}c@{}}$\Singl(\BLP+\AIP)$ \\ $\CSingl(\BLP+\AIP)$\end{tabular}&$(\underbrace{\overline{\ell+1,\ell},\overline{\ell+1,\ell},\dots,\overline{\ell+1,\ell}}_{2n})$-palette symmetric function\\ \hline
$\CSingl(\BLP\wedge\AIP)$&$(\underbrace{\overline{m,2\ell+1},\overline{m,2\ell+1},\dots,\overline{m,2\ell+1}}_{2n})$-palette function whose odd blocks are symmetric and even blocks are alternating\\ \hline

\begin{tabular}[t]{@{}c@{}} $\Singl\BLP+\AIP$ {(CLAP)}\\ 
$\CSingl\BLP+\AIP$\end{tabular}&$(\underbrace{\overline{m},\overline{m},\dots,\overline{m}}_{n},2\ell+1)$-palette  function whose first $n$ blocks are symmetric and last block is alternating\\ \hline
\begin{tabular}[t]{@{}c@{}}$\Singl\AIP+\ArcConsistency$ \\ $\CSingl\AIP+\ArcConsistency$\end{tabular}&$(\underbrace{\overline{2\ell+1},\overline{2\ell+1},\dots,\overline{2\ell+1}}_{n},m)$-palette  function whose first $n$ blocks are alternating and last block is totally symmetric\\ \hline
$\CSingl\AIP+\BLP$&$(\underbrace{\overline{2\ell+1},\overline{2\ell+1},\dots,\overline{2\ell+1}}_{n},m)$-palette function whose first $n$ blocks are alternating and last block is symmetric\\ \hline
\begin{tabular}[t]{@{}c@{}}$\Singl\ArcConsistency+\BLP$ \\ $\CSingl\ArcConsistency+\BLP$\end{tabular}&$(\underbrace{\overline{m},\overline{m},\dots,\overline{m}}_{n},\ell)$-palette function whose first $n$ blocks are totally symmetric and last block is symmetric.\\ \hline
\begin{tabular}[t]{@{}c@{}}$\Singl\ArcConsistency+\AIP$ \\ $\CSingl\ArcConsistency+\AIP$\end{tabular}&$(\underbrace{\overline{m},\overline{m},\dots,\overline{m}}_{n},2\ell+1)$-palette function whose first $n$ blocks are totally symmetric and last block is alternating.   \\ \tabucline[2pt]{-}
\end{tabu}
\label{TableCharacterization}
\end{table}

Thus, up to domain size 7 (even in the multi-sorted case), all tractable CSPs can be solved by this very simple algorithm.
However, this is no longer true for domain size 8: as observed by Kompatscher, the structure on 8 elements arising from the dihedral group 
$\mathbf D_4$ 
 does not admit palette symmetric polymorphisms, implying that 
$\Singl(\BLP+\AIP)$ fails to solve it.

As a byproduct, we show that the two versions of the singleton algorithms --- fixing a tuple of a constraint ($\CSingl$) or a value of a variable ($\Singl$) --- have the same power for most of the algorithms.
Surprisingly, in Section
\ref{SECTIONTemporalCSP} we discover two structures 
$\mathbb A$ and $\mathbb B$, where $\mathbb A$ is finite and 
$\mathbb B$ is infinite, 
such that 
$\PCSP(\mathbb A,\mathbb B)$ can be solved by 
$\CSingl\ArcCons$ but not by 
$\Singl\ArcCons$.
This emphasizes that 
the equivalence of $\Singl$ and $\CSingl$ for finite structures 
is nontrivial 
and the use of the Hales-Jewett theorem is essential.

In Section \ref{SECTIONLimitations}
we define several CSP templates 
to illustrate the limitations of linear programming algorithms such as 
BLP, AIP, and $\BLP+\AIP$. 
By providing concrete unsatisfiable instances on which  
$\BLP+\AIP$ incorrectly returns ``Yes'',
we demonstrate that the failure of the algorithm
arises because the corresponding template is, in a certain sense, disconnected.
We then show that the singleton version overcomes this difficulty, as fixing an element allows the algorithm to solve each connected part separately.
Finally, we show that 
the template arising from the dihedral group $\mathbf D_4$ cannot be solved  
even by the singleton version of $\BLP+\AIP$. Thus, we identify another fundamental reason why linear programming algorithms fail.
In this section we not only present instances that fool the algorithms, but also witness the characterization of the singleton versions by providing concrete palette symmetric polymorphisms.

In Section \ref{SECTIONTemporalCSP}, we apply 
our ideas to temporal CSPs, where 
a relational structure 
is called \emph{temporal} if its domain is $\mathbb Q$ and 
its relations are definable by Boolean combinations of atomic formulas of the form $x < y$.
The classification of the complexity of temporal CSPs has been known for a long time \cite{bodirskyTemporalCSP},
but only recently Mottet proposed a new uniform algorithm
based on universal algorithms \cite{mottet2025TemporalReduction}.
Notice that most of the universal algorithms cannot be applied 
directly to infinite domain CSPs.
For instance, in linear relaxations one introduces 
a variable for every element of the domain, which would result in a system with infinitely many variables.

One of the ways to overcome this difficulty
is ``sampling'':
to solve an instance with $N$ variables one first 
restricts the domain to any $N$-element subset of $\mathbb Q$, 
then runs a universal algorithm on the restricted instance, 
and returns its answer. 
Mottet proved in \cite{mottet2025TemporalReduction}
that every tractable temporal CSP can be solved this way
using local consistency and $\AIP$ algorithms.
The soundness of the algorithm 
follows from the fact that $\PCSP(\mathbb A|_{[N]},\mathbb A)$ 
is solvable by the corresponding universal algorithms, 
where $\mathbb A|_{[N]}$ is the induced substructure obtained  
by restricting the domain to $\{1,2,\dots,N\}$.

In the paper we present an alternative proof of this result 
by providing (generalized) palette polymorphisms, which
can be used 
to obtain an actual solution from the outputs 
of the algorithms $\CSingl\ArcCons$ and $\Singl\AIP$.
For one of these structures $\mathbb A$ we observe that 
$\CSingl\ArcCons$ solves its CSP but $\Singl\ArcCons$ does not, 
and this structure is precisely the one for which $\Pol(\mathbb A|_{[N]},\mathbb A)$ does not admit palette symmetric polymorphisms --- which is also why we introduce generalized palette polymorphisms instead.
Thus, we prove the following slight strengthening of 
Mottet's result, where 
local consistency is replaced by $\CSingl\ArcCons$:

\begin{thm}\label{THMMyTemporalClassification}
Let $\mathbb A$ be a temporal structure, $N\in\mathbb N$.
Then one of the following holds:
\begin{itemize}
\item $\PCSP(\mathbb A|_{[N]}, \mathbb A)$ is solvable by
$\CSingl\ArcCons\wedge \Singl\AIP$;
\item  $\CSP(\mathbb A)$ is NP-hard.
\end{itemize}
\end{thm}

\subsection{From minions to palette symmetric polymorphisms}\label{SUBSECTIONPathFromMinionToPalette}

In this section we explain how we go from minions to palette symmetric operations in more detail (for the precise definitions see Section \ref{SECTIONSingletonCharacterization}).
The power of standard  algorithms
such as $\ArcCons$, $\BLP$, and $\AIP$ can be characterized by minions 
in the sense that 
$\ArcConsistency$ solves $\PCSP(\mathbb A, \mathbb B)$
correctly if and only if 
there exists a minion homomorphism $\mathcal M_{\ArcCons}\to \Pol(\mathbb A, \mathbb B)$  \cite{PCSPBible}.
The minions 
$\mathcal M_{\ArcCons}$, 
$\mathcal M_{\BLP}$,
$\mathcal M_{\AIP}$ corresponding to the algorithms 
are very simple, their objects are precisely
non-empty subsets,
rational non-negative vectors whose entries sum to 1, 
and
integer vectors whose entries sum to 1, respectively. 
These 3 minions contain a lot of symmetric objects, and it is not hard to find an object satisfying a lot of identities. 
For example, 
$\mathcal M_{\BLP}$ contains a vector $(1/n,\dots,1/n)$, which immediately means that if $\BLP$ solves $\PCSP(\mathbb A, \mathbb B)$
then $\Pol(\mathbb A, \mathbb B)$ has to contain a symmetric function of any arity $n$. 
Further, we notice that these symmetric functions can be used for rounding, i.e. to get an actual solution from the solution of BLP (see Section \ref{SECTIONUniversalAlgorithms}), which implies that 
symmetric functions completely characterize the power of the $\BLP$
algorithm.

Similarly, following the ideas of 
\cite{ciardo2023clap}, we can define minions characterizing the power of the constraint singleton version ($\CSingl$) of $\ArcConsistency$, $\BLP$, $\AIP$, and even $\BLP$+$\AIP$. 
For example, 
the minion characterizing $\CSingl\BLP$ can be represented as 
a skeletal matrix whose columns are non-negative vectors with sum 1,
where skeletal means that whenever some row contains a nonzero element, 
there must be a column having $1$ in this row. 
Here is an example of a skeletal matrix, which corresponds to 
a $4$-ary object of $\mathcal M_{\CSingl\BLP}$:

$$
\begin{array}{|c|c|c|c|c|c|c|c|c|c|c|c|c|c|c|}
\hline
0.5 &\mathbf{1} &  0.2 & 0.3 & \mathbf{0} & 0.1 & 0.6 & 1 & 0.8 & 0.9 & 0.3 & \mathbf{0} & 0.5 & 0.4 & 0.6 \\ \hline
0.5 & \mathbf{0} &  0.3 & 0.3 & \mathbf{0} & 0.4 & 0.2 & 0 & 0.1 & 0.1 & 0.5 & \mathbf{1} & 0.3 & 0.2 & 0.2 \\ \hline
0 &\mathbf{0} &  0 & 0 & \mathbf{0} & 0 & 0 & 0 & 0 & 0 & 0 & \mathbf{0} & 0 & 0 & 0 \\ \hline
0   &\mathbf{0} &  0.5 & 0.4 & \mathbf{1} & 0.5 & 0.2   & 0 & 0.1 & 0 & 0.2 & \mathbf{0} & 0.2 & 0.4 & 0.2 \\ \hline
\end{array}
$$

The skeletal property is fundamentally asymmetric:
it is hard to write a single identity satisfied by a skeletal matrix. 
That is why, it was hard to use this minion characterization to show 
that a concrete PCSP can be solved by the corresponding singleton algorithm. 
In fact, proving this is not much different from just proving that 
the algorithm solves the PCSP directly.

The key insight is that even though the minion $\mathcal{M}_{\CSingl\BLP}$ 
has no symmetric objects, mapping it into a locally finite minion $\mathcal{N}$ 
forces many objects to be identified, and these identifications give rise to 
palette symmetric objects in $\mathcal{N}$. Note that $\Pol(\mathbb{A},\mathbb{B})$ 
is locally finite whenever $\mathbb{A}$ and $\mathbb{B}$ are finite, since there 
are only finitely many polymorphisms of every arity. To prove the existence of 
palette symmetric objects, we apply the Hales-Jewett theorem as follows. 
Let $\mathcal{A}$ be the set of all matrices of size $n \times L$, and let 
$\mathcal{B}$ be the set of all $n$-ary objects of $\mathcal{N}$ (i.e., the set 
of $n$-ary polymorphisms when $\mathcal{N} = \Pol(\mathbb{A},\mathbb{B})$). 
For matrices $M_1,\dots,M_s \in \mathcal{A}$, define a mapping 
$f\colon \mathcal{A}^s \to \mathcal{B}$ by setting $f(M_1,\dots,M_s)$ 
to be the $n$-ary object corresponding to the concatenation 
$[M_1 \mid M_2 \mid \cdots \mid M_s]$. 
Since both $\mathcal{A}$ and $\mathcal{B}$ are finite, the Hales-Jewett theorem 
guarantees that for sufficiently large $s$, substituting constant matrices for 
some of the inputs $M_1,\dots,M_s$ makes all remaining variable positions a 
dummy coordinate. The skeletal columns can then be hidden into this dummy 
coordinate, producing the desired symmetry.
The precise construction is given in 
Section~\ref{THMHalesJewettApplication}.

The advantage of palette polymorphisms over the minion characterization is 
demonstrated in subsequent sections. 
The power of palette functions is especially apparent in 
Section~\ref{SectionSmallDomains}, where we prove the existence of palette 
symmetric term operations in any small Taylor algebra. The construction proceeds by building palette symmetric terms from palette symmetric terms on retract algebras and subalgebras, and we hope that this construction will help to show that any Taylor algebra has some symmetric terms witnessing that some universal algorithm solves all tractable CSPs.

\subsection{Structure of the paper}

The paper is organized as follows. 
In Section \ref{SECTIONPreliminaries} 
we give necessary definitions. 
Section \ref{SECTIONUniversalAlgorithms} 
introduces classic universal algorithms and their combinations, defines palette symmetric functions, 
and formulates  criteria for these algorithms to 
solve $\PCSP(\mathbb A,\mathbb B)$.
In Section \ref{SECTIONLimitations} we discuss the limitations and weaknesses of the classic universal algorithms such as BLP, AIP, and BLP+AIP, 
and demonstrate why the singleton versions are more powerful.
Additionally, we show that the template coming from 
the dihedral group $\mathbf D_4$ cannot be solved by 
the singleton version of $\BLP+\AIP$.
In Section \ref{SECTIONSingletonCharacterization} 
we prove the characterization of the singleton versions of the algorithms in terms of the existence of palette symmetric polymorphisms.
Section \ref{SECTIONTemporalCSP}
provides concrete palette symmetric polymorphisms for 
temporal relational structures
and shows that 
all tractable temporal CSPs can be solved (after sampling)
by $\Singl\AIP$ or $\CSingl\ArcCons$.
In Section \ref{SECTIONDihedralGroupProofs} 
we prove the relational characterization of the dihedral group $\mathbf D_4$
and show that it does not admit palette symmetric polymorphisms.
Finally, in Section \ref{SectionSmallDomains} we
prove that every finite template on a domain of size at most 7 that admits a WNU polymorphism also admits a palette symmetric polymorphism, and therefore can be solved by the singleton version of $\BLP+\AIP$.

\section{Preliminaries}\label{SECTIONPreliminaries}

For $n\in \mathbb N$ we denote $[n] = \{1,2,\dots,n\}$.
We usually denote tuples 
$\mathbf a \in A^{n}$ in bold and refer to their components 
as 
$(\mathbf a^{1},\dots,\mathbf a^{n})$.
For a permutation $\sigma\colon[n]\to [n]$, 
we denote by $\mathbf a^{\sigma}$ the tuple 
$(\mathbf a^{\sigma(1)},\dots,\mathbf a^{\sigma(n)})$.

\emph{A signature} is a set of relational symbols $R_1,\dots,R_s$, each with
respective arity $\arity(R_{i})\in\mathbb N$. 
\emph{A $\sigma$-structure} $\mathbb A$ consists of a domain $A$ and, 
for each $R_{i}\in \sigma$, a relation $R_{i}^{\mathbb A}\subseteq A^{\arity(R_{i})}$. 

For two $\sigma$-structures $\mathbb A$ and $\mathbb B$, \emph{a homomorphism from $\mathbb A$ to $\mathbb B$} is a mapping $\phi\colon A \to B$ such
that for each symbol $R\in \sigma$ and each tuple 
$\mathbf a = (a_1, \dots, a_{\arity(R)})\in R^{\mathbb A}$,
we have 
$\phi(\mathbf a) := (\phi(a_1), \dots , \phi(a_{\arity(R)})) \in R^{\mathbb B}$. If such a homomorphism exists, we write $\mathbb A \to \mathbb B$;
otherwise, we write $\mathbb A \not\to \mathbb B$.

We say that a pair of finite structures $(\mathbb A,\mathbb B)$ is 
\emph{a PCSP template} if they have the same signature and 
there exists a homomorphism $\mathbb A\to \mathbb B$.

Let $\mathcal O_{A,B}^{(n)}$ denote the set of all mappings
$A^{n}\to B$, 
called \emph{$n$-ary functions},
and let $\mathcal O_{A,B}= \bigcup\limits_{n=1}^{\infty} \mathcal O_{A,B}^{(n)}$.
If $A = B$, 
we simply write 
$\mathcal O_{A}$ and $\mathcal O_{A}^{(n)}$,
and we usually use the word \emph{operation}
instead of function.

To distinguish between equations and identities, 
we use the symbol $\approx$ to emphasize that 
the equality holds for all evaluations of the variables.
For instance, an operation $f\in \mathcal O_{A}$ is called 
\emph{idempotent} if 
$f(x,\dots,x) \approx x$, 
which means that $f(a,\dots,a) = a$ for all $a\in A$.
A function $w$ is called \emph{a weak near unanimity (WNU)} function if it  
satisfies the identity:
$$w(y,x,\dots,x) \approx w(x,y,x,\dots,x) \approx \dots \approx w(x,\dots,x,y).$$

A function
$f\colon A^{n}\to B$ is called \emph{a polymorphism} from a $\sigma$-structure $\mathbb A$
to a $\sigma$-structure $\mathbb B$
if,
for each symbol $R\in \sigma$ and  
tuples $\mathbf a_1,\dots,\mathbf a_{n}\in R^{\mathbb A}$,
we have 
$f(\mathbf a_1,\dots,\mathbf a_n) := 
(f(\mathbf a_1^1,\dots,\mathbf a_n^1),\dots,
f(\mathbf a_1^{\arity(R)},\dots,\mathbf a_n^{\arity(R)}))
 \in R^{\mathbb B}$.
By $\Pol(\mathbb A,\mathbb B)$ we denote the set of all polymorphisms 
from $\mathbb A$ to $\mathbb B$.
When $\mathbb A = \mathbb B$, we write simply 
$\Pol(\mathbb A)$.

A \emph{clone} on a set $A$ is a set of operations $\mathcal{C}\subseteq \mathcal O_{A}$ such that:
\begin{enumerate}
    \item $\mathcal{C}$ contains all projection operations $\pi_i^n(x_1, \dots, x_n) = x_i$ for all $n \geq 1$ and $i\in[n]$;
    \item $\mathcal{C}$ is closed under composition, i.e., if $f \in \mathcal{C}$ is $k$-ary and $g_1, \dots, g_k \in \mathcal{C}$ are $n$-ary, then the $n$-ary operation
    $$h(x_1, \dots, x_n) = f\bigl(g_1(x_1, \dots, x_n), \dots, g_k(x_1, \dots, x_n)\bigr)$$
    also belongs to $\mathcal{C}$.
\end{enumerate}

For a set of operations $F\subseteq \mathcal O_{A}$
by $\Clo(F)$ we denote the minimal clone containing $F$.
For a set $N$, 
we denote by 
$\mathrm{id}_{N}\colon N\to N$
the identity map.
For a relation $R\subseteq A^{n}$ and $S\subseteq [n]$ 
by $\proj_{S}(R)$ we denote the projection of $R$ onto the coordinates indexed by $S$.






\section{Universal algorithms}\label{SECTIONUniversalAlgorithms}

\subsection{Arc-consistency algorithm}\label{SUBSECTIONArcConsistency}

The most basic algorithm for the CSP 
is to enforce some local consistency.
Considering the lowest level of consistency we get the arc-consistency algorithm. 

Suppose $\mathcal I=\langle X , D , \mathbf{C} \rangle$ is an instance of 
$\FiniteCSP$.
The arc-consistency algorithm is implemented by the following pseudocode.
By $\ReduceDomain(\mathcal I,x,B)$ we denote the function that returns a new instance obtained from $\mathcal I$ by reducing the domain of the variable $x$ to $B$ 
and removing the corresponding tuples from the constraints containing $x$.

\begin{algorithm}
\begin{algorithmic}[1]
\Function{\ArcConsistency}{$\mathcal I$}
    \Repeat     
         \For{each constraint $C= R(x_{1},\dots,x_{s})\in\mathcal I$}
            \For{$i\in\{1,2,\dots,s\}$}
        
        \If {$\proj_{i}(R) \neq D_{x_i}$}
            \State{$B:= \proj_{i}(R)$}
                    \If {$B=\varnothing$}
                        \State \Return{``No"}
                    \EndIf
                    \State{$\mathcal I:= \ReduceDomain(\mathcal I, x_i,B)$}
        \EndIf
            \EndFor 
        \EndFor
    \Until{nothing changed}    
\State \Return{``Yes"}
\EndFunction
\end{algorithmic}
\end{algorithm}

\begin{obs}\label{OBSArcCons}
The algorithm $\ArcCons$
returns ``Yes'' on a CSP instance $\mathcal I$ if and only if 
there exists 
a mapping $D'$ that assigns a nonempty set $D_{x}'\subseteq D_{x}$ to every variable $x$ such that 
for any constraint $R(x_{1},\dots,x_{s})$ of $\mathcal I$
and any $j\in[s]$ we have $\proj_{j}(R\cap (D_{x_1}'\times\dots\times D_{x_s}')) = D_{x_j}'$.
\end{obs}

The power of the arc-consistency algorithm has been 
completely characterized for both 
$\CSP(\mathbb A)$ and $\PCSP(\mathbb A,\mathbb B)$.
A function $f\colon A^{n}\to B$ is called 
\emph{totally symmetric} if 
$f(a_1,\dots,a_n) = f(b_1,\dots,b_n)$
whenever $\{a_1,\dots,a_n\} =\{b_1,\dots,b_n\}$.

\begin{thm}[Theorem 7.4 in \cite{PCSPBible}]
Let $(\mathbb A, \mathbb B)$ be a PCSP template. Then $\ArcCons$ solves $\PCSP(\mathbb A,\mathbb B)$ if and only if 
$\Pol(\mathbb A,\mathbb B)$ contains totally symmetric functions of all arities. 
\end{thm}

One direction of this theorem can be proved as follows.
Suppose $\Pol(\mathbb A,\mathbb B)$ contains totally symmetric polymorphisms of all arities.
First, we choose a sequence of totally symmetric polymorphisms 
$f_{1},f_2,f_{n},\dots,$ such that each $f_{n}$ is of arity $n$, 
and identifying two variables in $f_{n}$ yields $f_{n-1}$.
Then, for a set $S = \{a_1,\dots,a_s\}\subseteq A$,
we define $f(S):=f_{s}(a_1,\dots,a_s)$.
If the algorithm $\ArcCons$
returns ``Yes'' on a CSP instance $\mathcal I$, then it assigns 
a nonempty set $D_{x}'$ to every variable $x$.
We claim that a solution to $\mathcal I^{\mathbb B}$ can be defined by
$s(x):= f(D_{x}')$
for every variable $x$. 
Let $R(x_{1},\dots,x_{k})$ be a constraint of $\mathcal I$.
Consider $R^{\mathbb A}\cap (D_{x_1}'\times\dots\times D_{x_k}') = \{\mathbf a_1,\dots,\mathbf a_m\}$.
Since $f_m$ is a polymorphism,
we have $f_m(\mathbf a_1,\dots,\mathbf a_m)\in R^{\mathbb B}$.
By the definition of $f$, we obtain 
$f_{m}(\mathbf a_1,\dots,\mathbf a_m) = (s(x_{1}),\dots,s(x_{k}))$.
Therefore, $(s(x_{1}),\dots,s(x_{k}))\in R^{\mathbb B}$. 
Thus, $s$ is a solution of $\mathcal I$ in $\mathbb B$,
and $\ArcCons$ solves the instance correctly.

\subsection{Linear relaxations}


The next algorithms 
arise from the idea of applying 
algorithms for solving linear equations 
(such as Gaussian elimination) to arbitrary CSPs over finite domains. 
Suppose $\mathcal I$ is a CSP instance given by the triple
$\langle X , D , \mathbf{C} \rangle$.
By $\mathcal I^{L}$ we denote the following system of linear equations (here + is a formal symbol; an interpretation is specified separately for BLP and AIP):

\begin{tabular}{ll}
    \textbf{Variables:} & $C^{\mathbf a}$ for every constraint $C = R(x_1,\dots,x_m)$ in $\mathcal I$ and every $\mathbf a\in R$,\\
    & $x^{b}$ for every variable $x\in X$ and $b\in D_{x}$;\\
    \textbf{Equations:}&  $\sum\limits_{\mathbf a\in R} C^{\mathbf a} = 1$ for every constraint $C = R(x_1,\dots,x_m)$;\\
        &$\sum\limits_{\substack{\mathbf a\in R\\ \mathbf a(i) = b}} C^{\mathbf a} = x_{i}^{b}$ for 
        every constraint $C = R(x_1,\dots,x_m)$, every $i\in[m]$ and $b\in D_{x_i}$.
\end{tabular}

    



If we restrict all the variables to the set $\{0,1\}$ and interpret 
addition in $\mathbb Z$, then we obtain 
a one-to-one correspondence between solutions of $\mathcal I$ and 
solutions of 
$\mathcal I^{L}$ where $x = b\Leftrightarrow x^{b}=1$.

\textbf{BLP Algorithm.} The algorithm takes a CSP instance 
$\mathcal I$, transforms it into a system of linear equations $\mathcal I^{L}$, and solves it over $\mathbb Q\cap [0,1]$
using linear programming \cite{BLPRef}. It returns ``Yes'' if a solution exists and ``No'' otherwise.

The power of the BLP algorithm has been
completely characterized for both 
$\CSP(\mathbb A)$ and $\PCSP(\mathbb A,\mathbb B)$.
A function $f\colon A^{n}\to B$ is called 
\emph{symmetric} if 
$f(x_1,\dots,x_n) \approx f(x_{\sigma(1)},\dots,x_{\sigma(n)})$
for any permutation $\sigma\colon [n]\to[n]$.

\begin{thm}[Theorem 7.9 in \cite{PCSPBible}]
Let $(\mathbb A, \mathbb B)$ be a PCSP template. Then
$\BLP$ solves $\PCSP(\mathbb A,\mathbb B)$ if and only if 
$\Pol(\mathbb A,\mathbb B)$ contains symmetric functions of all arities. 
\end{thm}

\textbf{AIP Algorithm.} The algorithm takes a CSP instance 
$\mathcal I$, transforms it into a system of linear equations $\mathcal I^{L}$, and solves it over $\mathbb Z$
using linear programming \cite{papa}. It returns ``Yes'' if a solution exists and ``No'' otherwise.

The power of the AIP algorithm has been 
completely characterized for both 
$\CSP(\mathbb A)$ and $\PCSP(\mathbb A,\mathbb B)$.
A function $f\colon A^{n}\to B$ is called 
\emph{alternating} if 
$f(x_1,\dots,x_n) \approx f(x_{\sigma(1)},\dots,x_{\sigma(n)})$
for any permutation $\sigma\colon [n]\to[n]$
preserving parity, 
and $f(x_1,\dots,x_{n-2},y,y) \approx f(x_1,\dots,x_{n-2},z,z)$.
Sometimes we will need a weaker version of it where
the last identity is required only when the values $y$ and $z$ already appear among other variables.
Formally, a function $f$ is called \emph{weakly alternating} if 
$f(x_1,\dots,x_n) \approx f(x_{\sigma(1)},\dots,x_{\sigma(n)})$
for every parity-preserving permutation $\sigma\colon [n]\to[n]$, 
and 
$f(x_1,\dots,x_{n-2},x_i,x_i) \approx f(x_1,\dots,x_{n-2},x_j,x_j)$
for any $i,j\in[n-2]$.

\begin{thm}[Theorem 7.19 in \cite{PCSPBible}]
Let $(\mathbb A, \mathbb B)$ be a PCSP template. Then
$\AIP$ solves $\PCSP(\mathbb A,\mathbb B)$ if and only if 
$\Pol(\mathbb A,\mathbb B)$ contains 
alternating functions of all odd arities.
\end{thm}



A solution of $\mathcal I^{L}$ in $\mathbb Q\cap [0,1]$
    is called \emph{a BLP solution}, and 
a solution of $\mathcal I^{L}$ in $\mathbb Z$ 
    is called \emph{an AIP solution}.






\subsection{Combinations of the algorithms}\label{SubsectionCombinationOfAlgorithms}

We will use the following notations for combinations of algorithms.
For universal algorithms $\mathfrak A$ and $\mathfrak B$ by
$\mathfrak A\wedge \mathfrak B$ we denote the algorithm that 
runs both algorithms independently and 
returns ``Yes'' if both return ``Yes''.

By BLP+$\mathfrak A$ we denote the algorithm that 
first finds a solution of $\mathcal I^{L}$ in $\mathbb Q\cap[0,1]$ with the maximal number of nonzero values, and returns
``No'' if no such solution exists.
Then,  it removes  
all tuples $\mathbf a$ from every constraint $C$ in $\mathcal I$ such that $C^{\mathbf a}=0$ in the solution.
Finally, the algorithm runs $\mathfrak A$ and returns its answer.

Similarly, by $\ArcCons$+$\mathfrak A$ we denote the algorithm that
first runs the arc-consistency algorithm, returns ``No'' if the latter returns ``No'', and otherwise applies 
$\mathfrak A$ to the reduced instance and returns its answer.

The strength of the BLP+AIP algorithm has the following nice characterization for both 
$\CSP(\mathbb A)$ and $\PCSP(\mathbb A,\mathbb B)$ in terms of symmetric polymorphisms.
A function $f:A^{k_1+\dots+k_n}\to B$ is called 
\emph{$(k_1,\dots,k_n)$-block symmetric}
if for any permutations $\sigma_1,\dots, \sigma_{n}$ on 
$[k_1],\dots,[k_n]$ respectively, 
we have $$f(\mathbf x_1,\dots,\mathbf x_n) \approx f(\mathbf x_{1}^{\sigma_1},\dots,
\mathbf x_{n}^{\sigma_n}).$$

\begin{thm}[Theorem 5.1 in \cite{BLPplusAIP}]\label{thmBLPAIPCharacterization}
Let $(\mathbb A, \mathbb B)$ be a PCSP template. Then
the following conditions are equivalent:
\begin{enumerate}
    \item[(1)] $\BLP+\AIP$ solves $\PCSP(\mathbb A,\mathbb B)$;
    \item[(2)] for any $L\in\mathbb N$ there exist $s\in\mathbb N$ and numbers
    $\ell_1,\dots,\ell_s$ greater than $L$ such that $\Pol(\mathbb A,\mathbb B)$ contains an
    $(\ell_{1},\dots,\ell_{s})$-block symmetric function.
    \item[(3)] for any $L\in\mathbb N$ the set $\Pol(\mathbb A,\mathbb B)$ contains an
    $(L+1,L)$-block symmetric function.
\end{enumerate}
\end{thm}

To explain the idea behind this theorem let us 
show the implication (2)$\Rightarrow$(1).
Let $\mathcal I^{L,d}$ be obtained from 
$\mathcal I^{L}$ by replacing $1$ in the equations by $d\in\mathbb N$.
A non-negative integer solution of 
$\mathcal I^{L,d}$ is called \emph{a $d$-solution of $\mathcal I$}.
First, observe the following facts:
\begin{itemize}
    \item[(f1)] $\mathcal I$ has a BLP solution $\Leftrightarrow$ there exists 
    $Q\in\mathbb N$ such that $\mathcal I$ has a $Q$-solution.
    \item[(f2)] The sum of a $d_1$-solution and a $d_2$-solution is a 
    $(d_1+d_2)$-solution.
    \item[(f3)] $\BLP+\AIP$ returns ``Yes'' on an instance $\mathcal I$ $\Leftrightarrow$ 
    there exists $L\in\mathbb N$ such that 
    $\mathcal I$ has a $d$-solution for any $d\ge L$.
\end{itemize}

(f1) follows from the fact that multiplying a BLP solution by 
the least common multiple $Q$ of the denominators gives a $Q$-solution, and dividing a $d$-solution by $d$ gives a BLP solution.
(f2) is immediate. Let us show (f3). 
If $\BLP+\AIP$ returns ``Yes'' on $\mathcal I$ then 
there exists a BLP solution $s_{BLP}$ of $\mathcal I^{L}$, 
an AIP solution $s_{AIP}$ of $\mathcal I^{L}$ 
such that 
$s_{BLP}(v)>0$ whenever $s_{AIP}(v)\neq 0$. 
Choose a sufficiently large integer $Q$ 
such that 
$Q \cdot s_{BLP}+s_{AIP}$ is a $(Q+1)$-solution (all values are nonnegative integers).
Then $Q \cdot s_{BLP}$ is a $Q$-solution.
By taking linear combinations
$m\cdot (Q\cdot s_{BLP}+s_{AIP}) + \ell ( Q \cdot s_{BLP})$,
we obtain a $d$-solution for any $d\ge Q\cdot (Q+1)$.
To see the opposite direction of (f3) choose 
a $d_0$-solution, with $d_0\ge L$, having the maximal number of nonzero values
among all the $d$-solutions (for all $d\in\mathbb N$). Such a solution exists by (f2).
Considering the difference between a $(d_0+1)$-solution and 
the $d_{0}$-solution we obtain an AIP solution, while dividing the $d_{0}$-solution by $d_{0}$ gives a BLP solution. Together, these confirm that 
$\BLP+\AIP$ returns ``Yes'' on $\mathcal I$.

Now let us show (2)$\Rightarrow$(1) in Theorem \ref{thmBLPAIPCharacterization}.
Assume that $\BLP+\AIP$ returns ``Yes'' on an instance $\mathcal I$.
Then there exists $L$ coming from (f3). 
Consider  an
    $(\ell_{1},\dots,\ell_{s})$-block symmetric function $f\in\Pol(\mathbb A,\mathbb B)$
    with $\ell_j\ge L$ for all $j\in[s]$.
For every $j\in[s]$ choose an $\ell_{j}$-solution $s_{j}$ of $\mathcal I$.
For every variable $x$ of $\mathcal I$, 
let 
$\mathbf a_{x,j}$ be a tuple of length $\ell_{j}$ such that each element 
$c$ appears in it exactly  
$s_{j}(x^{c})$ times.
We claim that 
a (true) solution of $\mathcal I$ 
can be defined by 
$x := f(\mathbf a_{x,1},\dots,\mathbf a_{x,s})$ for every variable $x$.
To show that a given constraint, say $C = R^{\mathbb A}(x_1,x_2)$, is satisfied, 
let $M_{j}$ be a $(2\times \ell_{j})$-matrix
such that every column $(c_1,c_2)\in R^{\mathbb A}$ appears in it 
exactly $s_{j}(C^{(c_1,c_2)})$ times. 
Notice that the first row of $M_{j}$ can be obtained from 
$\mathbf a_{x_1,j}$ by a permutation of coordinates,
and the second row can be obtained from 
$\mathbf a_{x_2,j}$ by a permutation of coordinates.
Then 
$$(x_1,x_2):=(f(\mathbf a_{x_1,1},\dots,\mathbf a_{x_1,s}),f(\mathbf a_{x_2,1},\dots,\mathbf a_{x_2,s})) = f(M_1,\dots,M_s) \in R^{\mathbb A},$$
where the equality follows from the $(\ell_{1},\dots,\ell_{s})$-block symmetry of $f$, and 
membership follows from $f\in\Pol(\mathbb A,\mathbb B)$.


Similarly, other combinations of the 
above algorithms (see the top part of Table \ref{TableCharacterization}) can be characterized, but since the proofs do not differ substantially from those in \cite{PCSPBible,BLPplusAIP}, we omit them here.

    



\subsection{Singleton algorithms and palette  functions}\label{SUBSECTIONPaletteOperationDefinition}

For a universal algorithm $\Algorithm$, by  $\SinglAlgorithm$
and $\CSingl\Alg$
we denote the algorithms 
defined by the pseudocode below. 
Here, $\ChangeConstraint(\mathcal I, C,R)$ denotes the function that returns a new instance obtained from $\mathcal I$ by changing the constraint relation of the constraint $C$ to $R$. 
The difference between $\SinglAlg$ and $\CSinglAlg$ is that 
in $\SinglAlgorithm$ we fix a variable to an element, whereas in 
$\CSingl\Alg$ we fix a constraint to a tuple.

\begin{remark}\label{REMAEKCSinglvsSingl}
It is not generally true that $\SinglAlgorithm$
is weaker than $\CSingl\Alg$, since 
the algorithm $\SinglAlgorithm$, when fixing a variable $x$ to a value, may reduce all constraints containing $x$,  while  $\CSingl\Alg$ 
always reduces only one constraint. 
Nevertheless, if $\Alg\in \{\ArcCons, \BLP,\ArcCons+\mathfrak A,\BLP+\mathfrak A\}$, then $\SinglAlgorithm$ is clearly weaker 
than $\CSingl\Alg$.
\end{remark}

\begin{question}
Does there exist 
a PCSP template $(\mathbb A, \mathbb B)$ such that 
$\Singl\AIP$ solves $\PCSP(\mathbb A, \mathbb B)$ 
but 
$\CSingl\AIP$ does not?
\end{question}

\begin{algorithm}
\begin{algorithmic}[1]
\Function{\SinglAlg}{$\mathcal I$}
    \Repeat     
    
        \For{each variable $x$ of $\mathcal I$}
            \For{$a\in D_{x}$}
                \If {$\neg\Alg(\ReduceDomain(\mathcal I,x,\{a\}$))}
                    \State{$B:= D_{x}\setminus \{a\}$}
                    \If {$B=\varnothing$}
                        \State \Return{``No"}
                    \EndIf
                    \State{$\mathcal I:= \ReduceDomain(\mathcal I, x,B)$}
                    \Comment{remove all the tuples}
                   
                    \Comment{such that $x\notin B$}
                \EndIf
            \EndFor 
        \EndFor
    \Until{nothing changed}    
\State \Return{``Yes"}
\EndFunction
\end{algorithmic}
\end{algorithm}

 \begin{algorithm}
 \begin{algorithmic}[1]
 \Function{\CSinglAlg}{$\mathcal I$}
     \Repeat     
         \For{each constraint $C= R(x_{1},\dots,x_{s})\in\mathcal I$}
             \For{$(a_1,\dots,a_s)\in R$}
                 \If {$\neg\Alg(\ChangeConstraint(\mathcal I, C,\{(a_1,\dots,a_s)\})$}
                    
\Comment{reduce $R$ to the one-tuple relation}  
                 \State{$R':=R\setminus \{(a_1,\dots,a_s)\}$}
                   \If {$R'=\varnothing$}
                         \State \Return{``No"}
                     \EndIf
		\State{$\mathcal I:=\ChangeConstraint(\mathcal I, C,R')$}
                    \Comment{remove the tuple}

                    \Comment{from the constraint relation of $C$}

                 \EndIf
             \EndFor 
         \EndFor
     \Until{nothing changed}    

 \State \Return{``Yes"}
 \EndFunction
 \end{algorithmic}
 \end{algorithm}

For universal algorithms $\mathfrak A$ and $\mathfrak B$ we denote by
$\Singl\mathfrak A+\mathfrak B$ the algorithm that 
first runs $\Singl\mathfrak A$ and returns ``No'' if it returns ``No''; 
otherwise, it returns the answer of $\mathfrak B$ on the reduced 
instance $\mathcal I$ (see the pseudocode of $\SinglAlg$).
Similarly, by
$\CSingl\mathfrak A+\mathfrak B$ we denote the algorithm that 
first runs $\CSingl\mathfrak A$ and returns ``No'' if it returns ``No''; 
otherwise, it returns the answer of $\mathfrak B$ on the reduced 
instance $\mathcal I$ (see the pseudocode of $\CSinglAlg$).

To characterize the power of singleton algorithms 
we need a generalization of symmetric functions.
We call a function 
$f\colon A^{k_1+\dots+k_n}\to B$  
a \emph{$(k_1,\dots,k_n)$-block function} to emphasize that its arguments are divided into blocks: 
the first $k_1$ arguments form the first block, 
the next $k_2$ arguments form the second block, and so on.
We say that \emph{the $i$-th block of 
a $(k_1,\dots,k_n)$-block function is (totally) symmetric/(weakly) alternating} if 
for any $\mathbf a_1\in A^{k_1},
\dots, 
\mathbf a_{i-1}\in A^{k_{i-1}},
\mathbf a_{i+1}\in A^{k_{i+1}},
\dots,
\mathbf a_{n}\in A^{k_n}$
the function 
$h(\mathbf x) =
f(\mathbf a_1,\dots,\mathbf a_{i-1},\mathbf x,
\mathbf a_{i+1},\dots,\mathbf a_{n})$ is 
(totally) symmetric/(weakly) alternating, respectively.
A function is \emph{$(k_1,\dots,k_n)$-block (totally) symmetric} if all of its blocks are (totally) symmetric,
and 
\emph{$(k_1,\dots,k_n)$-block (weakly) alternating} if
all of its blocks are (weakly) alternating.

We can divide a tuple $\mathbf a\in A^{k_1+\dots+k_n}$
into blocks and refer to each block by its position.
Let $\mathcal P$ be a 
set of disjoint subsets of $[n]$. 
A tuple $\mathbf a = (\mathbf a_1,\dots,\mathbf a_n)\in A^{k_1+\dots+k_n}$ is called 
\emph{$\mathcal P$-$(k_1,\dots,k_n)$-palette} if 
for every element $c\in A$ appearing in $\mathbf a$
there exists $L\in\mathcal P$ such that  
$\mathbf a_{\ell} = (c,c,\dots,c)$ for every $\ell\in L$.
We usually use an overline to specify the set $\mathcal P$.
For example, 
$(\ell, \overline{k_1},\dots,\overline{k_n})$ corresponds to 
$\mathcal P = \{\{2\},\{3\},\dots,\{n+1\}\}$ and  
$(\overline{k_1,m_1},\dots,\overline{k_n,m_n})$ corresponds to
$\mathcal P = \{\{1,2\},\{3,4\},\dots,\{2n-1,2n\}\}$.

Similarly to block functions, we 
define palette functions, where we require the block properties only on palette tuples.
Formally, a function $f\colon A^{k_1+\dots+k_n}\to B$ is called 
\emph{$\mathcal P$-$(k_1,\dots,k_n)$-palette function},
and the $i$-th block of it is \emph{(totally) symmetric/(weakly) alternating} if (1) and (2) imply (3):
\begin{itemize}
    \item[(1)] $\mathbf a$ and $\mathbf b$ are 
$\mathcal P$-$(k_1,\dots,k_n)$-palette tuples;
\item[(2)] $g(\mathbf a)  = g(\mathbf b)$ 
for any 
$(k_1,\dots,k_n)$-block function $g\colon A^{k_1+\dots+k_n}\to B$ 
whose $i$-th block is (totally) symmetric/(weakly) alternating;
\item[(3)] $f(\mathbf a) = f(\mathbf b)$.
\end{itemize}

A function is 
\emph{$\mathcal P$-$(k_1,\dots,k_n)$-palette (totally) symmetric/(weakly) alternating} if 
it is a $\mathcal P$-$(k_1,\dots,k_n)$-palette function 
whose blocks are (totally) symmetric/(weakly) alternating.

The power of singleton algorithms and their combinations 
can be characterized using palette functions. 
Below we present a characterization for 
the algorithms $\Singl\ArcCons$ and 
$\Singl(\BLP+\AIP)$, which are proved in Section \ref{SubsectionSARC}
and \ref{SUBSECTIONSinglBLPAIPCharacterization}, respectively.

\begin{thm}\label{THMSinglArcCharacterization}
Let $(\mathbb A, \mathbb B)$ be a PCSP template. The following conditions are equivalent:
\begin{enumerate}
\item[(1)] $\SinglArcCons$ solves $\PCSP(\mathbb A,\mathbb B)$;
\item[(2)] $\CSingl\ArcCons$ solves $\PCSP(\mathbb A,\mathbb B)$;
\item[(3)] $\Pol(\mathbb A,\mathbb B)$ has an $(\overline{m_1},\dots,\overline{m_n})$-palette totally symmetric 
function for every $n, m_1,\dots,m_n\in\mathbb N$;
\item[(4)] for every $N\in\mathbb N$ there exist 
$n,m_1,\dots,m_n\ge N$ such that 
$\Pol(\mathbb A,\mathbb B)$ has an $(\overline{m_1},\dots,\overline{m_n})$-palette totally symmetric 
function.
\end{enumerate}
\end{thm}

\begin{thm}\label{THMSinglBLPAIP}
Let $(\mathbb A, \mathbb B)$ be a PCSP template.
The following conditions are equivalent:
\begin{enumerate}
\item[(1)] $\Singl(\BLP+\AIP)$ solves $\PCSP(\mathbb A,\mathbb B)$;
\item[(2)] $\CSingl(\BLP+\AIP)$ solves $\PCSP(\mathbb A,\mathbb B)$;
\item[(3)] $\Pol(\mathbb A,\mathbb B)$ contains an $(\underbrace{\overline{\ell+1,\ell},\overline{\ell+1,\ell},\dots,\overline{\ell+1,\ell}}_{2n})$-palette symmetric function for every $\ell,n\in\mathbb N$;
\item[(4)] for every $N\in\mathbb N$ there exist 
$n\ge N$, 
$k_1,\dots,k_n\in\mathbb N$, and $m_{1,1},\dots,m_{1,k_1},
\dots,
m_{n,1},\dots,m_{n,k_n}\ge N$ such that $\Pol(\mathbb A,\mathbb B)$ contains an 
$(\overline{m_{1,1},\dots,m_{1,k_1}},
\overline{m_{2,1},\dots,m_{2,k_2}},\dots,\overline{m_{n,1},\dots,m_{n,k_n}})$-palette symmetric function.
\end{enumerate}
\end{thm}

To explain the idea behind palette symmetric functions and Theorem \ref{THMSinglBLPAIP}, let us prove the following claim.

\begin{claim}
Suppose $(\mathbb A,\mathbb B)$ is a PCSP template, 
and for every $L\in\mathbb N$ there exist $n,\ell_1,\dots,\ell_n\ge L$ such that 
$\Pol(\mathbb A,\mathbb B)$ contains 
an $(\overline{\ell_1},\dots,\overline{\ell_n})$-palette symmetric function.
Then $\Singl(\BLP+\AIP)$ solves $\PCSP(\mathbb A,\mathbb B)$.
\end{claim}

\begin{proof}
It is sufficient to prove that if 
$\Singl(\BLP+\AIP)$ returns ``Yes'' on an instance $\mathcal I$, then 
$\mathcal I$ has a solution.
For every variable $x$ and every element $a$ from the obtained domain of $x$ let $\mathcal I_{x}^{a}$ be the instance obtained in the algorithm 
$\SinglAlg$ by reducing the domain of $x$ to $\{a\}$.
Let $\Omega$ be the set of all such instances.
Since $\BLP+\AIP$ returns ``Yes'' on each of these instances, 
the fact (f3) from Section \ref{SubsectionCombinationOfAlgorithms}
implies that there exists a large number $L$ such that
$L>|\Omega|$ and 
each instance $\mathcal I_{x}^{a}$ has a $d$-solution for any $d\ge L$.
Choose an $(\overline{\ell_1},\dots,\overline{\ell_n})$-palette symmetric function 
$f\in\Pol(\mathbb A,\mathbb B)$, where $n,\ell_1,\dots,\ell_n\ge L$.

Similar to the construction in Section \ref{SubsectionCombinationOfAlgorithms},
for each block (say the $j$-th block) we choose 
an instance $\mathcal I_{x}^{a}$ and an $\ell_{j}$-solution $s_{j}$ of this instance.
We only require every instance from $\Omega$ to be chosen for at least one block, 
which is possible since $n\ge L>|\Omega|$.
For every variable $x$ of $\mathcal I$, 
let 
$\mathbf a_{x,j}$ be a tuple of length $\ell_{j}$ such that each element 
$c$ appears in this tuple 
exactly $s_{j}(x^{c})$ times.
We claim that 
a (true) solution of $\mathcal I$ 
can be defined by 
$x := f(\mathbf a_{x,1},\dots,\mathbf a_{x,s})$ for every variable $x$.
To follow the argument from Section \ref{SubsectionCombinationOfAlgorithms},
it is sufficient to check that the tuple 
$(\mathbf a_{x,1},\dots,\mathbf a_{x,s})$
is $(\overline{\ell_1},\dots,\overline{\ell_n})$-palette, which follows immediately from 
the definition of the instances in $\Omega$.
\end{proof}

As mentioned in the introduction, 
the algorithm $\Singl(\BLP+\AIP)$ is very powerful and solves 
all tractable $\CSP(\mathbb A)$ on small domains, as follows from the following 
algebraic result proved in Section \ref{SectionSmallDomains}.

\begin{thm}\label{THMExistenceWBSForSmallAlgebras}
Suppose $\mathbf A_{1},\dots,\mathbf A_{s}$ are similar algebras
admitting WNU term operations, 
$|A_{i}|<8$ for every $i\in[s]$. 
Then 
$\mathbf A_{1}\times\dots\times \mathbf A_{s}$ has 
a $(\underbrace{p^{m},\dots,p^{m}}_{n})$-palette  symmetric term operation 
for every prime $p>8$ and every $m,n\in\mathbb N$.
\end{thm}%
All necessary definitions are given in Section \ref{SectionSmallDomains}.
An important corollary for our purposes is the following theorem.

\begin{THMMainSmallerThanEightTHM}
Suppose $\mathbb A$ is a finite relational structure 
having a WNU polymorphism
such that 
$|\proj_{i}(R^{\mathbb A})|<8$ for every $R$ and  $i\in [\arity(R)]$.
Then $\Singl(\BLP+\AIP)$ solves $\CSP(\mathbb A)$.
\end{THMMainSmallerThanEightTHM}

An alternative formulation 
is to introduce the notion of a multi-sorted constraint language and restrict the domains by 7 instead of restricting the projections,
but we chose the current formulation to avoid additional notions that are not essential in this paper.

Most other combinations of the algorithms can be characterized similarly (see Table \ref{TableCharacterization}). 
Since their proofs are almost identical to those for 
$\Singl\ArcCons$ and $\Singl(\BLP+\AIP)$, we 
omit them.
We can also characterize 
the algorithm $\Singl\mathfrak A\wedge \mathfrak B$ 
for $\mathfrak A,\mathfrak B\in\{\ArcCons,\BLP,\AIP\}$, but 
this requires a different notion and is less interesting, since the algorithms are simply run independently.
For example, 
$\Singl\BLP\wedge\AIP$ solves 
$\PCSP(\mathbb A,\mathbb B)$ if and only if 
for all $m,\ell,n\in\mathbb N$, $\Pol(\mathbb A,\mathbb B)$ contains
an $(\underbrace{\overline{m},\overline{m},\dots,\overline{m}}_{n},\ell)$-block function
such that the first $n$ blocks are symmetric on $(\underbrace{\overline{m},\overline{m},\dots,\overline{m}}_{n},\ell)$-palette tuples
and the last block is always alternating.


\section{Limitations of linear programming}\label{SECTIONLimitations}

In this section, we discuss the limitations of linear programming and provide concrete templates that illustrate its weaknesses. The results are collected in Table \ref{TableTemplateExamples}, 
while templates are defined below.

\begin{table}[H]
\centering
\caption{Examples of constraint languages}
\begin{tabu}{|[2pt]c|c|c|[2pt]}
\tabucline[2pt]{-}
Constraint language & Strongest not working & Weakest working\\
\tabucline[2pt]{-}
Horn-SAT & $\AIP$& $\ArcCons$ \\
\hline
2-SAT & $\BLP$, $\Singl\AIP$ & $\Singl\ArcCons$, $\BLP+\AIP$\\
\hline
$\LIN{3}{p}$ & $\Singl\BLP$ & $\AIP$\\
\hline
$\left(\begin{smallmatrix}0&1&2&3&4\\1&0&3&4&2\end{smallmatrix}\right)$ & $\BLP+\AIP$ & $\Singl\ArcCons$, $\Singl\AIP$\\
\hline
$\mathbb Z_{2}\cup\mathbb Z_{3}$ & $\Singl\BLP+\AIP$ & $\Singl\AIP$\\
\hline
$\mathbb Z_{2}\leftarrow 2$ & $\Singl\BLP+\AIP$ & $\Singl\AIP$\\
\hline
$\mathbb D_4^{idemp}$ & $\Singl(\BLP+\AIP)$ & ???\footnotemark \\
\tabucline[2pt]{-}
\end{tabu}
\label{TableTemplateExamples}
\end{table}
\footnotetext{We already know that $\CSP(\mathbb D_{4}^{idemp})$ can be solved by the 
second level of $\AIP$; however higher-level algorithms are outside the 
scope of this paper.}


The first simple idea demonstrating the weakness of the $\BLP$ algorithm can be formulated as follows.
\begin{claim}\label{CLAIMABLPUniformSolution}
Suppose 
$\mathbb A$ is a finite relational structure,
for
every relation $R^{\mathbb A}$ of an arity $m$, every $b,c\in A$, and every $i\in[m]$, 
we have $|\{\mathbf a\in R^{\mathbb A}\mid \mathbf a^{i} = b\}| =
|\{\mathbf a\in R^{\mathbb A}\mid \mathbf a^{i} = c\}|$.
Then $\BLP$ answers ``Yes'' on any instance of $\CSP(\mathbb A)$.
\end{claim}
\begin{proof}
The uniform distribution is always a solution to 
$\mathcal I^{L}$. 
Specifically, it is sufficient to 
set
$C^{\mathbf a} = 1/|R^{\mathbb A}|$
for every constraint $C = R(x_{1},\dots,x_{m})$
and $\mathbf a\in R$, and 
$x^{b} = 1/|A|$ for every variable $x$ and $b\in A$.
\end{proof}

Thus, $\BLP$ returns ``Yes'' whenever the uniform distribution is feasible, making it useless for solving linear equations over finite fields.

The next claim shows the main weakness of $\AIP$: it produces the same result even if we replace each constraint relation by its parallelogram closure. 
We say that a relation $R\subseteq A^{m}$
satisfies \emph{the parallelogram property} if 
for any $I\subseteq[m]$, $\mathbf a,\mathbf b,\mathbf c\in R$, 
and $\mathbf d\in A^{m}$, we have
$$\left(\proj_{I}(\mathbf a) = \proj_{I}(\mathbf b)\wedge
\proj_{I}(\mathbf c) = \proj_{I}(\mathbf d)\wedge
\proj_{[m]\setminus I}(\mathbf a) = \proj_{[m]\setminus I}(\mathbf c)\wedge
\proj_{[m]\setminus I}(\mathbf b) = \proj_{[m]\setminus I}(\mathbf d)\right)
\Longrightarrow \mathbf d\in R.$$
For binary relations $R$ this simply means that 
$(b_1,b_2),(b_1,a_2),(a_1,b_2)\in R$ implies $(a_1,a_2)\in R$. 
The \emph{parallelogram closure of $R$} is the minimal relation 
$R'\supseteq R$ having the parallelogram property.

\begin{claim}\label{CLAIMAIPParallelogram}
Suppose 
$\mathcal I$ is an instance of $\FiniteCSP$
containing a constraint $C = R(x_{1},\dots,x_{m})$, and 
\begin{itemize}
    \item $(a_1,\dots,a_{k},b_{k+1},\dots,b_{m}),(b_1,\dots,b_{k},a_{k+1},\dots,a_{m}),
(b_1,\dots,b_{k},b_{k+1},\dots,b_{m})\in R$, 
\item $R' = R \cup\{(a_1,\dots,a_{k},a_{k+1},\dots,a_m)\}$, 
\item $\mathcal I'$ is obtained from $\mathcal I$ by replacing $R$ in $C$ by $R'$.
\end{itemize}
Then 
$\AIP$ returns ``Yes'' on $\mathcal I$ if and only if 
$\AIP$ returns ``Yes'' on $\mathcal I'$.
\end{claim}
\begin{proof}
It is sufficient to construct an AIP solution for $\mathcal I$ from an AIP solution 
for $\mathcal I'$.
To do so, remove the variable $C^{(a_1,\dots,a_{k},a_{k+1},\dots,a_{m})}$ and adjust 
the values of only 3 variables as follows:
\begin{align*}
C^{(a_1,\dots,a_{k},b_{k+1},\dots,b_{m})}:=C^{(a_1,\dots,a_{k},b_{k+1},\dots,b_{m})} + C^{(a_1,\dots,a_{k},a_{k+1},\dots,a_{m})}& \\ 
C^{(b_1,\dots,b_{k},a_{k+1},\dots,a_{m})}:=C^{(b_1,\dots,b_{k},a_{k+1},\dots,a_{m})} + C^{(a_1,\dots,a_{k},a_{k+1},\dots,a_{m})}& \\ 
C^{(b_1,\dots,b_{k},b_{k+1},\dots,b_{m})}:=C^{(b_1,\dots,b_{k},b_{k+1},\dots,b_{m})} - C^{(a_1,\dots,a_{k},a_{k+1},\dots,a_{m})}& 
\end{align*}%
\end{proof}

Therefore, $\AIP$ generally cannot solve CSPs involving relations 
without the parallelogram property (see Sections \ref{SubsectionClassicTemplates} and \ref{SubsectionDisconnectedTemplates}).
Another weakness of $\AIP$ is
its inability to solve linear equations modulo 2 and modulo 3 simultaneously, as 
shown in Section \ref{SubsectionAbsorbingToZDva}.

\subsection{Classic templates}\label{SubsectionClassicTemplates}

\textbf{Horn-3-SAT} $= (\{0,1\};\{0,1\}^{3}\setminus\{(1,1,0)\},\{0\},\{1\})$.
Horn-3-SAT can be solved by the arc-consistency algorithm, as witnessed 
by the fact that the $n$-ary conjunction, which is a totally symmetric function, belongs to 
$\Pol(\text{Horn-3-SAT})$ for every $n\ge 2$.
On the other hand, $\AIP$ cannot solve Horn-3-SAT. 
Indeed, the minimal relation with the parallelogram property containing 
$\{0,1\}^{3}\setminus\{(1,1,0)\}$
is the full relation $\{0,1\}^{3}$.
By Claim \ref{CLAIMAIPParallelogram}, running $\AIP$ on Horn-3-SAT is equivalent to 
running $\AIP$ on $(\{0,1\};\{0,1\}^{3},\{0\},\{1\})$, 
which can only return ``No'' on trivial instances.

\textbf{2-SAT} $= (\{0,1\};\text{all binary relations})$. 
Let 
$\mathcal I = (x_1\le x_2)\wedge (x_2\le x_3)\wedge (x_3\le x_4)\wedge (x_4\le x_1)\wedge 
(x_1\neq x_3)$, which clearly has no solution. However, it has a BLP solution: assign the value $\frac{1}{2}$
to the tuples $(0,0),(1,1)$ in each $\le$-constraint and
to the tuples $(0,1),(1,0)$ in the $\neq$-constraint,
and assign $0$ to other tuples.
Similarly, $\Singl\AIP$ returns ``Yes'' on $\mathcal I$, since fixing one variable does not yield a contradiction by itself;
contradictions only arise when combining several $\le$-constraints. 
By Claim \ref{CLAIMAIPParallelogram}, AIP treats $\le$ as indistinguishable from the full relation $\{0,1\}^{2}$, 
which implies that $\Singl\AIP$ also returns ``Yes''.
Thus, both $\BLP$ and $\Singl\AIP$ fail to solve 2-SAT.
To see that $\BLP+\AIP$ solves 2-SAT, notice that 
$\Pol(\text{2-SAT})$ contains a majority function of any odd arity, which is a symmetric function.
To see that $\Singl\ArcCons$ solves 2-SAT, 
define, for every $m,n\in\mathbb N$, an $(\underbrace{\overline{m},\dots,\overline{m}}_{n})$-palette totally symmetric function
$f\in\Pol(\text{2-SAT})$ 
as follows:
$f(\mathbf x_1,\dots,\mathbf x_n)$ returns the value from the first constant block if it exists,
and $\mathbf x_{1}^1$ otherwise.

$\boldLIN{n}{p}=(\mathbb Z_{p}; \{a_{1}x_1+\dots+a_{n}x_{n} = a_{0}\mid a_{0},a_1,\dots,a_n\in\mathbb Z_{p}\})$ consists of linear equations with $n$ variables modulo $p$. 
$\AIP$ solves it as $\Pol(\LIN{n}{p})$ includes
an alternating function $x_{1}-x_2+x_3-\dots-x_{2n}+x_{2n+1}$.
To see that $\Singl\BLP$ does not solve 
$\LIN{3}{p}$, consider an instance 
consisting of two contradictory equations: 
$x_1+x_2+x_3=0\wedge x_1+x_2+x_{3}=1$. 
Fixing a variable to a concrete value and assigning a uniform distribution to 
the rest (as in Claim \ref{CLAIMABLPUniformSolution}) yields a BLP solution, 
so $\Singl\BLP$ incorrectly returns ``Yes''.

\subsection{Disconnected templates}\label{SubsectionDisconnectedTemplates}

$\left(\begin{smallmatrix}\mathbf{0}&\mathbf{1}&\mathbf{2}&\mathbf{3}&\mathbf{4}\\\mathbf{1}&\mathbf{0}&\mathbf{3}&\mathbf{4}&\mathbf{2}\end{smallmatrix}\right) =
(\{0,1,2,3,4\};\{(0,1),(1,0),(2,3),(3,4),(4,2)\})$, 
which is a structure containing a single binary relation given by the graph of a bijection. 
To show that both $\Singl\ArcCons$ and $\Singl\AIP$ solve it, let us define 
$f\in\Pol\left(\begin{smallmatrix}0&1&2&3&4\\1&0&3&4&2\end{smallmatrix}\right)$ 
that is simultaneously 
$(\underbrace{\overline{2\ell+1},\dots,\overline{2\ell+1}}_{n})$-palette totally symmetric and palette alternating. Specifically, let
$f(\mathbf x_1,\dots,\mathbf x_n)$ return the value from the first constant block if it exists,
and $\mathbf x_{1}^1$ otherwise.

Denote the only relation in this structure by $\rho$,
and consider the instance 
$\mathcal I = \rho(x_1,x_2)\wedge \rho(x_2,x_3)\wedge \rho(x_3,x_1)\wedge 
\rho(x_1,x_4)\wedge \rho(x_4,x_1)$, which has no solution.
By Claim \ref{CLAIMABLPUniformSolution}, $\mathcal I$ nevertheless has a BLP solution in which all the tuples receive positive assignments. 
To show that $\BLP+\AIP$ also returns ``Yes'' on $\mathcal I$, it is sufficient to construct an AIP 
solution. One such solution assigns value $1$ to the tuples $(2,3),(3,4),(4,2)$ in each constraint, and
value $-1$ to the tuples $(0,1),(1,0)$ in each constraint. This example illustrates
that $\AIP$ is powerless when the instance is disconnected (here the components $\{0,1\}$ and $\{2,3,4\}$ are disconnected). This weakness becomes even more evident in the next template.

$\bm{\mathbf Z_{2}\cup\mathbf Z_{3}} = 
(\{0,1,0',1',2'\};l_{0}^{\mathbb Z_{2}},l_{1}^{\mathbb Z_{2}},l_{0}^{\mathbb Z_{3}},l_{1}^{\mathbb Z_{3}},l_{2}^{\mathbb Z_{3}})$, where the relations are defined by: 
\begin{align*}l_{a}^{\mathbb Z_{2}}(x_1,x_2,x_3) =& 
(x_1,x_2,x_3\in \{0,1\} 
\wedge x_1+x_2+x_3=a\pmod 2)\vee
x_1,x_2,x_3\in \{0',1',2'\},\\
l_{b}^{\mathbb Z_{3}}(x_1,x_2,x_3) =& 
(x_1,x_2,x_3\in \{0',1',2'\} 
\wedge x_1+x_2+x_3=b\pmod 3)\vee
x_1,x_2,x_3\in \mathbb \{0,1\}.
\end{align*}
Thus, we have linear equations modulo 2 on $\{0,1\}$ and independent linear equations modulo 3 on $\{0',1',2'\}$. Any instance of $\CSP(\mathbb Z_{2}\cup\mathbb Z_{3})$
decomposes into two independent systems, one over $\mathbb Z_2$ and one over $\mathbb Z_{3}$. 
Solving the instance amounts to solving each system separately and returning ``Yes'' if one of them has a solution. 
Nevertheless, $\AIP$ fails here: it returns ``Yes'' on any instance 
of $\CSP(\mathbb Z_{2}\cup\mathbb Z_{3})$. The key idea comes from the following claim 
(the definition of a $d$-solution was given in Section \ref{SubsectionCombinationOfAlgorithms}).

\begin{claim}\label{CLAIMTrivialDSolutionToLinear}
Suppose 
$\mathcal I$ is an instance of $\LIN{n}{p}$ for a prime $p$ and $n\in\mathbb N$.
Then $\mathcal I$ has a 
$p^{N}$-solution for every $N\ge n-1$.
\end{claim}
\begin{proof}
It is sufficient to 
set
$C^{\mathbf a} = p^{n}/|R^{\mathbb A}|$
for every constraint $C = R(x_{1},\dots,x_{m})$ and every tuple $\mathbf a\in R$,
and
$x^{b} = p^{n-1}$ for every variable $x$ and every $b\in \mathbb Z_{p}$.
\end{proof}

Applying Claim \ref{CLAIMTrivialDSolutionToLinear} to the restriction of an instance of $\CSP(\mathbb Z_{2}\cup\mathbb Z_{3})$ to $\{0,1\}$ and to $\{0',1',2'\}$, 
 we obtain:
\begin{itemize}
    \item every instance  has 
    a $2^{n}$-solution for any $n\ge 2$;
    \item every instance has 
    a $3^{n}$-solution for any $n\ge 2$;
    \item 
    a linear combination of a $2^{n}$-solution and a $3^n$-solution yields an AIP solution.
\end{itemize}  

Thus, $\AIP$ returns ``Yes'' on all instances of $\CSP(\mathbb Z_{2}\cup\mathbb Z_{3})$.
Even $\Singl\BLP+\AIP$ fails, since $\Singl\BLP$ cannot reduce 
a system of linear equations (with more than 2 variables); 
therefore, running $\Singl\BLP+\AIP$ is not substantially different from running 
$\AIP$ alone.

In contrast, 
$\Singl\AIP$  solves the problem. Fixing a concrete element in the $\Singl$-algorithm restricts the instance to either the $\{0,1\}$-part or the $\{0',1',2'\}$-part, and each part can be solved by $\AIP$.
This is also witnessed by 
a $(\underbrace{\overline{2\ell+1},\dots,\overline{2\ell+1}}_{n})$-palette alternating function $f(\mathbf x_{1},\dots,\mathbf x_{n})$ 
defined as follows:
\begin{itemize}
    \item choose the first block $\mathbf x_{i}$ consisting entirely of elements from either $\{0,1\}$ or 
$\{0',1',2'\}$;
    \item if no such block exists, return $\mathbf x_{1}^{1}$;
    \item otherwise, apply the alternating function to $\mathbf x_{i}$, that is, 
$f(\mathbf x_{1},\dots,\mathbf x_{n}) = \mathbf x_{i}^{1}-\mathbf x_{i}^{2}+\mathbf x_{i}^{3}-\dots -\mathbf x_{i}^{2\ell}+
\mathbf x_{i}^{2\ell+1}$, where $+$ and $-$ are interpreted modulo 2 if 
$\mathbf x_{i}\in\{0,1\}^{2\ell+1}$, and modulo 3 if 
$\mathbf x_{i}\in\{0',1',2'\}^{2\ell+1}$.    
\end{itemize}

Although this template is extremely simple, 
it serves as a counterexample to several algorithms that were conjectured to solve all tractable CSPs, including
$\mathbb Z$-affine $k$-consistency, BLP+AIP, $\mathrm{BA}^{k}$, and CLAP
(see \cite{lichter2024limitations} for details).
Since hierarchies and higher levels of consistency are
beyond the scope of this paper, we omit the proofs.


\subsection{Implications of linear equations}\label{SubsectionAbsorbingToZDva}

The next template has a similar flavor to the previous ones, but 
its instances are more interesting because they combine  
implications with linear equations.

$\bm{\mathbf Z_2\leftarrow 2} = (\{0,1,2\};\mathcal R_{L}\cup \{\{0,1\},S_{0,1}^{=}\})$, 
where 
$S_{0,1}^{=}(x,y,z)=(x=y\vee z\in\{0,1\})$, and
$$\mathcal R_{L} = 
\left\{(x_1+\dots+x_{n}=c\wedge x_1,\dots,x_{n}\in\{0,1\}) 
\vee (x_1=\dots=x_n=2)\mid c\in\{0,1\},n\in\{3,4\}\right\}.$$
Thus, $S_{0,1}^{=}(x,y,z)$ can be read as an implication $z=2\rightarrow x=y$, 
and $\mathcal R_{L}$ consists of linear equations over $\{0,1\}$, augmented by the  additional all-2 tuple $(2,2,\dots,2)$.
The polymorphisms of this structure are well-understood, as 
it corresponds to one of the minimal Taylor clones \cite{jankovec2023minimalni}.
Define the ternary operation:  
$$f(x_1,x_2,x_3) = \begin{cases}
x_1+x_2+x_3 \pmod 2,& \text{if $x_{1},x_2,x_3\in\{0,1\}$}\\
2,& \text{if $x_{1}=x_2=x_3=2$}\\
\text{first element different from 2 in $x_1,x_2,x_3$}, & \text{otherwise}
\end{cases}$$


\begin{claim}[\cite{jankovec2023minimalni}]
$\Pol(\mathbb Z_{2}\leftarrow 2) = \Clo(f)$. 
\end{claim}



Below we define an instance $\mathcal I$ of 
$\CSP(\mathbb Z_{2}\leftarrow 2)$ 
that fools $\Singl\BLP+\AIP$.
Here, each linear equation 
should be understood as the corresponding  
relation from $\mathcal R_{L}$, that is, a linear equation modulo 2 augmented by the tuple $(2,2,\dots,2)$.

$$
\begin{cases}
x_{1}\in\{0,1\}\\
\hline
x_1+x_2+x_3+x_4 = 0\\
x_3+x_5+x_6 = 0 \;\;\;\;\;\;\;\;\;\;\;\;\;\;\;\;\;\;\;\;\;\;\;\;\;(L_1)\\
x_4+x_5+x_6 = 1\\
\hline
x_1=x_2\vee x_7\in\{0,1\}\\
\hline
x_7+x_8+x_9+x_{10} = 0\\
x_{9}+x_{11}+x_{12} = 0 \;\;\;\;\;\;\;\;\;\;\;\;\;\;\;\;\;\;\;\;\;\;\;(L_2)\\
x_{10}+x_{11}+x_{12} = 1\\
\hline
x_7=x_8\vee x_{13}\in\{0,1\}\\
\hline
x_{13}+x_{14}+x_{15}+x_{16}=0\\
x_{13}+x_{14}+x_{17}+x_{18}=0 \;\;\;\;\;\;\;\;\;\;\;\;(L_3)\\
x_{15}+x_{16}+x_{17}+x_{18}=1
\end{cases}
$$

We have three systems of linear equations $L_1$, $L_2$, and $L_{3}$, 
with implications connecting them.
Each system admits the trivial fake solution where all variables are set to 2. 
In addition, $L_1$ and $L_2$ have nontrivial solutions on $\{0,1\}$ but 
$L_3$ does not. 
A natural algorithm proceeds as follows. Solve $L_3$ first. 
Since it has no solutions on $\{0,1\}$, the variable $x_{13}$ must be equal to 2.
Then $x_7=x_8$, which can be added as an equation to $L_2$. 
We check that $L_2$ has no solution satisfying this additional equation, which implies that $x_7$ must be equal to 2 and $x_1=x_2$.
Finally, the linear system $L_1$ together with $x_1=x_2$ has no solution, 
contradicting $x_1\in\{0,1\}$.

However, $\Singl\BLP+\AIP$ incorrectly returns ``Yes'' on the instance $\mathcal I$. 
$\Singl\BLP$ does not reduce any domains (except the trivial restriction $D_{x_{1}}=\{0,1\}$),
since every linear equation can be satisfied by a uniform distribution.
By Claim \ref{CLAIMAIPParallelogram}, each constraint relation can be replaced by its parallelogram closure without changing the outcome of AIP. 
Since the parallelogram closure of $S_{0,1}^{=}$ is the full relation, 
$\AIP$ simply ignores all
$S_{0,1}^{=}$-constraints.
Taking a genuine solution of $L_1$ over $\{0,1\}$, 
together with fake all-2 solutions for $L_2$ and $L_3$, yields a global 
solution that violates only $S_{0,1}^{=}$-constraints, 
which AIP ignores.
Thus $\AIP$ (and hence $\Singl\BLP+\AIP$) returns ``Yes'' on the unsatisfiable instance $\mathcal I$.

In contrast, $\Singl\AIP$ returns ``No'' on  $\mathcal I$.
After fixing $x_{13}$ to $0$ and $1$ 
the $\AIP$ algorithm returns ``No'', which excludes these values from the domain of $x_{13}$ and forces $x_{13} = 2$.
Hence, $x_7=x_8$. 
After fixing $x_7$ to 0 or 1, $\AIP$ again returns ``No'', reducing the domain of $x_{7}$ to $\{2\}$ and forcing $x_1=x_2$. 
With $x_1=x_2$, $\AIP$ returns ``No'' on the reduced instance. 
Thus, $\Singl\AIP$ outputs ``No'' on 
$\mathcal I$, as desired.

To witness that  $\Singl\AIP$
correctly solves 
$\CSP(\mathbb Z_{2}\leftarrow 2)$ 
we define 
a $(\underbrace{\overline{2\ell+1},\dots,\overline{2\ell+1}}_{n})$-palette 
alternating function 
$f(\mathbf x_{1},\dots,\mathbf x_{n})$ from $\Pol(\mathbb Z_{2}\leftarrow 2)$ as follows:
\begin{itemize}
    \item choose the first block $\mathbf x_{i}$ consisting entirely of elements from $\{0,1\}$;
    \item if no such block exists, return $\mathbf x_{1}^{1}$;
    \item otherwise, apply the alternating function to $\mathbf x_{i}$, that is, 
$f(\mathbf x_{1},\dots,\mathbf x_{n}) = \mathbf x_{i}^{1}+\mathbf x_{i}^{2}+\dots+\mathbf x_{i}^{2\ell+1} \pmod 2$.
\end{itemize}


The crucial difference between $\mathbb Z_{2}\leftarrow 2$ and 
$\mathbb Z_{2}\cup\mathbb Z_3$ is that $\mathbb Z_{2}\cup\mathbb Z_3$
can be solved by running the $\AIP$ algorithm once for 
each variable and each value, 
while for $\mathbb Z_{2}\leftarrow 2$ we must run $\AIP$ again and again, 
gradually reducing the instance. 
We believe that $\mathbb Z_{2}\leftarrow 2$ also cannot be solved by 
$\mathbb Z$-affine $k$-consistency or $\mathrm{BA}^{k}$, but we do not prove this formally.






\subsection{Dihedral group}\label{SUBSECTIONDihedralGroup}

The last and perhaps most important example comes from the 
dihedral group $\mathbf D_4$,  the group of symmetries of a square. 
This group is generated by two elements, 
$r$ and $s$, and is completely defined by the equalities
$r^4 = s^2 = 1$ and $rs = sr^3$.

We use the following representation of the dihedral group: 
$\mathbf A = ( \mathbb Z_{2}^{3};\circ)$,
where the operation is defined by $x\circ y=
(x^{1}+y^{1},x^{2}+y^{2}, 
x^{3}+y^{3}+x^{1}\cdot y^{2})$.
Then $\mathbf A$ is isomorphic to $\mathbf D_{4}$, 
as witnessed by the mapping
$(0,0,0)\to 1$, $(0,0,1)\to r^2$,
$(1,1,0)\to r$, $(1,1,1)\to r^3$,
$(0,1,0)\to s$, $(0,1,1)\to sr^2$,
$(1,0,0)\to sr$,  $(1,0,1)\to sr^3$. 
Although the elements of $\mathbb Z_{2}^{3}$ are vectors of length 3, 
we do not use boldface notation for them in order to distinguish them from tuples of such vectors.
However, we will still refer to their components as 
$a^{1}$, $a^{2}$, and $a^{3}$.
Let us define a relational structure $\mathbb D_4$ corresponding to the dihedral group.

$\bm{\mathbb  D_4}=(\mathbb Z_{2}^{3}; 
L_{3}^{1,2},L_{2,3}^{1}, L_{1,3}^{2}, E_{1-2},E_{2-3}^{0}, E_{1-3}^{0},R,\{(0,0,0)\})$, where  
\begin{align*}
L_{3}^{1,2}(x_1,x_2,x_3,x_4) = 
&(x_1^{1}=x_2^{1}\wedge 
x_1^{2}=x_2^{2}\wedge 
x_3^{1}=x_4^{1}\wedge 
x_3^{2}=x_4^{2}\wedge 
x_1^{3}+x_2^{3} = x_3^{3}+x_4^{3}),\\
L_{2,3}^{1}(x_1,x_2,x_3,x_4) =
&(x_1^{1}=x_2^{1}=x_3^{1}=x_4^{1}
\wedge 
x_1^{2}+x_2^{2} = x_3^{2}+x_4^{2}
\wedge 
x_1^{3}+x_2^{3} = x_3^{3}+x_4^{3}),\\
L_{1,3}^{2}(x_1,x_2,x_3,x_4) = 
&(x_1^{2}=x_2^{2}=x_3^{2}=x_4^{2}
\wedge 
x_1^{1}+x_2^{1} = x_3^{1}+x_4^{1}
\wedge 
x_1^{3}+x_2^{3} = x_3^{3}+x_4^{3}),\\
E_{1-2}(x,y) =& (x^{1} = y^{2}),\\
E_{2-3}^{0}(x,y) = 
&(x^{1}=y^{1}=0
\wedge 
x^{2}=y^{3}),\\
E_{1-3}^{0}(x,y) = 
&(x^{2}=y^{2}=0
\wedge 
x^{1}=y^{3}),\\
R(x,y) =& (x^{1} = y^{2}\wedge 
x^{2} = y^{1}\wedge
x^{3}+y^{3}=x^{1}\cdot y^{1}).
\end{align*}
In the above notation  $L_{1,3}^{2}$, $L_{2,3}^{1}$, and $L_{3}^{1,2}$
superscripts indicate which coordinate is "fixed" and subscripts indicate which coordinate satisfies a linear equation.
In Section \ref{SECTIONDihedralGroupProofs} 
we prove the following lemma.  

\begin{lem}\label{LEMRelationsForD4}
$\Pol(\mathbb D_{4})
=\Clo(x\circ y)$.
\end{lem}

Since $\CSP(\mathbb D_4)$ is trivial --- 
every instance admits a solution by assigning all variables to 0 --- 
we extend $\mathbb D_4$ with all constant relations $\{\{a\} \mid a\in\mathbb Z_{2}^{3}\}$ and denote the resulting structure by $\mathbb  D_4^{idemp}$.

\begin{cor}
$\Pol(\mathbb D_{4}^{idemp})$ is the set of all 
idempotent operations from $\Clo(x\circ y)$.
\end{cor}

We now show that $\Singl(\BLP+\AIP)$ does not solve $\CSP(\mathbb D_{4}^{idemp})$. 
To witness this, we construct a concrete instance of $\CSP(\mathbb D_{4}^{idemp})$ with no solution, for which $\Singl(\BLP+\AIP)$ nevertheless returns ``Yes''. 
Notice that all relations in $\mathbb D_{4}^{idemp}$, except for $R$, are defined by conjunctions of linear equations on the components of the variables. 
Hence, any instance of $\CSP(\mathbb D_{4}^{idemp})$ not involving  $R$ can be solved directly by the $\AIP$ algorithm. However, the nonlinear part in $R$ 
breaks $\AIP$. 
Let us build a concrete CSP instance. 
We start with a very symmetric formula $\Phi_{x,y}$ described 
in Figure \ref{fig:triangle}, where vertices are constraints and edges are variables.
Formally, $\Phi_{x,y}$  is the following formula
$$
\bigwedge_{\substack{\{i,j,k\}=\{1,2,3\}\\ j<k}}
L_{2,3}^{1}(x_{i,j},y_{i,j},x_{i,k},y_{i,k})\wedge 
\bigwedge_{i<j}
R(x_{i,j},x_{j,i}) \wedge
\bigwedge_{i<j}
R(y_{i,j},y_{j,i}).$$
Let  
$\Phi_{x,y}'$ be the formula obtained from 
$\Phi_{x,y}$ by replacing the constraint 
$R(x_{2,3},x_{3,2})$ by 
$R(x_{2,3},x_{3,2}')$ (see Figure \ref{fig:triangleprime}).
Let 
$\Phi_{u,v}'$ 
be obtained from 
$\Phi_{x,y}'$ by replacing each $x$-variable with the corresponding $u$-variable, 
and each $y$-variable with the corresponding $v$-variable.
The CSP instance 
$\mathcal I$ is defined by the following formula and shown in Figure \ref{fig:d4instance}:
$$\Phi_{x,y}'\wedge \Phi_{u,v}'\wedge 
L_{3}^{1,2}(x_{3,2},x_{3,2}',w_1,w_2)\wedge 
L_{3}^{1,2}(u_{3,2},u_{3,2}',w_2,w_3)\wedge
w_1 = (0,0,0)\wedge w_3 = (0,0,1).$$

\begin{figure}[h]
\centering
\begin{minipage}{0.45\textwidth}
    \centering
\begin{tikzpicture}[
  main/.style  = {circle, draw, fill=green!20, minimum size=30pt, inner sep=2pt,
                  font=\small},
    lthree/.style  = {rectangle,
  draw,
  fill=blue!20!yellow,
  minimum width=30pt,
  minimum height=20pt,
  inner sep=2pt,
  font=\small},
  rvert/.style = {circle, draw, fill=white,   minimum size=18pt, inner sep=1pt,
                  font=\footnotesize},
  lbl/.style   = {font=\scriptsize, fill=white, inner sep=1pt, sloped},
  ell/.style = {
  ellipse,
  draw,
  fill=blue!20,
  minimum width=30pt,
  minimum height=18pt,
  inner sep=2pt,
  font=\small
}
]
\node[main] (v2) at (  0.000,  2.500) {$L_{2,3}^{1}$};
\node[main] (v1) at ( -2.165, -1.250) {$L_{2,3}^{1}$};
\node[main] (v3) at (  2.165, -1.250) {$L_{2,3}^{1}$};

\node[rvert] (R12a) at (-1.603,  0.925) {$R$};
\node[rvert] (R12b) at (-0.563,  0.325) {$R$};

\node[rvert] (R23a) at ( 1.603,  0.925) {$R$};
\node[rvert] (R23b) at ( 0.563,  0.325) {$R$};

\node[rvert] (R13a) at ( 0.000, -0.650) {$R$};
\node[rvert] (R13b) at ( 0.000, -1.850) {$R$};

\draw[thick] (v2) -- node[lbl] {$x_{2,1}$}
             (R12a) -- node[lbl] {$x_{1,2}$} (v1);
\draw[thick] (v2) -- node[lbl] {$y_{2,1}$}
             (R12b) -- node[lbl] {$y_{1,2}$} (v1);

\draw[thick] (v2) -- node[lbl] {$x_{2,3}$}
             (R23a) -- node[lbl] {$x_{3,2}$} (v3);
\draw[thick] (v2) -- node[lbl] {$y_{2,3}$}
             (R23b) -- node[lbl] {$y_{3,2}$} (v3);

\draw[thick] (v1) -- node[lbl] {$x_{1,3}$}
             (R13a) -- node[lbl] {$x_{3,1}$} (v3);
\draw[thick] (v1) -- node[lbl] {$y_{1,3}$}
             (R13b) -- node[lbl] {$y_{3,1}$} (v3);

\end{tikzpicture}
\caption{Formula $\Phi_{x,y}$}   \label{fig:triangle}
\end{minipage}
\hfill
\begin{minipage}{0.45\textwidth}
    \centering
\begin{tikzpicture}[
  main/.style  = {circle, draw, fill=green!20, minimum size=30pt, inner sep=2pt,
                  font=\small},
    lthree/.style  = {rectangle,
  draw,
  fill=blue!20!yellow,
  minimum width=30pt,
  minimum height=20pt,
  inner sep=2pt,
  font=\small},
  rvert/.style = {circle, draw, fill=white,   minimum size=18pt, inner sep=1pt,
                  font=\footnotesize},
  lbl/.style   = {font=\scriptsize, fill=white, inner sep=1pt, sloped},
  ell/.style = {
  ellipse,
  draw,
  fill=blue!20,
  minimum width=30pt,
  minimum height=18pt,
  inner sep=2pt,
  font=\small
}
]
\node[main] (v2) at (  0.000,  2.500) {$L_{2,3}^{1}$};
\node[main] (v1) at ( -2.165, -1.250) {$L_{2,3}^{1}$};
\node[main] (v3) at (  2.165, -1.250) {$L_{2,3}^{1}$};

\node[rvert] (R12a) at (-1.603,  0.925) {$R$};
\node[rvert] (R12b) at (-0.563,  0.325) {$R$};

\node[rvert] (R23a) at ( 1.603,  0.925) {$R$};
\node[rvert] (R23b) at ( 0.563,  0.325) {$R$};

\node[rvert] (R13a) at ( 0.000, -0.650) {$R$};
\node[rvert] (R13b) at ( 0.000, -1.850) {$R$};

\draw[thick] (v2) -- node[lbl] {$x_{2,1}$}
             (R12a) -- node[lbl] {$x_{1,2}$} (v1);
\draw[thick] (v2) -- node[lbl] {$y_{2,1}$}
             (R12b) -- node[lbl] {$y_{1,2}$} (v1);

\draw[thick] (v2) -- node[lbl] {$x_{2,3}$}(R23a);
\draw[thick] (v2) -- node[lbl] {$y_{2,3}$}
             (R23b) -- node[lbl] {$y_{3,2}$} (v3);

\draw[thick] (v1) -- node[lbl] {$x_{1,3}$}
             (R13a) -- node[lbl] {$x_{3,1}$} (v3);
\draw[thick] (v1) -- node[lbl] {$y_{1,3}$}
             (R13b) -- node[lbl] {$y_{3,1}$} (v3);
\draw[thick] (R23a) -- node[lbl] {$x_{3,2}'$} ( 4.165,  0.925);
\draw[thick] (v3) -- node[lbl] {$x_{3,2}$} (  4.165, -1.250);
\end{tikzpicture}
\caption{Formula $\Phi_{x,y}'$}\label{fig:triangleprime}
\end{minipage}
\end{figure}
\begin{figure}[h]
    \centering
\begin{tikzpicture}[
  main/.style  = {circle, draw, fill=green!20, minimum size=30pt, inner sep=2pt,
                  font=\small},
    lthree/.style  = {rectangle,
  draw,
  fill=blue!20!yellow,
  minimum width=30pt,
  minimum height=20pt,
  inner sep=2pt,
  font=\small},
  rvert/.style = {circle, draw, fill=white,   minimum size=18pt, inner sep=1pt,
                  font=\footnotesize},
  lbl/.style   = {font=\scriptsize, fill=white, inner sep=1pt, sloped},
  ell/.style = {
  ellipse,
  draw,
  fill=blue!20,
  minimum width=30pt,
  minimum height=18pt,
  inner sep=2pt,
  font=\small
}
]


\node[main] (Av2) at (  0.000,  2.500) {$L_{2,3}^{1}$};
\node[main] (Av1) at ( -2.165, -1.250) {$L_{2,3}^{1}$};
\node[main] (Av3) at (  2.165, -1.250) {$L_{2,3}^{1}$};

\node[rvert] (AR12a) at (-1.603,  0.925) {$R$};
\node[rvert] (AR12b) at (-0.563,  0.325) {$R$};

\node[rvert] (AR23a) at ( 1.603,  0.925) {$R$};
\node[rvert] (AR23b) at ( 0.563,  0.325) {$R$};

\node[rvert] (AR13a) at ( 0.000, -0.650) {$R$};
\node[rvert] (AR13b) at ( 0.000, -1.850) {$R$};

\draw[thick] (Av2) -- node[lbl] {$x_{2,1}$}
             (AR12a) -- node[lbl] {$x_{1,2}$} (Av1);
\draw[thick] (Av2) -- node[lbl] {$y_{2,1}$}
             (AR12b) -- node[lbl] {$y_{1,2}$} (Av1);

\draw[thick] (Av2) -- node[lbl] {$x_{2,3}$}(AR23a);
\draw[thick] (Av2) -- node[lbl] {$y_{2,3}$}
             (AR23b) -- node[lbl] {$y_{3,2}$} (Av3);

\draw[thick] (Av1) -- node[lbl] {$x_{1,3}$}
             (AR13a) -- node[lbl] {$x_{3,1}$} (Av3);
\draw[thick] (Av1) -- node[lbl] {$y_{1,3}$}
             (AR13b) -- node[lbl] {$y_{3,1}$} (Av3);


\node[main] (Bv2) at ( 10.700,  2.500) {$L_{2,3}^{1}$};
\node[main] (Bv3) at (  8.535, -1.250) {$L_{2,3}^{1}$};
\node[main] (Bv1) at ( 12.865, -1.250) {$L_{2,3}^{1}$};

\node[rvert] (BR23a) at ( 9.097,  0.925) {$R$};
\node[rvert] (BR23b) at (10.137,  0.325) {$R$};

\node[rvert] (BR12a) at (12.303,  0.925) {$R$};
\node[rvert] (BR12b) at (11.263,  0.325) {$R$};

\node[rvert] (BR13a) at (10.700, -0.650) {$R$};
\node[rvert] (BR13b) at (10.700, -1.850) {$R$};

\draw[thick] (Bv2) -- node[lbl] {$u_{2,3}$}(BR23a);
\draw[thick] (Bv2) -- node[lbl] {$v_{2,3}$}
             (BR23b) -- node[lbl] {$v_{3,2}$} (Bv3);

\draw[thick] (Bv2) -- node[lbl] {$u_{2,1}$}
             (BR12a) -- node[lbl] {$u_{1,2}$} (Bv1);
\draw[thick] (Bv2) -- node[lbl] {$v_{2,1}$}
             (BR12b) -- node[lbl] {$v_{1,2}$} (Bv1);

\draw[thick] (Bv3) -- node[lbl] {$u_{3,1}$}
             (BR13a) -- node[lbl] {$u_{1,3}$} (Bv1);
\draw[thick] (Bv3) -- node[lbl] {$v_{3,1}$}
             (BR13b) -- node[lbl] {$v_{1,3}$} (Bv1);

\node[lthree] (M1) at (4.150, 0.025) {$L_{3}^{1,2}$};
\node[lthree] (M2) at (6.550, 0.025) {$L_{3}^{1,2}$};
\node[ell] (w1) at (4.150, 2.025) {$(0,0,0)$};
\node[ell] (w3) at (6.550, 2.025) {$(0,0,1)$};

\draw[thick] (M1) to[out=-33, in=213] node[lbl] {$w_2$} (M2);
\draw[thick] (AR23a) to[out=-10, in=147] node[lbl] {$x_{3,2}'$} (M1);
\draw[thick] (Av3) to[out=40, in=213] node[lbl] {$x_{3,2}$} (M1);
\draw[thick] (BR23a) to[out=190, in=33] node[lbl] {$u_{3,2}'$} (M2);
\draw[thick] (Bv3) to[out=160, in=-33]  node[lbl] {$u_{3,2}$} (M2);

\draw[thick] (w1) to[out=-40, in=33]  node[lbl] {$w_1$} (M1);
\draw[thick] (w3) to[out=-160, in=147] node[lbl] {$w_3$} (M2);
\end{tikzpicture}
    \caption{Instance $\mathcal I$ fooling $\Singl(\BLP$+$\AIP)$}
    \label{fig:d4instance}
\end{figure}


\begin{claim}\label{claim:D4FakedSolution}
$\Phi_{x,y}'$ has an AIP solution such that 
the weight assigned to the evaluations 
$x_{3,2}=(1,1,1)$ and 
$x_{3,2}'=(1,1,0)$ is 1, and other weights for 
the variables $x_{3,2}$ and $x_{3,2}'$ are 0.
\end{claim}

\begin{proof}
We assign weights to the tuples from the $L_{2,3}^{1}$-constraints according to the following table:
$$\begin{array}{cccc|cccc|cccc}
\toprule
& \multicolumn{3}{c}{\text{weights}} 
& & \multicolumn{3}{c}{\text{weights}} 
& & \multicolumn{3}{c}{\text{weights}} \\
\cmidrule(lr){2-4}\cmidrule(lr){6-8}\cmidrule(lr){10-12}
L_{2,3}^{1} & 1 & 1 & -1 & L_{2,3}^{1} & 1 & 1 & -1 & L_{2,3}^{1} & 1 & 1 & -1 \\
\midrule
x_{1,2} & (0,0,0) & (1,1,0) & (0,0,0) & x_{2,1} & (0,0,0) & (1,1,1) & (0,0,0) & x_{3,1} & (0,0,0) & (1,1,1) & (0,0,0) \\
y_{1,2} & (0,1,0) & (1,0,0) & (0,0,0) & y_{2,1} & (0,1,0) & (1,0,0) & (0,0,0) & y_{3,1} & (0,1,0) & (1,0,0) & (0,0,0) \\
x_{1,3} & (0,0,0) & (1,1,0) & (0,0,0) & x_{2,3} & (0,0,0) & (1,1,1) & (0,0,0) & x_{3,2} & (0,0,0) & (1,1,1) & (0,0,0) \\
y_{1,3} & (0,1,0) & (1,0,0) & (0,0,0) & y_{2,3} & (0,1,0) & (1,0,0) & (0,0,0) & y_{3,2} & (0,1,0) & (1,0,0) & (0,0,0) \\
\bottomrule
\end{array}$$

For the $R$-constraints applied to 
$y$-variables the weight of  
$((1,0,0),(0,1,0))$ and $((0,1,0),(1,0,0))$ is 1, 
and the weight of  $((0,0,0),(0,0,0))$ 
is -1.
For the $R$ constraints applied to 
$x$-variables we just choose the tuple
$((1,1,1),(1,1,0))$ with weight 1.
As a result, we obtain the same distribution on all $y$-variables. 
For each $x$-variable we either choose value $(1,1,0)$, or $(1,1,1)$.
Knowing that $x_{2,3} = (1,1,1)$, we derive that 
$x_{3,2}' = (1,1,0)$.
The key feature of this AIP solution is that exactly one of the two $R$-constraints on every edge of the triangle receives the evaluation 
$((1,1,0),(1,1,1))$, and only in such evaluations can the third coordinate be flipped. 
As we have 3 edges, we flip the last coordinate three times, which gives us the required result.
\end{proof}

\begin{claim}\label{claim:TriangleImpliesEquality}
$\Phi_{x,y}'\Rightarrow x_{3,2}^{3} = x_{3,2}'^{\;3}$.
\end{claim}
\begin{proof}
First, notice that 
$x_{1,2}^{1} = y_{1,2}^{1}$ by the left constraint $L_{2,3}^{1}$, 
then by the $R$-constraints,
we obtain 
$x_{2,1}^{2} = y_{2,1}^{2}$.
By the top $L_{2,3}^{1}$-constraint,
we obtain $x_{2,3}^{2} = y_{2,3}^{2}$.
Then by the $R$-constraints  we derive 
$x_{3,2}'^{\;1} = y_{3,2}^{1}$.
Using the $L_{2,3}^{1}$-constraints, 
we also obtain
$x_{i,j}^{1} = y_{i,j}^{1}$
whenever $i\neq j$.
It remains to calculate the sum of the equations on
the third coordinate coming from 
all the constraints 
of 
$\Phi_{x,y}'$. After canceling repeated terms, we obtain:
$x_{3,2}^{3} + x_{3,2}'^{\;3}=0$.
\end{proof}

\begin{cor}\label{cor:NoSolutionForD4Example}
$\mathcal I$ has no solutions.     
\end{cor}
\begin{proof}
 By Claim \ref{claim:TriangleImpliesEquality}
 applied to the left and right triangles we obtain 
 $x_{3,2}^{3} = x_{3,2}'^{\;3}$
 and  $u_{3,2}^{3} = u_{3,2}'^{\;3}$.
 Then the left $L_{3}^{1,2}$-constraint  implies
 $w_1^{3}=w_2^{3}$ and the right $L_{3}^{1,2}$-constraint  implies 
 $w_2^{3}=w_3^{3}$, which together with 
 $w_1 = (0,0,0)$ and $w_3 = (0,0,1)$ gives a contradiction.
\end{proof}

\begin{claim}\label{claim:D4ExampleFixing}
For any $a,b,c\in \mathbb Z_{2}$
there exists a solution of 
$\Phi_{x,y}$ such that 
$x_{1,2} = y_{1,2} = x_{1,3} = y_{1,3}= (a,b,c)$.
\end{claim}
\begin{proof}
To get a solution we 
set
$x_{2,1}=y_{2,1} = x_{3,1}=y_{3,1} = (b,a,c+ab)$, 
$x_{2,3}=y_{2,3}= (b,b,0)$, and
$x_{3,2} = y_{3,2} = (b,b,b)$.
\end{proof}




\begin{claim}
$\Singl(\BLP+\AIP)$ returns ``Yes'' on $\mathcal I$.
\end{claim}
\begin{proof}
Let us show that after running the algorithm 
the domain of all the variables but $w_1$, $w_2$, and $w_3$ is the whole domain 
$\mathbb Z_{2}^{3}$, 
the domain of $w_1$ is $\{(0,0,0)\}$,
the domain of $w_2$ is $\{(0,0,0),(0,0,1)\}$ and 
the domain of $w_3$ is $\{(0,0,1)\}$.
We choose a variable in the left triangle and fix it to some 
value. 
By Claim \ref{claim:D4ExampleFixing}, there exists a solution to 
the left triangle such that 
$x_{3,2} = x_{3,2}'$ and the chosen variable takes the prescribed value. 
We fix the variables in the left triangle to these values
and show that a BLP+AIP solution can be found for the remaining variables.
By Claim \ref{CLAIMABLPUniformSolution},
there exists a uniform BLP solution for the remaining variables, which 
confirms that BLP does not reduce any domain.
By Claim \ref{claim:D4FakedSolution},
there exists an AIP solution to the right triangle
such that $u_{3,2}' = (1,1,0)$ and $u_{3,2} = (1,1,1)$.
To complete the AIP solution, 
we set $w_{1} =w_2 = (0,0,0)$, 
and 
$w_{3} = (0,0,1)$. A similar construction works if we first fix a variable in the right triangle and then find an AIP solution for the left one.
In this case we set $w_{1} =(0,0,0)$
and $w_2 = w_{3} = (0,0,1)$ in the end. 
This also confirms that elements $(0,0,0)$ and $(0,0,1)$
are not removed from the domain of $w_2$ if we reduce 
$w_2$ to one of them and run BLP+AIP.
\end{proof}



Another way to prove that 
$\Singl(\BLP+\AIP)$ does not solve 
$\CSP(\mathbb D_{4}^{idemp})$ is to apply Theorem \ref{THMSinglBLPAIP} together with the following lemma proved in Section \ref{SECTIONDihedralGroupProofs}.

\begin{lem}\label{LEMNoPaletteSymmetricDFour}
$\Clo(x\circ y)$ does not contain an idempotent  $(\underbrace{\overline{\ell+1,\ell},\overline{\ell+1,\ell},\dots,\overline{\ell+1,\ell}}_{2n})$-palette symmetric operation
for $n,\ell\ge 4$.\footnote{This result was first observed by Michael Kompatscher.}
\end{lem}

\begin{cor}
The algorithm $\Singl(\BLP+\AIP)$ does not solve 
$\CSP(\mathbb D_{4}^{idemp})$.
\end{cor}


\section{Characterization of Singleton Algorithms}\label{SECTIONSingletonCharacterization}

In this section we prove the characterization of the singleton algorithms 
in terms of palette symmetric functions.
We start with a formal definition of a minion and show how abstract minions can be used to characterize the power of algorithms.
Then, in Section \ref{SUBSECTIONHalesJewett}
we apply the Hales-Jewett theorem to show 
that, under certain conditions, a homomorphic image of 
a minion contains more symmetric objects than the original.
In the last two subsections we show how to apply this result to 
obtain a characterization of the algorithms $\Singl\ArcCons$ and 
$\Singl(\BLP+\AIP)$, respectively.


\subsection{Minions and free structure $\mathbb F_{\mathcal M}(\mathbb A)$}
\label{SUBSECTIONMinionDefinition}
To aid understanding of the notion of an abstract minion we start with a function minion. 
For a function 
$f\colon A^{n}\to B$ and 
a mapping $\pi\colon[n]\to[n']$, the \emph{minor of 
$f$} given by $\pi$, denoted by $f^{\pi}$, is 
the function of arity $n'$ defined by 
$$f^{\pi}(x_1,\dots,x_{n'}):= f(x_{\pi(1)},\dots,x_{\pi(n)}).$$
A subset $\mathcal M\subseteq \mathcal O_{A,B}$
is called a \emph{function minion} 
if 
for any $f\in\mathcal M\cap\mathcal O_{A,B}^{(n)}$ 
and any mapping $\pi\colon [n]\to[n']$,
the minor 
$f^{\pi}$ belongs to $\mathcal M$. 
Thus, a function minion can be viewed as a weaker version of a clone, 
where we only allow to permute and identify variables, and add dummy variables.
For instance, the sets $\mathcal O_{A,B}$ and $\Pol(\mathbb A,\mathbb B)$ are function minions.

Variables (arguments) of an $n$-ary function are usually indexed by numbers $1,2,\dots,n$, which sometimes 
makes argumentation unnecessarily cumbersome, as it requires an additional  encoding.
To avoid this,  we consider  
functions $f\colon A^{N}\to B$ for finite sets $N$,
denoted by $\mathcal O_{A,B}^{(N)}$;
that is, each function is a mapping equipped with argument names from the set $N$. 
Then, for a function $f\colon A^{N}\to B$ and  a mapping $\pi:N\to N'$, where $N'$ is also a finite set, \emph{the minor of $f$} given by $\pi$ is defined by 
$f^{\pi}(\mathbf a) = f(\mathbf a\circ \pi)$
for every $\mathbf a\in A^{N'}$.
The advantages of this approach can already be seen in the definition of a free structure below (compare with the respective definitions in 
\cite{PCSPBible, BLPplusAIP}).

Thus, a \emph{function minion} $\mathcal M$  is a collection of sets 
$\mathcal M^{(N)}\subseteq \mathcal O_{A,B}^{(N)}$ for finite sets $N$
such that 
for any $f\in\mathcal M^{(N)}$ 
and any mapping $\pi\colon N\to N'$, where $N'$ is also a finite set, 
the minor 
$f^{\pi}$ belongs to $\mathcal M^{(N')}$. 
Then the set of all functions $\mathcal O_{A,B}$ can be viewed as a 
minion, whose $\mathcal O_{A,B}^{(N)}$-part consists of all mappings 
$A^{N}\to B$. 
Similarly, $\Pol(\mathbb A,\mathbb B)$ is a minion 
such that $f\in \Pol(\mathbb A,\mathbb B)^{(N)}$ whenever $f\in \mathcal O_{A,B}^{(N)}$ and 
for any $R$ from the signature of $\mathbb A$ 
and any $\mathbf a\colon N\to R^{\mathbb A}$ 
we have $f(\mathbf a)\in R^{\mathbb B}$. 
Thus, $\Pol(\mathbb A,\mathbb B)$ is no different than
the usual set of polymorphisms with arbitrary names chosen for the arguments.

Now let us define the notion of an abstract minion, where instead of functions we have arbitrary objects, and taking minors satisfies the same properties we have for function minions.

\textbf{Minion.}
\emph{(Abstract) minion} $\mathcal M$  
is a collection of sets $\mathcal M^{(N)}$ for finite sets $N$ together with 
\emph{a minor map} $\mathcal M^{(\pi)}\colon \mathcal M^{(N)}\to\mathcal M^{(N')}$ for every mapping
$\pi\colon N\to N'$, such that 
$\mathcal M^{(\mathrm{id}_{N})} = \mathrm{id}_{\mathcal M^{(N)}}$
for all finite sets $N$, and 
$\mathcal M^{(\pi)}\circ \mathcal M^{(\pi')} = \mathcal M^{(\pi\circ\pi')}$ 
whenever such a composition is well defined.
When a minion is clear from context, we write $f^{\pi}$ instead of $\mathcal M^{(\pi)}(f)$.
We refer to $\mathcal M^{(N)}$ as \emph{the $N$-ary part} of $\mathcal M$.

A collection of sets $\mathcal M_{0}^{(N)}$ for finite sets $N$
is called 
\emph{a subminion} of a minion $\mathcal M$ if 
$\mathcal M_{0}^{(N)}\subseteq \mathcal M^{(N)}$
for all $N$ and 
the minor maps $\mathcal M_{0}^{(\pi)}$  
and $\mathcal M^{(\pi)}$
coincide on $\mathcal M_0$.

\textbf{Minion homomorphism.} Let $\mathcal M$ and $\mathcal M'$ be minions. \emph{A minion homomorphism from 
$\mathcal M$ to $\mathcal M'$} is a collection of mappings 
$\xi^{(N)}\colon \mathcal M^{(N)}\to\mathcal M'^{(N)}$ for finite sets $N$, such that they preserve taking minors; that is, 
$\xi^{(N')}(\mathcal M^{(\pi)}(f)) = \mathcal M'^{(\pi)}(\xi^{(N)}(f))$ 
for all finite sets $N,N'$, all $f\in\mathcal M^{(N)}$, and 
all mappings $\pi\colon N\to N'$.
We write 
$\mathcal M_1\to\mathcal M_2$ to indicate
the existence of a minion homomorphism from $\mathcal M_1$ to $\mathcal M_2$.

\textbf{Free structure
$\mathbb F_{\mathcal M}(\mathbb A)$.} For a relational structure $\mathbb A$ and a minion $\mathcal M$, the free structure
$\mathbb F_{\mathcal M}(\mathbb A)$ is the structure 
with domain $\mathcal M^{(A)}$  (potentially infinite) and with the same
signature as $\mathbb A$. 
For each $k$-ary relation $R^{\mathbb  A}$,
the relation $R^{\mathbb F_{\mathcal M}(\mathbb A)}$ 
consists of all tuples $(f_1,\dots,f_{k})\in (\mathcal M^{(A)})^{k}$ 
such that  
there exists $g\in\mathcal M^{(R^{\mathbb A})}$ satisfying 
$
f_{i} = g^{\pi_{i}}$ 
for every $i$, where $\pi_{i}\colon R^{\mathbb A}\to A$ is a projection onto the $i$-th coordinate.

\begin{lem}[\cite{PCSPBible,BLPplusAIP}]\label{LEMThreeEquivalentConditionsForPCSP}
Suppose 
$\mathbb A$ and $\mathbb B$ are finite relational structures in the same signature and
$\mathcal M$ is a minion.
Then the following conditions are equivalent:
\begin{itemize}
    \item[(1)] $\mathbb X\to \mathbb F_{\mathcal M}(\mathbb A)$ implies 
    $\mathbb X\to\mathbb B$ for any finite $\mathbb X$;
    \item[(2)] $\mathbb F_{\mathcal M}(\mathbb A)\to \mathbb B$;
    \item[(3)] 
    $\mathcal M\to \Pol(\mathbb A,\mathbb B)$.
\end{itemize}
\end{lem}

\begin{proof}
(2)$\Rightarrow$(1) follows from the fact that we can compose minion homomorphisms. 
The implication (1)$\Rightarrow$(2) follows from the standard compactness argument.
In fact, by condition (1) any finite substructure of $\mathbb F_{\mathcal M}(\mathbb A)$ 
can be mapped to $\mathbb B$, which by compactness implies 
condition (2).
For more details see Lemma A.6 in \cite{BLPplusAIP} or Remark 7.13 in \cite{PCSPBible}.

Finally, Lemma 4.4 in \cite{PCSPBible} implies (2)$\Leftrightarrow$(3).
\end{proof}

\begin{cor}\label{CORMainCharacterizationOfAlgorithm}
Suppose 
$\mathbb A$ and $\mathbb B$ are finite relational structures in the same signature, 
$\mathfrak{A}$ is an algorithm for $\PCSP(\mathbb A,\mathbb B)$,
$\mathcal M$ is a nonempty minion, 
and the following conditions are equivalent 
for any instance $\mathbb X$: 
\begin{itemize}
    \item $\mathfrak A$ returns ``Yes'' on $\mathbb X$; 
    \item 
$\mathbb X\to \mathbb F_{\mathcal M}(\mathbb A)$.
\end{itemize}
Then 
$\mathfrak A$ solves $\PCSP(\mathbb A,\mathbb B)$ if
and only if $\mathcal M\to \Pol(\mathbb A,\mathbb B)$.
\end{cor}

\begin{proof}

First, let us show that $\mathfrak A$ returns ``Yes'' on any Yes-instance $\mathbb X$. Let $\varphi\colon \mathbb X\to \mathbb A$ be a homomorphism. 
Since the minion $\mathcal M$ is nonempty, 
there should be a unary object $f\in \mathcal M^{(\{1\})}$.
Then a homomorphism 
$\psi\colon \mathbb X\to \mathbb F_{\mathcal M}(\mathbb A)$
can be defined by 
$\psi(x) = f^{\pi_x}$, where 
$\pi_{x}\colon \{1\}\to A$ such that 
$\pi_{x}(1) = \varphi(x)$.
This implies that $\mathfrak A$ returns ``Yes'' on $\mathbb X$.

Second, let us show that the following conditions are equivalent:
\begin{itemize}
    \item $\mathfrak A$ solves $\PCSP(\mathbb A,\mathbb B)$;
    \item $\mathbb X\to \mathbb F_{\mathcal M}(\mathbb A)$ implies 
    $\mathbb X\to\mathbb B$ for any finite $\mathbb X$.
\end{itemize}

Assume that $\mathfrak A$ solves $\PCSP(\mathbb A,\mathbb B)$.
Then $\mathbb X\to \mathbb F_{\mathcal M}(\mathbb A)$ implies
that the algorithm returns ``Yes'' on the instance $\mathbb X$, 
which means that 
$\mathbb X\to \mathbb A$ or $\mathbb X$ is incorrect ($\mathbb X\to \mathbb B$ and $\mathbb X\not\to\mathbb A$).
Hence,  
$\mathbb X\to \mathbb F_{\mathcal M}(\mathbb A)$
implies $\mathbb X\to \mathbb B$ 
for any $\mathbb X$.
For the converse, assume that 
$\mathbb X\to \mathbb F_{\mathcal M}(\mathbb A)$
implies $\mathbb X\to \mathbb B$. 
Let us show that 
$\mathfrak A$ works correctly on any instance $\mathbb X$.
For a Yes-instance, it was shown earlier.
If 
$\mathfrak A$ returns ``Yes'' on $\mathbb X$, then 
$\mathbb X\to \mathbb F_{\mathcal M}(\mathbb A)$, which implies 
$\mathbb X\to \mathbb B$, hence $\mathbb X$ is indeed a Yes-instance.

To complete the proof, it is sufficient to use the equivalence of conditions (1) and (3) in Lemma \ref{LEMThreeEquivalentConditionsForPCSP}.
\end{proof}

\subsection{General characterization of singleton}\label{SUBSECTIONHalesJewett}

\begin{thm}[Hales-Jewett Theorem, \cite{HalesJewett1963regularity}]
Suppose $A$ and $B$ are finite sets.
Then there exists 
$n\in\mathbb N$ such that 
for any mapping 
$f\colon A^{n}\to B$ 
there exists
$(v_1,\dots,v_n)\in (A\cup\{x\})^{n}\setminus A^{n}$, where $x$ is a variable symbol, 
such that the unary mapping 
$h(x)\approx f(v_1,\dots,v_n)$ is a constant.
\end{thm}

Suppose $\circ$ is an operation defined 
on $\mathcal M^{(n)}$ for every finite set $n$.
We say that an operation $\circ$ \emph{respects taking minors} if 
for any $\pi\colon n\to N$ and any $f_1,f_2\in \mathcal M^{(n)}$ 
we have $f_{1}^{\pi}\circ f_{2}^{\pi} = 
(f_1\circ f_2)^{\pi}$.

\begin{thm}\label{THMHalesJewettApplication}
Suppose 
  $n$ and $N$ are finite sets,  
$\mathcal M_0$ is a minion
equipped with a semigroup operation $\circ$ that respects taking minors, 
$\mathcal M\subseteq \mathcal M_0$ is a subminion, 
$\mathcal N$ is a minion, $\mathcal N^{(n)}$ is finite,
$h\colon \mathcal M\to \mathcal N$ is a minion homomorphism,
$\mathcal A\subseteq  \mathcal M_0^{(N)}$, 
$I\in \mathcal M_0^{(N)}$, 
$\Omega$ is a set of mappings 
$N\to n$, 
$\omega\subseteq \Omega\times\Omega$,  and 
\begin{enumerate}

\item[(1)] $U^{\sigma_1} = U^{\sigma_2}$ for any $(\sigma_1,\sigma_2)\in\omega$
and $U\in \mathcal A$;




\item[(2)] for any $\sigma\in\Omega$ 
there exists $U\in\mathcal A$ such that 
$U^{\sigma} = I^{\sigma}$;

\item[(3)] 
$U_1\circ \dots\circ U_\ell \in \mathcal M$
for any $\ell\in\mathbb N$ and 
$(U_1,\dots,U_\ell)\in (\mathcal A\cup\{I\})^{\ell}\setminus \mathcal A^{\ell}$.
\end{enumerate}
Then there exists 
$M\in\mathcal M^{(N)}$ such that 
$h(M)^{\sigma_1} = h(M)^{\sigma_2}$ for any $(\sigma_1,\sigma_2)\in\omega$.
\end{thm}

\begin{proof}
Let $\mathcal B$ be the set of all mappings $\Omega\to (\mathcal N^{(n)}\cup\{*\})$.
Notice that $\mathcal B$ is finite.
Let us choose a finite subset $\mathcal A_{0}\subseteq \mathcal A$ that 
still satisfies condition (2). 
For every $L\in\mathbb N$ define a mapping 
$\xi_{L}\colon \mathcal A_{0}^{L}\to \mathcal B$ by 
 $$\xi_{L}(U_1,\dots,U_{L})(\sigma) = \begin{cases}
 h((U_1\circ\dots\circ U_{L})^{\sigma}),& \text{if $(U_1\circ\dots\circ U_{L})^{\sigma}\in \mathcal M$;}\\
 *& \text{otherwise.}
 \end{cases}$$
By the Hales-Jewett theorem, choose 
$L$ and 
$(T_1,\dots,T_{L})\in (\mathcal A_0\cup\{I\})^{L}\setminus \mathcal A_0^{L}$, 
where $I$ replaces the variable symbol.
Let us show that 
$M = T_1\circ\dots\circ T_{L}$ 
satisfies all required properties. 
By condition (3) we have $M\in\mathcal M$. It remains to prove that 
$h(M)^{\sigma_1} = h(M)^{\sigma_2}$ for any $(\sigma_1,\sigma_2)\in\omega$.
By condition (2) there exist
$A_1,A_2\in \mathcal A_0$ such that 
$A_1^{\sigma_1}=I^{\sigma_1}$ 
and 
$A_2^{\sigma_2} = I^{\sigma_2}$. 
To simplify the argument, assume that 
$T_1 = T_2 = I$ and $T_{i}\neq I$ for $i\ge 3$.
Then
\begin{align*}
h(M)^{\sigma_1} \overset{1}{=}
h(M^{\sigma_1}) \overset{2}{=}
&h((T_1\circ\dots\circ T_{L})^{\sigma_1}) \overset{3}{=} \\
&h(T_{1}^{\sigma_1}\circ\dots\circ T_{L}^{\sigma_1}) \overset{4}{=} \\
&h(A_1^{\sigma_1}\circ A_{1}^{\sigma_1}\circ
T_{3}^{\sigma_1}\circ\dots\circ T_{L}^{\sigma_1}) \overset{5}{=}\\ 
&h(A_{1}^{\sigma_2}\circ A_{1}^{\sigma_2}\circ
T_{3}^{\sigma_2}\circ\dots\circ T_{L}^{\sigma_2}) \overset{6}{=}\\ 
&\xi_{L}(A_1,A_1,T_3,\dots,T_L)(\sigma_2) \overset{7}{=} \\
&\xi_{L}(A_2,A_2,T_3,\dots,T_L)(\sigma_2) \overset{8}{=}\\
&h(A_{2}^{\sigma_2}\circ A_2^{\sigma_2}\circ
T_{3}^{\sigma_2}\circ\dots\circ T_{L}^{\sigma_2}) \overset{9}{=}\\ 
&h(T_{1}^{\sigma_2}\circ\dots\circ T_{L}^{\sigma_2}) \overset{10}{=} \\
&h((T_{1}\circ \dots\circ T_{L})^{\sigma_2}) \overset{11}{=} 
h(M^{\sigma_2}) \overset{12}{=}
h(M)^{\sigma_2},
\end{align*}
where the equalities 1 and 12 are by the definition of the minion homomorphism, 
2 and 11 are by the definition, 
3 and 10 are by the property of the $\circ$,
4 and 9 are by the choice of $A_1$ and $A_2$, 
5 is by condition (1), 
6 and 8 are by the definition of $\xi_{L}$ and the property of $\circ$, 
7 is by the choice of $L$ and $T_1,\dots,T_L$.
\end{proof}

\subsection{Characterization of the singleton arc-consistency algorithm}
\label{SubsectionSARC}

We characterize minions for the singleton versions
via matrices, following the ideas of Ciardo and \v Zivn\'y in \cite{ciardo2023clap}.
Let $\mathcal M_{\SArc}$ be the minion 
whose ${N}$-ary part consists of 
matrices $M\colon N\times[L]\to\{0,1\}$, where $L\in\mathbb N$, 
such that:
\begin{enumerate}
    \item[(1)] for every $j\in[L]$ there exists $i\in N$ with 
    $M(i,j) = 1$ (there exists 1 in every column);
    \item[(2)] if $M(i,j) = 1$ for some $i$ and $j$, then there exists $\ell\in[L]$ such that 
    $M(i,\ell)=1$ and 
    $M(i',\ell) = 0$ for all $i'\neq i$ (skeleton property).
\end{enumerate}

\begin{lem}\label{LEMSinglArcMinion}
Suppose 
$\mathbb A$ and $\mathbb B$ are finite relational structures in the same signature. 
Then the following conditions are equivalent 
for any instance $\mathbb X$ of $\PCSP(\mathbb A,\mathbb B)$:
\begin{itemize}
    \item $\CSingl\ArcConsistency$ returns ``Yes'' on $\mathbb X$; 
    \item 
$\mathbb X\to \mathbb F_{\mathcal M_{\SArc}}(\mathbb A)$.
\end{itemize}
\end{lem}
\begin{proof}
$\Rightarrow$ Assume that $\CSingl\ArcConsistency$ returns ``Yes'' on $\mathbb X$.
First, notice that 
finding a homomorphism $\mathbb X\to \mathbb F_{\mathcal M_{\SArc}}(\mathbb A)$
is equivalent to finding matrices from $\mathcal M_{\SArc}$ for each constraint and each variable such that they are compatible. 
Since the algorithm $\CSingl\ArcConsistency$ returns ``Yes'',
 for every constraint $C= R(x_{1},\dots,x_{s})$ 
 we obtain a nonempty relation $R_f$ (the final version of $R$ in the pseudocode).
Let 
$P_1,\dots, P_{L}$ be the list of all pairs $(C,\mathbf a)$, where
$C$ is a constraint and $\mathbf a$ is a tuple from the corresponding $R_f$. 
Then for every $i \in[L]$, 
the algorithm $\ArcCons$ returns a nonempty 
subset $D_{y}^{i}$ for each variable $y$. 
The matrix $M_{x}\colon A\times[L]\to \{0,1\}$ corresponding to 
the variable $x$ 
is 
defined by 
$M_{x}(a,i):= \begin{cases}
1,&  \text{if $a\in D_{x}^{i}$}\\
0,&  \text{otherwise}
\end{cases}$.
Similarly, the matrix $M_{C}\colon R^{\mathbb A}\times[L]\to \{0,1\}$ corresponding to a constraint 
$C= R(x_{1},\dots,x_{s})$ is defined by 
$$M_{C}(\mathbf a,i):= \begin{cases}
1,&  \text{if $\mathbf a^{j}\in D_{x_{j}}^{i}$ for every $j\in[s]$}\\
0,&  \text{otherwise}
\end{cases}.$$
Let us show that the matrices for constraints and for 
variables are consistent, 
that is, for any constraint $C= R(x_{1},\dots,x_{s})$ 
and $i\in[s]$, we have 
$M_{C}^{\pi_{i}} = M_{x}$, 
where $\pi_{i}\colon R^{\mathbb A}\to A$
is a projection onto the $i$-th coordinate.
Since the $i$-th columns of all matrices come from the output of the 
arc-consistency algorithm, 
the $i$-th columns are consistent. 
Taking a minor of a matrix object is equivalent to taking the corresponding minor of 
each column and then gluing the results, 
so the matrices themselves are  consistent.
Hence the assignment we defined yields a homomorphism $\mathbb X\to \mathbb F_{\mathcal M_{\SArc}}(\mathbb A)$.


$\Leftarrow$ Assume that $h\colon\mathbb X\to \mathbb F_{\mathcal M_{\SArc}}(\mathbb A)$ is a homomorphism.
First, notice that 
it is sufficient to prove the claim for 
connected instances $\mathbb X$, since the 
algorithm $\CSingl\ArcConsistency$ essentially works with disconnected parts 
independently. Thus, we assume that $\mathbb X$ is connected.
Then for any two variables $x_1,x_2\in X$,
the matrices 
$h(x_1)$ and $h(x_2)$ have the same width $L\in\mathbb N$.
For each constraint $C= R(x_{1},\dots,x_{s})$,  
let $R'$ be the set of all 
tuples $\mathbf a\in R$ such that 
the $\mathbf a$-th row of the matrix corresponding to $C$ has a non-zero entry.
Since the matrices in $\mathcal M_{\SArc}$ are skeletal, 
for each tuple $\mathbf a\in R'$ 
there exists $\ell_{C,\mathbf a}\in[L]$ 
such that the $\ell_{C,\mathbf a}$-th column of
the matrix corresponding to $C$
has 1 only in the $\mathbf a$-th position. 
Let $D_{y}^{C,\mathbf a}$ be the set of all $b\in D_{y}$ such that 
$h(y)(b,\ell_{C,\mathbf a})=1$. 
Since the sets $(D_{y}^{C,\mathbf a}\colon y\in X)$ form an arc-consistent reduction (see Observation \ref{OBSArcCons}), the arc-consistency algorithm must return ``Yes'' after we restrict 
the constraint $C$ to $\mathbf a$.
Hence, $\CSingl\ArcCons$ returns ``Yes''.
\end{proof}

\begin{cor}\label{CORMainForSinglArc}
Suppose $\mathbb A$ and $\mathbb B$ are finite relational structures in the same signature.
Then 
$\CSingl\ArcCons$ solves $\PCSP(\mathbb A,\mathbb B)$ if 
and only if there exists a minion homomorphism 
$\mathcal M_{\SArc}\to\Pol(\mathbb A,\mathbb B)$.
\end{cor}

\begin{proof}
This follows from  
Corollary \ref{CORMainCharacterizationOfAlgorithm} and Lemma \ref{LEMSinglArcMinion}.
\end{proof}

\begin{lem}\label{LEMWeakOperationsForSinglArc}
Suppose $n,\ell\in\mathbb N$, 
$\mathcal N$ is a minion whose $n$-ary part is finite, 
$h\colon \mathcal M_{SArc}\to \mathcal N$ is a minion homomorphism.
Then $h(\mathcal M_{SArc})$
contains an $(\underbrace{\overline{\ell},\dots,\overline{\ell}}_{n})$-palette totally symmetric object.
\end{lem}

\begin{proof}
Let us define everything we need in Theorem \ref{THMHalesJewettApplication}.
Let $N = n\cdot \ell$, with $[N]$ playing the role of $N$ and $[n]$ the role of $n$.
Let $\mathcal M = \mathcal M_{SArc}$, and let 
$\mathcal M_0$ be the set of all matrices satisfying all the properties of 
$\mathcal M_{\SArc}$ but property (2) (skeleton property).
By $\circ$ denote the concatenation of two matrices, which is a semigroup operation that respects taking minors.
For $i\in [n]$, define 
$\phi(i)$ to be the matrix $[N]\times [1]\to \{0,1\}$ whose only column 
is $(\underbrace{0,\dots,0}_{(i-1)\ell},\underbrace{1,\dots,1}_{\ell},\underbrace{0,\dots,0}_{(n-i)\ell})$.
Let
$$\mathcal A = 
\{\phi(a_1)\circ\dots\circ\phi(a_N)\mid a_1,\dots,a_N\in[n]\},$$
and $\mathcal I\in \mathcal M$ be the identity matrix of size 
$N\times N$.

A mapping $\pi\colon [N]\to[n]$ can be viewed as a tuple
of length $N$, so the notions of a palette tuple and of the $i$-th block can be extended naturally to such mappings.
Let $\Omega$ be the set of 
all $(\underbrace{\overline{\ell},\dots,\overline{\ell}}_{n})$-palette mappings
$[N]\to[n]$.
Let $\omega$ be the set of all pairs $(\sigma_1,\sigma_2)$ 
such that,
for every $i\in[n]$, the $i$-th block of $\sigma_1$ contains the same set of elements as the $i$-th block of $\sigma_2$.

Let us check the required conditions (1)-(3) of Theorem \ref{THMHalesJewettApplication}.
Condition (1) follows from the definition of $\phi$.
For condition (2), let $\sigma\in\Omega$.
For each $i$ appearing in $\sigma$ (viewed as a tuple), 
let $u(i)$ be the position of a block that contains only $i$
(such a block exists since $\sigma$ is palette).
Then $\phi(u(i))^{\sigma}$ has $1$ at the $i$-th position and $0$ elsewhere.
 Hence,
 $(\phi(u(\sigma(1)))\circ\dots\circ \phi(u(\sigma(N))))^{\sigma}=I^{\sigma}$.
Condition (3) holds because adding $I$ makes the matrix skeletal.
By Theorem \ref{THMHalesJewettApplication} 
we obtain $M\in\mathcal M_{\SArc}^{(N)}$ such that 
$h(M)^{\sigma_1} = h(M)^{\sigma_2}$ for any $(\sigma_1,\sigma_2)\in\omega$.
Thus, $h(M)$ is an  $(\underbrace{\overline{\ell},\dots,\overline{\ell}}_{n})$-palette totally symmetric object.
\end{proof}

\begin{lem}\label{LEMWeakForSinglArcImpliesAlgorithm}
Suppose $\mathbb A$ and $\mathbb B$ are relational structures in the same signature, $\mathbb A$ is finite, $\Pol(\mathbb A,\mathbb B)$
has $(\overline{\ell_1},\dots,\overline{\ell_n})$-palette totally symmetric 
functions for all $\ell_1,\dots,\ell_n\in\mathbb N$.
Then 
$\SinglArcCons$ solves $\PCSP(\mathbb A,\mathbb B)$.
\end{lem}

\begin{proof}
Notice that by identifying variables within blocks of an 
$(L_1,\dots,L_n)$-palette totally symmetric (PTS) function
we obtain a PTS function with smaller block sizes.
Therefore, for all 
$L_1,\dots,L_n\in \mathbb N$
we can choose 
an $(L_1,\dots,L_n)$-PTS function
$f_{L_1,\dots,L_n}\in\Pol(\mathbb A,\mathbb B)$ such that 
identifying variables within blocks of
$f_{L_1,\dots,L_n}$ yields 
$f_{\ell_1,\dots,\ell_n}$, where $\ell_i\le L_i$ for every $i$.
To simplify notation, 
we write 
$f(S_1,\dots,S_n)$ meaning 
$f_{k_1,\dots,k_n}(s_1^1,\dots,s_{1}^{k_1},\dots,
s_n^1,\dots,s_{n}^{k_n})$,
where $S_{i} = \{s_i^{1},\dots,s_{i}^{k_i}\}$ for every $i\in[n]$.
This notation is justified because we only 
apply $f_{k_1,\dots,k_n}$ to palette tuples, so the order of elements in the sets $S_{i}$ does not matter. 

We need to check that if 
$\SinglArcCons$ returns ``Yes'' on an instance $\mathcal I$, then the instance has a solution in $\mathbb B$.
Let $P_1,\dots,P_S$ be all pairs $(x,a)$, where $x$ is a variable of $\mathcal I$ and $a\in D_{x}$ at the end of the algorithm. 
For each $P_j = (x,a)$ and each variable $y$, let 
$D_{y}^{j}$ be 
the domain of $y$ after 
restricting $x$ to $\{a\}$ and running arc-consistency.
Let us show that a solution to $\mathcal I^{\mathbb B}$ can be defined by
$s(x):= f(D_{x}^{1},\dots,D_{x}^{S})$
for every variable $x$. 
Let $C=R(x_{1},\dots,x_{k})$ be a constraint $\mathcal I$.
For $j\in[S]$ set $R_j = R^{\mathbb A}\cap (D^{j}_{x_1}\times\dots\times D^{j}_{x_k})$.
By the assumptions on $f$, we have
$f(R_1,\dots,R_S) = (s(x_{1}),\dots,s(x_{k}))$.
Since $f$ is a polymorphism, it follows that
$(s(x_{1}),\dots,s(x_{k}))\in R^{\mathbb B}$. 
Thus, $s$ is a solution of $\mathcal I$ in $\mathbb B$, completing the proof.
\end{proof}


\begin{THMSinglArcCharacterizationTHM}
Let $(\mathbb A, \mathbb B)$ be a PCSP template. The following conditions are equivalent:
\begin{enumerate}
\item[(1)] $\SinglArcCons$ solves $\PCSP(\mathbb A,\mathbb B)$;
\item[(2)] $\CSingl\ArcCons$ solves $\PCSP(\mathbb A,\mathbb B)$;
\item[(3)] $\Pol(\mathbb A,\mathbb B)$ has an $(\overline{m_1},\dots,\overline{m_n})$-palette totally symmetric 
function for every $n, m_1,\dots,m_n\in\mathbb N$;
\item[(4)] for every $N\in\mathbb N$ there exist 
$n,m_1,\dots,m_n\ge N$ such that 
$\Pol(\mathbb A,\mathbb B)$ has an $(\overline{m_1},\dots,\overline{m_n})$-palette totally symmetric 
function.
\end{enumerate}
\end{THMSinglArcCharacterizationTHM}


\begin{proof}
(2)$\Rightarrow$(4). By Corollary \ref{CORMainForSinglArc}
there exists a minion homomorphism 
$\mathcal M_{\SArc}\to\Pol(\mathbb A,\mathbb B)$.
Since $\Pol(\mathbb A,\mathbb B)^{(n)}$ is finite, 
Lemma \ref{LEMWeakOperationsForSinglArc} implies that 
$\Pol(\mathbb A,\mathbb B)$ has an $(\underbrace{\overline{\ell},\dots,\overline{\ell}}_{n})$-palette totally symmetric function for all $\ell$ and $n$.

(4)$\Rightarrow$(3) follows from the fact that identifying variables
inside blocks of a palette totally symmetric function
we can reduce blocks to any desired size, and identifying two blocks of the same size we can reduce the number of blocks.

(3)$\Rightarrow$(1) follows from Lemma \ref{LEMWeakForSinglArcImpliesAlgorithm}.

(1)$\Rightarrow$(2) holds because the algorithm $\CSingl\ArcCons$ is stronger than $\SinglArcCons$.
\end{proof}

\subsection{Characterization of the singleton BLP+AIP}
\label{SUBSECTIONSinglBLPAIPCharacterization}

Let $\mathcal M_{\SBLPAIP}$ be the minion 
whose ${N}$-ary part consists of 
pairs of matrices 
$M_{\BLP}\colon N\times[L]\to\mathbb Q\cap[0,1]$
and  
$M_{\AIP}\colon N\times[L]\to\mathbb Z$, where $L\in\mathbb N$, 
such that:
\begin{enumerate}
    \item[(1)] 
    $\sum\limits_{i\in N} M_{\BLP}(i,j) = 1$ and $\sum\limits_{i\in N} M_{\AIP}(i,j) = 1$ for every $j\in[L]$;
    \item[(2)] if $M_{\AIP}(i,j)\neq 0$ for some $i,j$, then 
    $M_{\BLP}(i,j)>0$;
    \item[(3)] if $M_{\BLP}(i,j)>0$ for some $i$ and $j$, then there exists $\ell\in[L]$ such that 
    $M_{\BLP} (i,\ell)=M_{\AIP} (i,\ell)=1$ and 
    $M_{\BLP} (i',\ell)=M_{\AIP} (i',\ell) = 0$ for all $i'\neq i$ (skeleton property).
\end{enumerate}

\begin{lem}\label{LEMSinglBLPAIPMinion}
Suppose 
$\mathbb A$ and $\mathbb B$ are finite relational structures in the same signature. 
Then the following conditions are equivalent 
for any instance $\mathbb X$ of $\PCSP(\mathbb A,\mathbb B)$:
\begin{itemize}
    \item $\CSingl(\BLP+\AIP)$ returns ``Yes'' on $\mathbb X$; 
    \item 
$\mathbb X\to \mathbb F_{\mathcal M_{\SBLPAIP}}(\mathbb A)$.
\end{itemize}
\end{lem}
\begin{proof}
$\Rightarrow$ Assume that $\CSingl(\BLP+\AIP)$ returns ``Yes'' on $\mathbb X$.
First, notice that 
finding a homomorphism $\mathbb X\to \mathbb F_{\mathcal M_{\SBLPAIP}}(\mathbb A)$
is equivalent to finding pairs of matrices from $\mathcal M_{\SBLPAIP}$ for each constraint and each variable so that they are compatible. 
Since the algorithm $\CSingl(\BLP+\AIP)$ returns ``Yes'',
 for every constraint $C= R(x_{1},\dots,x_{s})$ 
 we obtain a nonempty relation $R_f$ (the final version of $R$ in the pseudocode).
Let 
$P_1,\dots, P_{L}$ be the list of all pairs $(C,\mathbf a)$ with
$C$ a constraint and $\mathbf a$ a tuple from the corresponding $R_f$. 
For every $i \in[L]$ 
the algorithm $\BLP+\AIP$ produces 
a BLP solution $s_{\BLP}^{i}$ and an AIP solution $s_{\AIP}^{i}$. 
Then, the BLP and AIP matrices corresponding to 
a variable $x$ 
can be defined by 
$M^{\BLP}_{x}(a,i) = s_{\BLP}^{i}(x^a)$
and 
$M^{\AIP}_{x}(a,i) = s_{\AIP}^{i}(x^a)$.
Similarly, the BLP and AIP matrices corresponding to 
a constraint $C$ are
$M^{\BLP}_{C}(\mathbf a,i) = s_{\BLP}^{i}(C^{\mathbf a})$
and 
$M^{\AIP}_{C}(\mathbf a,i) = s_{\AIP}^{i}(C^{\mathbf a})$.
By definition of $\BLP+\AIP$, this assignment 
yields a homomorphism $\mathbb X\to \mathbb F_{\mathcal M_{\SBLPAIP}}(\mathbb A)$.

$\Leftarrow$ Assume that $h\colon\mathbb X\to \mathbb F_{\mathcal M_{\SBLPAIP}}(\mathbb A)$ is a homomorphism.
First, notice that 
it is sufficient to prove the claim for 
connected instances $\mathbb X$ because the 
algorithm $\CSingl(\BLP+\AIP)$ essentially works with disconnected parts 
independently. Thus, we assume that $\mathbb X$ is connected.
Then for any two variables $x_1,x_2\in X$ 
the matrices 
in $h(x_1)$ and $h(x_2)$ have the same width $L\in\mathbb N$.
For each constraint $C= R(x_{1},\dots,x_{s})$,  
let $R'$ be the set of all 
tuples $\mathbf a\in R$ such that 
the $\mathbf a$-th row of the BLP matrix corresponding to $C$ has a non-zero entry.
Since the matrices in $\mathcal M_{\SBLPAIP}$ are skeletal, 
for each tuple $\mathbf a\in R'$ 
there exists $\ell_{C,\mathbf a}\in[L]$ 
such that the $\ell_{C,\mathbf a}$-th column of
the BLP matrix corresponding to $C$
has 1 only in the $\mathbf a$-th position. 
Then BLP and AIP solutions that $\BLP+\AIP$ returns after fixing 
$C$ to $\mathbf a$ can be defined by
$s_{\BLP}(C_0^{\mathbf a_0}) =  M^{\BLP}_{C_0}(\mathbf a_0,\ell_{C,\mathbf a})$,
$s_{\AIP}(C_0^{\mathbf a_0}) =  M^{\AIP}_{C_0}(\mathbf a_0,\ell_{C,\mathbf a})$,
$s_{\BLP}(y^{b}) =  M^{\BLP}_{y}(b,\ell_{C,\mathbf a})$,
$s_{\AIP}(y^{b}) =  M^{\AIP}_{y}(b,\ell_{C,\mathbf a})$
for all constraints $C_0$, their tuples $\mathbf a_0$, 
variables $y$, and $b\in D_{y}$.
Hence, $\CSingl(\BLP+\AIP)$ returns ``Yes'' on $\mathbb X$.
\end{proof}

\begin{cor}\label{CORMainForSinglBLPAIP}
Suppose $\mathbb A$, $\mathbb B$ are finite relational structures in the same signature.
Then 
$\CSingl(\BLP+\AIP)$ solves $\PCSP(\mathbb A,\mathbb B)$ if 
and only if there exists a minion homomorphism 
$\mathcal M_{\SBLPAIP}\to\Pol(\mathbb A,\mathbb B)$.
\end{cor}

\begin{proof}
This follows from 
Corollary \ref{CORMainCharacterizationOfAlgorithm} and Lemma \ref{LEMSinglBLPAIPMinion}.
\end{proof}

\begin{lem}\label{LEMWeakOperationsForSinglBLPAIP}
Suppose $n,\ell\in\mathbb N$, 
$\mathcal N$ is a minion whose $n$-ary part is finite, 
$h\colon \mathcal M_{\SBLPAIP}\to \mathcal N$ is a minion homomorphism.
Then $h(\mathcal M_{\SBLPAIP})$
contains an $(\underbrace{\overline{\ell+1,\ell},\overline{\ell+1,\ell},\dots,\overline{\ell+1,\ell}}_{2n})$-palette symmetric object.
\end{lem}

\begin{proof}
We define all components needed to apply Theorem \ref{THMHalesJewettApplication}.
Let $N = n\cdot (2\ell+1)$, $[N]$ play the role of $N$, and $[n]$ play the role of $n$.
Let $\mathcal M = \mathcal M_{\SBLPAIP}$ and 
$\mathcal M_0$ be the set of all pairs of matrices satisfying all properties of $\mathcal M_{\SBLPAIP}$ except property (3) (skeleton property).
Define $\circ$ as the concatenation of matrices, applied independently for BLP and AIP matrices.
This is a semigroup operation that respects taking minors.
For $i\in [n]$ 
define a pair of matrices $\phi(i)$.
Its BLP matrix $[N]\times [1]\to \mathbb Q\cap [0,1]$ 
has one column 
$(\underbrace{0,\dots,0}_{(i-1)(2\ell+1)},\underbrace{\frac{1}{2\ell+1},\dots,\frac{1}{2\ell+1}}_{2\ell+1},\underbrace{0,\dots,0}_{(n-i)(2\ell+1)})$.
Its AIP matrix $[N]\times [1]\to \mathbb Z$ 
has one column 
$(\underbrace{0,\dots,0}_{(i-1)(2\ell+1)},\underbrace{1,\dots,1}_{\ell+1},\underbrace{-1,\dots,-1}_{\ell},\underbrace{0,\dots,0}_{(n-i)(2\ell+1)})$.
Let 
$\mathcal A = 
\{\phi(a_1)\circ\dots\circ\phi(a_N)\mid a_1,\dots,a_N\in[n]\},$ and let
$\mathcal I\in \mathcal M$ be the pair of identity matrix of size 
$N\times N$.

Mappings $\pi\colon [N]\to[n]$ can be viewed as tuples
of length $N$, 
and the notions of a palette tuple and the $i$-th block can be extended naturally to such mappings.
Let $\Omega$ be the set of 
all $(\underbrace{\overline{\ell+1,\ell},\overline{\ell+1,\ell},\dots,\overline{\ell+1,\ell}}_{2n})$-palette mappings
$[N]\to[n]$.
Let $\omega$ be the set of all pairs $(\sigma_1,\sigma_2)$ 
such that
for every $i\in[2n]$ the $i$-th blocks of $\sigma_1$ and $\sigma_2$ 
differ only by a permutation of their elements.

Let us check conditions (1)-(3) of Theorem \ref{THMHalesJewettApplication}.
Condition (1) follows from the definition of $\phi$.
Let us prove that condition (2) is satisfied 
for every $\sigma\in\Omega$.
For each $i$ appearing in $\sigma$ (viewed as a tuple) 
let $u(i)\in[n]$ be the position of an overlined block containing only $i$
(such a block exists since $\sigma$ is palette).
Then both matrices of $\phi(u(i))^{\sigma}$ have $1$ at the $i$-th position and $0$ elsewhere.
 Hence,
 $(\phi(u(\sigma(1)))\circ\dots\circ \phi(u(\sigma(N))))^{\sigma}=I^{\sigma}$.
Condition (3) holds because adding $I$ makes the matrices skeletal.
By Theorem \ref{THMHalesJewettApplication} 
we can find $M\in\mathcal M_{\SBLPAIP}^{(N)}$ such that 
$h(M)^{\sigma_1} = h(M)^{\sigma_2}$ for all $(\sigma_1,\sigma_2)\in\omega$.
Hence, $h(M)$ is an  $(\underbrace{\overline{\ell+1,\ell},\overline{\ell+1,\ell},\dots,\overline{\ell+1,\ell}}_{2n})$-palette symmetric object.
\end{proof}

\begin{lem}\label{LEMWeakForSinglBLPAIPImpliesAlgorithm}
Suppose  
for every $N\in\mathbb N$ there exist 
$n\ge N$, 
$k_1,\dots,k_n\in\mathbb N$, and $m_{1,1},\dots,m_{1,k_1},
\dots$,
$m_{n,1},\dots,m_{n,k_n}\ge N$ such that $\Pol(\mathbb A,\mathbb B)$ contains an 
$(\overline{m_{1,1},\dots,m_{1,k_1}},
\overline{m_{2,1},\dots,m_{2,k_2}},\dots,\overline{m_{n,1},\dots,m_{n,k_n}})$-palette symmetric function.
Then 
$\Singl(\BLP+\AIP)$ solves $\PCSP(\mathbb A,\mathbb B)$.
\end{lem}

\begin{proof}
We need to check that if 
$\Singl(\BLP+\AIP)$ returns ``Yes'' on an instance $\mathcal I$, then the instance has a solution in $\mathbb B$.
Let $P_1,\dots,P_S$ be all pairs $(x,a)$ where $x$ is a variable of $\mathcal I$ and $a\in D_{x}$ at the end of the algorithm. 
For each $P_j = (x,a)$ 
there exist BLP and AIP solutions 
to the instance obtained by 
restricting 
the variable $x$ to $\{a\}$.
Then the argument we used 
after the Theorem \ref{thmBLPAIPCharacterization}
implies that for some $N>S$,  
there exists a $d$-solution $s_{d}^{j}$ to this instance for every $d\ge N$.

Choose 
an $(\overline{m_{1,1},\dots,m_{1,k_1}},
\overline{m_{2,1},\dots,m_{2,k_2}},\dots,\overline{m_{n,1},\dots,m_{n,k_n}})$-palette symmetric function $f\in\Pol(\mathbb A,\mathbb B)$, where
$n\ge N$, 
$k_1,\dots,k_n\in\mathbb N$, and $m_{1,1},\dots,m_{1,k_1},
\dots,
m_{n,1},\dots,m_{n,k_n}\ge N$.
For $j\in[S]$ and $i\in [k_{j}]$, let $\alpha_{x}^{j,i}$ be a tuple 
of length $m_{j,i}$, where each element $a\in A$
appears $s_{m_{j,i}}^{j}(x^{a})$ times.
For $j\in\{S+1,S+2,\dots,n\}$ and $i\in[k_{j}]$,
let $\alpha_{x}^{j,i}$ be a tuple 
of length $m_{j,i}$ where each element $a\in A$
appears  
$s_{m_{j,i}}^{1}(x^{a})$ times.
Let us show that a solution to $\mathcal I^{\mathbb B}$ can be defined by
$s(x):= f(\alpha_{x}^{1,1},\dots, \alpha_{x}^{1,k_1},
\dots,\alpha_{x}^{n,1},\dots, \alpha_{x}^{n,k_n})$
for every variable $x$. 
Let $C=R(x_{1},\dots,x_{k})$ be a constraint of $\mathcal I$.
For $j\in [S]$ and $i\in [k_{j}]$ let $M_{j,i}$ be the 
matrix whose columns are tuples from $R$ where each 
tuple $\mathbf a$ appears 
$s_{m_{j,i}}^{j}(C^{\mathbf a})$ times.
Similarly, 
for $j\in \{S+1,S+2,\dots,n\}$ and $i\in [k_{j}]$ let $M_{j,i}$ be the 
matrix whose columns are tuples from $R$ where each 
tuple $\mathbf a$ appears 
$s_{m_{j,i}}^{1}(C^{\mathbf a})$ times.
Our assumptions about $f$ imply that 
$f(M_{1,1},\dots,M_{1,k_1},\dots,M_{n,1},\dots,M_{n,k_n}) = (s(x_{1}),\dots,s(x_{k}))$.
Since $f\in\Pol(\mathbb A,\mathbb B)$, it follows that 
$(s(x_{1}),\dots,s(x_{k}))\in R^{\mathbb B}$.
Thus, $s$ is a solution of $\mathcal I$ in $\mathbb B$.
\end{proof}


\begin{THMSinglBLPAIPTHM}
Let $(\mathbb A, \mathbb B)$ be a PCSP template.
The following conditions are equivalent:
\begin{enumerate}
\item[(1)] $\Singl(\BLP+\AIP)$ solves $\PCSP(\mathbb A,\mathbb B)$;
\item[(2)] $\CSingl(\BLP+\AIP)$ solves $\PCSP(\mathbb A,\mathbb B)$;
\item[(3)] $\Pol(\mathbb A,\mathbb B)$ contains an $(\underbrace{\overline{\ell+1,\ell},\overline{\ell+1,\ell},\dots,\overline{\ell+1,\ell}}_{2n})$-palette symmetric function for every $\ell,n\in\mathbb N$;
\item[(4)] for every $N\in\mathbb N$ there exist 
$n\ge N$, 
$k_1,\dots,k_n\in\mathbb N$, and $m_{1,1},\dots,m_{1,k_1},
\dots,
m_{n,1},\dots,m_{n,k_n}\ge N$ such that $\Pol(\mathbb A,\mathbb B)$ contains an 
$(\overline{m_{1,1},\dots,m_{1,k_1}},
\overline{m_{2,1},\dots,m_{2,k_2}},\dots,\overline{m_{n,1},\dots,m_{n,k_n}})$-palette symmetric function.
\end{enumerate}
\end{THMSinglBLPAIPTHM}



\begin{proof}
(2)$\Rightarrow$(3). By Corollary \ref{CORMainForSinglBLPAIP}
there exists a minion homomorphism 
$\mathcal M_{\SBLPAIP}\to\Pol(\mathbb A,\mathbb B)$.
Since $\Pol(\mathbb A,\mathbb B)^{(n)}$ is finite, 
Lemma \ref{LEMWeakOperationsForSinglBLPAIP} implies that 
$\Pol(\mathbb A,\mathbb B)$ has an 
$(\underbrace{\overline{\ell+1,\ell},\overline{\ell+1,\ell},\dots,\overline{\ell+1,\ell}}_{2n})$-palette symmetric function for all $\ell$ and $n$.

(3)$\Rightarrow$(4) is immediate.

(4)$\Rightarrow$(1) follows from Lemma \ref{LEMWeakForSinglBLPAIPImpliesAlgorithm}.

(1)$\Rightarrow$(2) holds because $\CSingl(\BLP+\AIP)$ is stronger than $\Singl(\BLP+\AIP)$.
\end{proof}

\section{Temporal constraint satisfaction problem}\label{SECTIONTemporalCSP}

As we show in this section, palette functions can also be useful for solving CSPs over infinite domain.
Notice that such CSPs can be very complicated. 
In fact, 
as was shown in \cite{bodirskyInfiniteHell}, 
every computational problem is
equivalent (under polynomial-time Turing reductions) to a problem of the form $\CSP(\mathbb A)$ for some infinite 
structure $\mathbb A$.
To bring the problem back to the class NP, 
mathematicians usually require the constraint language to satisfy additional assumptions 
\cite{InfiniteDomainSurvey,barto2016algebraic}.

\subsection{Tractable temporal CSPs}

One of the main results on the complexity of CSPs over infinite domain is the following classification of temporal CSPs.
A relational structure 
$\mathbb A$ is \emph{temporal} if its domain is $\mathbb Q$ and 
its relations are 
definable by Boolean combinations of atomic formulas of the form $x < y$. 

\begin{thm}\cite{bodirskyTemporalCSP}
Suppose $\mathbb A$ is a temporal relational structure.
Then $\CSP(\Gamma)$ is solvable in polynomial time, or NP-complete.
\end{thm}

Even though this classification has been known for many years, 
only recently Mottet \cite{mottet2025TemporalReduction} discovered a new uniform algorithm for the tractable cases, based on the following theorem.
For a temporal relational structure $\mathbb A = (\mathbb Q;R_1,\dots,R_s)$ we denote  by 
$\mathbb{A}|_{[N]}$ 
its restriction to the set $[N]$, that is, 
the structure 
$([N]; R_{1}\cap [N]^{\arity(R_1)},
\dots, R_{s}\cap [N]^{\arity(R_s)})$.


\begin{thm}[Theorem 3 in \cite{mottet2025TemporalReduction}]\label{THMTemporalMottet}
Let $\mathbb A$ be a temporal structure, $N\in\mathbb N$.
Then one of the following holds:
\begin{itemize}
\item $\PCSP(\mathbb A|_{[N]}, \mathbb A)$ is solvable by local consistency or $\Singl\AIP$,
\item  $\CSP(\mathbb A)$ is NP-hard.
\end{itemize}
\end{thm}

This gives a very simple algorithm for 
$\CSP(\mathbb A)$ whenever it is tractable. 
In fact, 
for any instance $\mathbb X$ of $\CSP(\mathbb A)$
with $|X| = N$,
the existence of a homomorphism 
$\mathbb X\to \mathbb A|_{[N]}$ 
is equivalent to the existence of 
a homomorphism $\mathbb X\to \mathbb A$.
The idea of going from a temporal CSP instance to a finite one by restricting its domain first appeared in \cite{SamplingIdeaForTemporal} under the name ``sampling''.
Hence, by Theorem \ref{THMTemporalMottet}, it is sufficient to 
run the consistency algorithm and 
$\Singl\AIP$ on $\mathbb X$ as an instance of 
$\PCSP(\mathbb A|_{[N]}, \mathbb A)$.

The proof in \cite{mottet2025TemporalReduction}
still uses a lot of knowledge about temporal CSPs.
We show that the soundness 
of the corresponding algorithms can be witnessed by 
concrete (palette) symmetric functions 
from $\Pol(\mathbb A|_{[N]}, \mathbb A)$.

Let us define several relational structures that 
are the most expressible among all temporal structures whose CSP is tractable.
Put 
\begin{align*}
\mathbb D &= (\mathbb Q;x\neq y, y\le x\vee z<x, y\neq x\vee z\le x),\\
\mathbb E &= (\mathbb Q;x=y<z\vee x=z<y\vee y=z<x),\\
\mathbb F &= (\mathbb Q;x\neq y, x=y\Rightarrow u=v,x>y\vee x>z\vee x=y=z).
\end{align*}

Thus, $\mathbb D$ has one binary and two ternary relations, and 
$\mathbb E$ has only one ternary.
For a relation $R$ by $R^{dual}$ we denote the relation 
$\{(a_1,\dots,a_{\arity(R)})\mid (-a_1,\dots,-a_{\arity(R)})\in R\}$.
For a temporal structure 
$\mathbb A$ by $\mathbb A^{dual}$ we denote the structure
obtained by replacing each relation $R$ with its dual relation $R^{dual}$.
We have the following classification of tractable temporal CSPs:

\begin{thm}[\cite{bodirsky2021Book}]\label{THMTemporalClassification}
Suppose $\mathbb A$ is a temporal relational structure such that 
$\CSP(\mathbb A)$ is not NP-hard. Then one of the following holds:
\begin{itemize}
    \item[(c)] $\Pol(\mathbb A)$ contains a constant operation;
    \item[(min)] $\min\in \Pol(\mathbb A)$ or $\max\in \Pol(\mathbb A)$;
    \item[(mi)] $\mathbb D$ or $\mathbb D^{dual}$ pp-defines $\mathbb A$;
    \item[(mx)] $\mathbb E$ or $\mathbb E^{dual}$ pp-defines $\mathbb A$;
    \item[(ll)] $\mathbb F$  or $\mathbb F^{dual}$  pp-defines $\mathbb A$.
\end{itemize}
\end{thm}

\subsection{Palette functions}

In this subsection we build concrete palette polymorphisms 
for $\Pol(\mathbb D|_{[N]},\mathbb D)$
and $\Pol(\mathbb E|_{[N]},\mathbb E)$, 
and prove that 
$\Pol(\mathbb F|_{[N]},\mathbb F)$ does not contain one.

\begin{lem}\label{LEMDPolymorphisms}
    $\Pol(\mathbb D|_{[N]},\mathbb D)$ contains 
    an $(\overline{\ell_1},\dots,\overline{\ell_n})$-palette totally symmetric function
    for all $n,N,\ell_1,\dots$, $\ell_n\in \mathbb N$.
\end{lem}

\begin{proof}
Put $M = \ell_1+\dots+\ell_n$. 
For $\mathbf a\in \mathbb Q^{M}$ by 
$\phi(\mathbf a)$ denote the tuple 
$\mathbf b\in \{0,1\}^{M}$ such that 
$\mathbf b^{i} = \begin{cases}
    0, & \text{if $\mathbf a^{i} = \min_{j} \mathbf a^j$}\\
    1, & \text{otherwise}    
\end{cases}.$
Choose any 
injective monotone mapping 
$h\colon \{0,1\}^{M}\to \mathbb Q\cap (0,1)$.
For a block tuple 
$(\mathbf a_1,\dots,\mathbf a_n)$ by $\mathcal{F}(\mathbf a_1,\dots,\mathbf a_n)$
we denote the minimal $k\in[n]$ such that 
$\mathbf a_{k} = (\min_{i,j}\mathbf a_{i}^{j},\dots,\min_{i,j}\mathbf a_{i}^{j})$.
If no such $k$ exists, 
set $\mathcal{F}(\mathbf a_1,\dots,\mathbf a_n)=\infty$.
Now define an $(\overline{\ell_1},\dots,\overline{\ell_n})$-palette totally symmetric function by 
$$f(\mathbf a_1,\dots,\mathbf a_n) = 
\begin{cases}
2\min_{i,j}\mathbf a_{i}^{j}+h(\phi(\mathbf a_1,\dots,\mathbf a_n))+1,  & \text{if $\mathcal{F}(\mathbf a_1,\dots,\mathbf a_n)=\infty$}\\
2\min_{i,j}\mathbf a_{i}^{j}+\frac{\mathcal{F}(\mathbf a_1,\dots,\mathbf a_n)}{n+1},  & \text{otherwise}
\end{cases}.$$

Notice that any $(\overline{\ell_1},\dots,\overline{\ell_n})$-palette 
tuple $(\mathbf a_1,\dots,\mathbf a_n)$ contains a block $\mathbf a_{i}$ consisting only of the minimal element. Hence 
$\mathcal F(\mathbf a_1,\dots,\mathbf a_n)<\infty$, which implies that 
$f$ is $(\overline{\ell_1},\dots,\overline{\ell_n})$-palette totally symmetric.

Let us show that $f\in\Pol(\mathbb D|_{[N]},\mathbb D)$:

$x\neq y$. Consider two block tuples 
$\mathbf a = (\mathbf a_1,\dots,\mathbf a_n)$ and $\mathbf b = (\mathbf b_1,\dots,\mathbf b_n)$
such that 
$\mathbf a_{i}^{j}\neq \mathbf b_{i}^{j}$ for all $i,j$.
We need to show that 
$f(\mathbf a)\neq f(\mathbf b)$.
If $\min_{i,j}\mathbf a_{i}^{j}\neq \min_{i,j}\mathbf b_{i}^{j}$
or
$\mathcal{F}(\mathbf a)\neq\mathcal{F}(\mathbf b)$, then it follows immediately from the definition of $f$.
If $\mathcal{F}(\mathbf a)=\mathcal{F}(\mathbf b)=\infty$, then 
we use injectivity of $h$ and the fact that 
$\phi(\mathbf a)\neq\phi(\mathbf b)$.
It remains to notice that 
the case $\mathcal{F}(\mathbf a)=\mathcal{F}(\mathbf b)<\infty$ cannot happen.

$y\le x\vee z<x$.
Consider three block tuples 
$\mathbf a = (\mathbf a_1,\dots,\mathbf a_n)$, $\mathbf b = (\mathbf b_1,\dots,\mathbf b_n)$, 
and $\mathbf c = (\mathbf c_1,\dots,\mathbf c_n)$
such that 
$\mathbf b_{i}^{j}\le \mathbf a_{i}^{j}\vee \mathbf c_{i}^{j}<\mathbf a_{i}^{j}$
for all $i,j$.
We need to prove that 
$f(\mathbf b)\le f(\mathbf a)\vee f(\mathbf c)<f(\mathbf a)$.
If $\mathcal{F}(\mathbf a)<\infty$,
then $\min_{i,j}\mathbf b_{i}^{j}< \min_{i,j}\mathbf a_{i}^{j}$, 
or $\min_{i,j}\mathbf c_{i}^{j}< \min_{i,j}\mathbf a_{i}^{j}$,
or $\mathcal{F}(\mathbf b)\le\mathcal{F}(\mathbf a)$ and $\min_{i,j}\mathbf b_{i}^{j}= \min_{i,j}\mathbf a_{i}^{j}$.
In each case the required condition follows.  
If $\mathcal{F}(\mathbf a)=\infty$, 
then $\min_{i,j}\mathbf b_{i}^{j}< \min_{i,j}\mathbf a_{i}^{j}$, or $\min_{i,j}\mathbf c_{i}^{j}< \min_{i,j}\mathbf a_{i}^{j}$,
or $\phi(\mathbf b)\le \phi(\mathbf a)$ and $\min_{i,j}\mathbf b_{i}^{j}= \min_{i,j}\mathbf a_{i}^{j}$. 
It remains to use the fact that $h$ is monotone.

$y\neq x\vee z\le x$.
Consider three block tuples 
$\mathbf a = (\mathbf a_1,\dots,\mathbf a_n)$, $\mathbf b = (\mathbf b_1,\dots,\mathbf b_n)$, 
and $\mathbf c = (\mathbf c_1,\dots,\mathbf c_n)$
such that 
$\mathbf b_{i}^{j}\neq \mathbf a_{i}^{j}\vee \mathbf c_{i}^{j}\le\mathbf a_{i}^{j}$
for all $i,j$.
We need to prove that 
$f(\mathbf b)\neq f(\mathbf a)\vee f(\mathbf c)\le f(\mathbf a)$.
If $\min_{i,j}\mathbf c_{i}^{j}< \min_{i,j}\mathbf a_{i}^{j}$, then 
$f(\mathbf c)< f(\mathbf a)$, so we are done.
Thus, we assume that $\min_{i,j}\mathbf c_{i}^{j}\ge \min_{i,j}\mathbf a_{i}^{j}$.
Additionally we assume that $f(\mathbf a) = f(\mathbf b)$.
This implies that $\min_{i,j}\mathbf a_{i}^{j}= \min_{i,j}\mathbf b_{i}^{j}$
and $\mathcal{F}(\mathbf a)=\mathcal{F}(\mathbf b)$.
If $\mathcal{F}(\mathbf a)=\mathcal{F}(\mathbf b)<\infty$, then
$\min_{i,j}\mathbf c_{i}^{j}= \min_{i,j}\mathbf a_{i}^{j}$ and $\mathcal{F}(\mathbf c)\le\mathcal{F}(\mathbf a)$, which implies the required condition.
Otherwise, we have $\phi(\mathbf a) = \phi(\mathbf b)$. Therefore 
$\min_{i,j}\mathbf c_{i}^{j}= \min_{i,j}\mathbf a_{i}^{j}$, 
and    $\mathcal{F}(\mathbf c)<\infty$ or 
$\phi(\mathbf c) \le \phi(\mathbf a)$.
Since $h$ is monotone, we derive 
$f(\mathbf c)\le f(\mathbf a)$.
\end{proof}

Using a compactness argument 
we can prove the following corollary.

\begin{cor}\label{CORDSinglArc}
    $\Pol(\mathbb D)$ contains 
    an $(\overline{\ell_1},\dots,\overline{\ell_n})$-palette totally symmetric function
    for all $n,\ell_1,\dots,\ell_n\in \mathbb N$.
\end{cor}

\begin{lem}
    Suppose 
    $\mathbb D$ pp-defines $\mathbb A$. Then 
    $\PCSP(\mathbb A|_{[N]},\mathbb A)$ can be solved by 
$\Singl\ArcCons$  for all $N\in\mathbb N$.
\end{lem}
\begin{proof}
By Corollary \ref{CORDSinglArc},
$\Pol(\mathbb D)$ contains 
an $(\overline{\ell_1},\dots,\overline{\ell_n})$-palette totally symmetric function
    for all $n,\ell_1,\dots,\ell_n\in \mathbb N$.
  Therefore, $\Pol(\mathbb A)$ and $\Pol(\mathbb A|_{[N]},\mathbb A)$ also contain them.
  By Lemma \ref{LEMWeakForSinglArcImpliesAlgorithm},
 $\PCSP(\mathbb A|_{[N]},\mathbb A)$ can be solved by 
$\Singl\ArcCons$.
\end{proof}

\begin{lem}\label{LEMEPolymorphisms}
    $\Pol(\mathbb E|_{[N]},\mathbb E)$ contains 
    a $(\overline{2\ell_1+1},\dots,\overline{2\ell_n+1})$-palette alternating function
    for all $n,N,\ell_1$, $\dots$, $\ell_n\in \mathbb N$.
\end{lem}

\begin{proof}
Put $M = (2\ell_1+1) +\dots + (2\ell_n+1)$.
For a block tuple 
$(\mathbf a_1,\dots,\mathbf a_n)$ by $\mathcal{S}(\mathbf a_1,\dots,\mathbf a_n)$
we denote the minimal $k\in[n]$ such that 
$\mathbf a_{k}$ contains an odd number of elements equal to $\min_{i,j}\mathbf a_{i}^{j}$.
If no such $k$ exists, 
set $\mathcal{S}(\mathbf a_1,\dots,\mathbf a_n)=\infty$.
By $\mathcal P(\mathbf a)$ we denote 
the minimal $k$ such that $\mathbf a^{k} = \min_{i}\mathbf a^{i}$.
Thus, $\mathcal{P}(\mathbf a_1,\dots,\mathbf a_n)$ is the position 
$k\in[M]$ of the first element in the tuple equal to the minimal element.

A $(\overline{2\ell_1+1},\dots,\overline{2\ell_n+1})$-palette alternating function 
can be defined by 
$$f(\mathbf a_1,\dots,\mathbf a_n) = 
\begin{cases}
2\min_{i,j}\mathbf a_{i}^{j}+\frac{\mathcal{P}(\mathbf a_1,\dots,\mathbf a_n)}{M}+1, & \text{if $\mathcal{S}(\mathbf a_1,\dots,\mathbf a_n)=\infty$}\\
2\min_{i,j}\mathbf a_{i}^{j}+\frac{\mathcal{S}(\mathbf a_1,\dots,\mathbf a_n)}{n},  & \text{otherwise}
\end{cases}.$$

Notice that any $(\overline{2\ell_1+1},\dots,\overline{2\ell_n+1})$-palette 
tuple $(\mathbf a_1,\dots,\mathbf a_n)$ contains a block $\mathbf a_{i}$ consisting only of  the minimal element. Hence 
$\mathcal S(\mathbf a_1,\dots,\mathbf a_n)<\infty$, which implies that 
$f$ is $(\overline{2\ell_1+1},\dots,\overline{2\ell_n+1})$-palette alternating.
Also notice that 
if $\min_{i,j}\mathbf a_{i}^{j}<\min_{i,j}\mathbf b_{i}^{j}$
then $f(\mathbf a_1,\dots,\mathbf a_n)<f(\mathbf b_1,\dots,\mathbf b_n)$.

Let us show that $f\in\Pol(\mathbb E|_{[N]},\mathbb E)$.
Denote the only ternary relation of $\mathbb E$ by $R$.
Consider three block tuples 
$\mathbf a = (\mathbf a_1,\dots,\mathbf a_n)$, $\mathbf b = (\mathbf b_1,\dots,\mathbf b_n)$, 
and $\mathbf c = (\mathbf c_1,\dots,\mathbf c_n)$
such that 
$(\mathbf a_{i}^{j},\mathbf b_{i}^{j},\mathbf c_{i}^{j})\in R$
for all $i,j$.
We need to prove that 
$(f(\mathbf a),f(\mathbf b),f(\mathbf c))\in R$.
Let $m = \min\{\min_{i,j}\mathbf a_{i}^{j},\min_{i,j}\mathbf b_{i}^{j},\min_{i,j}\mathbf c_{i}^{j}\}$.
If $m$ appears only in two tuples, say $\mathbf a$ and $\mathbf b$, 
then $\mathcal S(\mathbf a) = \mathcal S(\mathbf b)$ and $\mathcal P(\mathbf a) = \mathcal P(\mathbf b)$, which implies 
$f(\mathbf a) = f(\mathbf b)<f(\mathbf c)$.
Thus, assume that $m$ appears in each of $\mathbf a,\mathbf b,\mathbf c$.

Let $q = \min\{\mathcal S(\mathbf a),\mathcal S(\mathbf b),\mathcal S(\mathbf c)\}$.
If $q=\infty$ then exactly two of the numbers
$\mathcal P(\mathbf a),
\mathcal P(\mathbf b),
\mathcal P(\mathbf c)$ are equal and smaller than the remaining one, 
which implies $(f(\mathbf a), f(\mathbf b), f(\mathbf c))\in R$.
If $q<\infty$, we may assume (by symmetry of $R$) that 
$\mathcal S(\mathbf a)=q$. 
Thus, $\mathbf a_{q}$ contains an odd number of $m$.
Since
$(\mathbf a_{i}^{j},\mathbf b_{i}^{j},\mathbf c_{i}^{j})$ contains an even number of
$m$ for all $i,j$,
it follows that
either $\mathbf b_{q}$ or $\mathbf c_{q}$ also contains an odd number of $m$.
Hence 
one of $\mathcal S(\mathbf b)$ or $\mathcal S(\mathbf c)$
is equal to $q$, while the other is greater. This implies $(f(\mathbf a),f(\mathbf b),f(\mathbf c))\in R$.
\end{proof}

\begin{cor}\label{CORESinglAIP}
$\Pol(\mathbb E)$ contains 
    a $(\overline{2\ell_1+1},\dots,\overline{2\ell_n+1})$-palette alternating function
    for all $n,\ell_1,\dots,\ell_n\in \mathbb N$.
\end{cor}

\begin{lem}
    Suppose 
    $\mathbb E$ pp-defines $\mathbb A$. Then 
    $\PCSP(\mathbb A|_{[N]},\mathbb A)$ can be solved by 
$\Singl\AIP$ for all $N\in\mathbb N$.
\end{lem}

\begin{proof}
By Corollary 
\ref{CORESinglAIP}, $\Pol(\mathbb E)$ contains 
    a $(\overline{2\ell_1+1},\dots,\overline{2\ell_n+1})$-palette alternating function
    for all $n,\ell_1,\dots,\ell_n\in \mathbb N$.
Therefore, $\Pol(\mathbb A)$ and $\Pol(\mathbb A|_{[N]},\mathbb A)$ also contain them.
  Following the argument in Section \ref{SUBSECTIONPaletteOperationDefinition}
  we can show that these functions can be used to get a (true) 
  solution from the output of $\Singl\AIP$, which completes the proof.
\end{proof}

\begin{lem}
    $\Pol(\mathbb F|_{[N]},\mathbb F)$ does not contain
    an $(\overline{\ell_1},\dots,\overline{\ell_n})$-palette symmetric function
    for any $n,\ell_1$, $\dots$, $\ell_n,N\ge 4$.
\end{lem}

\begin{proof}
Assume the contrary. Suppose $f\in \Pol(\mathbb F|_{[N]},\mathbb F)$ is an $(\overline{\ell_1},\dots,\overline{\ell_n})$-palette symmetric function.
For $a\in [N]$ and $\ell\in\mathbb N$, by $a^{\ell}$ we denote the $\ell$-tuple consisting of an element $a$. 
Choose 3 different elements $a,b,c\in [N]$,
and tuples $\mathbf d_{i}, \mathbf e_{i}\in \{a,b,c\}^{\ell_i}$ 
for each $i\in\{3,4,\dots,n\}$
such that $\mathbf d_{i}$ and $\mathbf e_{i}$ contain the same multiset of elements but 
$\mathbf d_{i}^{j}\neq \mathbf e_{i}^{j}$ for every $j\in[\ell_i]$.
Since $f$ is $(\overline{\ell_1},\dots,\overline{\ell_n})$-palette symmetric,
$f(a^{\ell_1},b^{\ell_2},c^{\ell_3},\mathbf d_3,\dots,\mathbf d_{n}) = 
f(a^{\ell_1},b^{\ell_2},c^{\ell_3},\mathbf e_3,\dots,\mathbf e_{n}).$
Since $f$ preserves the relation $x=y\Rightarrow u=v$, 
we obtain
$$f(\mathbf u_1,\mathbf u_2,\mathbf u_3,\mathbf u_4,\dots,\mathbf u_{n}) = 
f(\mathbf u_1,\mathbf u_2,\mathbf u_3,\mathbf v_{4},\dots,\mathbf v_{n})$$
for all 
$\mathbf u_1,\dots,\mathbf u_{n}, \mathbf v_{4},\dots,\mathbf v_{n}$.
Since $n\ge 4$, the last $\ell_n$ coordinates of $f$ are dummy. By permuting the arguments, we conclude that all the coordinates are dummy, so $f$ must be a constant function.
But a constant function does not preserve the relation $x\neq y$.
This contradiction completes the proof.
\end{proof}

\subsection{Singleton vs constraint-singleton}

Here we show that 
for infinite relational structures,
the algorithm $\CSingl\ArcCons$ can be stronger than $\Singl\ArcCons$, 
as witnessed by the structure $\mathbb F$.

\begin{lem}
    $\Singl(\BLP+\AIP)$ does not solve 
    $\PCSP(\mathbb F|_{[N]},\mathbb F)$ for $N\ge 3$.
\end{lem}

\begin{proof}
Consider an instance $\mathcal I$ of $\PCSP(\mathbb F|_{[N]},\mathbb F)$
given by
$(x=x\Rightarrow y=z)\wedge 
(y\neq z)$, which clearly has no solution.
Let us show that $\Singl(\BLP+\AIP)$ nevertheless returns ``Yes'' on this instance. 
If we restrict 
$x$ to some $c\in[N]$, 
then for a BLP solution of the first constraint we take a uniform distribution 
over all tuples $(c,c,a,a)$ where $a\in[N]$,
and uniform distribution over tuples $(a,b)$ with $a\neq b$ for the second constraint.
By Claim \ref{CLAIMAIPParallelogram}, 
AIP gives the same result if we replace the constraint $y\neq z$ by its parallelogram closure, which is the full relation.
Therefore, $\BLP+\AIP$ returns ``Yes'' after this restriction.

If we restrict 
$y$ to some $c\in[N]$, then for the first constraint we take a uniform distribution 
over tuples $(a,b,c,d)$ with $a\neq b$, $c\neq d$, $a,b,d\in [N]$, 
and for the second constraint we take a uniform distribution 
over tuples $(c,d)$ with $c\neq d$ and $d\in [N]$.
Again, by Claim \ref{CLAIMAIPParallelogram}, 
AIP gives the same result if we replace the first constraint by its parallelogram closure, which includes all tuples $(a,a,c,c)$ where $a\in[N]$.
Hence, $\BLP+\AIP$ also returns ``Yes'' after this restriction.

Thus, $\Singl(\BLP+\AIP)$ never restricts any domain and therefore returns ``Yes'' on this instance.
\end{proof}

To show that $\PCSP(\mathbb F|_{[N]},\mathbb F)$ can be solved 
by $\CSingl\ArcCons$, we need a palette version of a polymorphism.
Let $n,k_1,\dots,k_n\in\mathbb N$ with $k_1+\dots+k_n = m$.
A function
$f\colon A^{m}\to B$ is called a \emph{$(\overline{k_1},\dots,\overline{k_n})$-palette polymorphism} from a $\sigma$-structure $\mathbb A$
to a $\sigma$-structure $\mathbb B$
if,
for each symbol $R\in \sigma$ and  
each $(\overline{k_1},\dots,\overline{k_n})$-palette tuple $(\mathbf a_1,\dots,\mathbf a_{m})\in (R^{\mathbb A})^{m}$,
we have 
$f(\mathbf a_1,\dots,\mathbf a_m) 
 \in R^{\mathbb B}$.
By $(\overline{k_1},\dots,\overline{k_n})$-$\Pol(\mathbb A,\mathbb B)$ we denote the set of all 
$(\overline{k_1},\dots,\overline{k_n})$-palette polymorphisms from $\mathbb A$
to $\mathbb B$.
Informally, for functions from $(\overline{k_1},\dots,\overline{k_n})$-$\Pol(\mathbb A,\mathbb B)$
we only require the resulting tuple to belong to the relation 
if each tuple that appears fills at least one block completely.

\begin{lem}\label{LEMppDefinitionInPalettePol}
Suppose $\mathbb A$ pp-defines $\mathbb B$ and 
$f\in(\overline{k_1},\dots,\overline{k_n})$-$\Pol(\mathbb A)$.
Then 
$f\in(\overline{k_1},\dots,\overline{k_n})$-$\Pol(\mathbb B)$.
\end{lem}

\begin{proof}
Let $m = k_1+\dots+k_n$.
The proof follows almost word for word the standard argument showing
that $f\in\Pol(\mathbb A) \Rightarrow f\in\Pol(\mathbb B)$.
Here, we illustrate why this holds in the case where 
$\mathbb A$ has two binary relations $R_1$ and $R_{2}$,
$\mathbb B$ has one relation $R$, 
and $R(x,y) = \exists z R_1(x,z)\wedge R_{2}(y,z)$.
Choose a mapping $\phi\colon R\to A$ such that 
$(a_1,\phi(a_1,a_2))\in R_{1}$
and $(a_2,\phi(a_1,a_2))\in R_{2}$
for all $(a_1,a_2)\in R$.
Let $(\mathbf a_1,\dots,\mathbf a_m)\in (R)^{m}$ be a $(\overline{k_1},\dots,\overline{k_n})$-palette tuple.
Set $\mathbf b_{i} = (\mathbf a_{i}^{1},\phi(\mathbf a_i))$,
$\mathbf c_{i} = (\mathbf a_{i}^{2},\phi(\mathbf a_i))$
for every $i\in [m]$.
Then $(\mathbf b_1,\dots,\mathbf b_m)\in (R_1)^{m}$ and 
$(\mathbf c_1,\dots,\mathbf c_m)\in (R_2)^{m}$ 
are $(\overline{k_1},\dots,\overline{k_n})$-palette tuples.
Let $f(\mathbf b_1,\dots,\mathbf b_m) = (b,d)$ and
$f(\mathbf c_1,\dots,\mathbf c_m) = (c,d)$.
Since $f\in(\overline{k_1},\dots,\overline{k_n})$-$\Pol(\mathbb A)$, 
$(b,d)\in R_{1}$ and
$(c,d)\in R_2$.
Therefore $f(\mathbf a_1,\dots,\mathbf a_m) = (b,c)\in R$.
\end{proof}

\begin{lem}\label{LEMCSinglArcCriterion}
Suppose $\mathbb A$ and $\mathbb B$ are relational structures over the same signature $\sigma$, with $A$ finite, 
and suppose that for all $n,k_1,\dots,k_n\in\mathbb N$
there exists 
a $(\overline{k_1},\dots,\overline{k_n})$-palette totally symmetric function 
$f\in (\overline{k_1},\dots,\overline{k_n})$-$\Pol(\mathbb A,\mathbb B)$.
Then 
$\CSingl\ArcCons$ solves $\PCSP(\mathbb A,\mathbb B)$.
\end{lem}
\begin{proof}
The argument follows the proof of Lemma \ref{LEMWeakForSinglArcImpliesAlgorithm} with one modification.
As before, notice that by identifying variables within the blocks of a 
$(\overline{k_1},\dots,\overline{k_n})$-palette totally symmetric (PTS) function
we obtain a PTS function with smaller block sizes.
Therefore, for all 
$L_1,\dots,L_n\in \mathbb N$
we can choose 
an $(\overline{L_1},\dots,\overline{L_n})$-PTS function
$f_{L_1,\dots,L_n}\in(\overline{L_1},\dots,\overline{L_n})$-$\Pol(\mathbb A,\mathbb B)$ such that 
identifying variables within blocks of
$f_{L_1,\dots,L_n}$ yields 
$f_{\ell_1,\dots,\ell_n}$, where $\ell_i\le L_i$ for every $i$.
We write 
$f(S_1,\dots,S_n)$ meaning 
$f_{k_1,\dots,k_n}(s_1^1,\dots,s_{1}^{k_1},\dots,
s_n^1,\dots,s_{n}^{k_n})$,
where $S_{i} = \{s_i^{1},\dots,s_{i}^{k_i}\}$ for every $i\in[n]$.
This notation is justified because we only 
apply $f_{k_1,\dots,k_n}$ to palette tuples, so the order of elements in the sets $S_{i}$ does not matter. 

We need to check that if 
$\CSingl\ArcCons$ returns ``Yes'' on an instance $\mathcal I$, then the instance has a solution in $\mathbb B$.
Let $P_1,\dots,P_S$ be all pairs $(C,\mathbf a)$, where $C$ is a constraint of $\mathcal I$ and $\mathbf a$ is a tuple from the reduced constraint relation at the end of the algorithm. 
For each $P_j = (C,\mathbf a)$ and each constraint $C_0=R_{0}(x_1,\dots,x_k)$, let 
$R_0^{j}$ be 
the reduced constraint relation of $C_{0}$ after 
restricting $C$ to $\mathbf a$ and running arc-consistency.
Then every constraint $C_{0}$ can be consistently evaluated 
 to the tuple 
$f(R_{0}^{1},\dots,R_{0}^{S})\in R_{0}^{\mathbb B}$.
Moreover, 
these evaluations agree on shared variables across all constraints 
because the function $f$ is palette totally symmetric.
This confirms that the instance has a solution in $\mathbb B$.
\end{proof}

\begin{lem}\label{LEMFPolymorphisms}
    $(\overline{k_1},\dots,\overline{k_n})$-$\Pol(\mathbb F|_{[N]},\mathbb F)$ contains 
    a block totally symmetric function
    for all $n,N,k_1,\dots$, $k_n\in \mathbb N$.
\end{lem}

\begin{proof}

Let us define 
a block totally symmetric function $f\in (\overline{k_1},\dots,\overline{k_n})$-$\Pol(\mathbb F|_{[N]},\mathbb F)$ by
$$f(\mathbf a_1,\dots,\mathbf a_n) = 
\sum\limits_{i\in[n]}(\min_{j}\mathbf a_{i}^{j})/(N+1)^{i}.$$

Notice that the function 
$g\colon [N]^{n}\to \mathbb Q$ defined by 
$g(x_1,\dots,x_n) = \sum\limits_{i\in[n]}x_{i}/(N+1)^{i}$
preserves the lexicographic order, that is, 
if $a_1=b_1,\dots,a_{j-1}=b_{j-1}, a_{j}<b_{j}$, 
then $g(a_1,\dots,a_n)<g(b_1,\dots,b_{n})$.
This follows from
$$(b_{j}-a_{j})\ge 1= N\cdot \sum\limits_{i=1}^{\infty}\frac{1}{(N+1)^{i}}>
N\cdot \sum\limits_{i=1}^{n-j-1}\frac{1}{(N+1)^{i}}>
(b_{j+i}-a_{j+i})\cdot \sum\limits_{i=1}^{n-j-1}\frac{1}{(N+1)^{i}}.$$

It follows from the definition 
that $f$ depends only on the set of elements in each block.
Hence, $f$ is block totally symmetric.
Set 
$M = k_1+\dots+k_n$.
Let us show that $f\in(\overline{k_1},\dots,\overline{k_n})$-$\Pol(\mathbb F|_{[N]},\mathbb F)$.

$x\neq y$. Consider
two block tuples 
$\mathbf a = (\mathbf a_1,\dots,\mathbf a_n)$ and $\mathbf b = (\mathbf b_1,\dots,\mathbf b_n)$
such that 
$\mathbf a_{i}^{j}\neq \mathbf b_{i}^{j}$ for all $i,j$.
Assume moreover that 
when we form the $M$-tuple whose entries (from $[N]^{2}$) are taken 
from $\mathbf a$ and $\mathbf b$, we obtain a $(\overline{k_1},\dots,\overline{k_n})$-palette tuple.
We need to show that 
$f(\mathbf a)\neq f(\mathbf b)$.
Because of the palette property 
there exists 
$i\in[n]$ such that $\mathbf a_{i}^{j} = \mathbf a_{1}^{1}$ and $\mathbf b_{i}^{j} = \mathbf b_{1}^{1}$ for every $j\in[k_{i}]$.
Since $\mathbf a_{1}^{1}\neq \mathbf b_{1}^{1}$, 
$\mathbf a$ and $\mathbf b$ have two different constant blocks. 
Hence $f(\mathbf a)\neq f(\mathbf b)$.

$x=y\Rightarrow u=v$.
Consider four block tuples 
$\mathbf a = (\mathbf a_1,\dots,\mathbf a_n)$, $\mathbf b = (\mathbf b_1,\dots,\mathbf b_n)$, $\mathbf c = (\mathbf c_1,\dots,\mathbf c_n)$,
and $\mathbf d = (\mathbf d_1,\dots,\mathbf d_n)$
such that $\mathbf a_{i}^{j}=\mathbf b_{i}^{j}\Rightarrow \mathbf c_{i}^{j}=\mathbf d_{i}^{j}$
for all $i,j$.
Assume moreover that 
when we form an $M$-tuple whose entries (from $[N]^{4}$) 
are taken from $\mathbf a$, $\mathbf b$,$\mathbf c$, and $\mathbf d$, we obtain a $(\overline{k_1},\dots,\overline{k_n})$-palette tuple.
We need to prove that 
$f(\mathbf a)=f(\mathbf b)\Rightarrow f(\mathbf c)=f(\mathbf d)$.
Consider two cases. Assume that 
$\mathbf a = \mathbf b$.
Then $\mathbf c= \mathbf d$, 
and, therefore, $f(\mathbf a)=f(\mathbf b)$ and 
$f(\mathbf c)=f(\mathbf d)$, which completes this case.
Assume that $\mathbf a_{i}^{j} \neq \mathbf b_{i}^{j}$ for some $i,j$.
By the palette property, 
there exists $m\in [n]$ such that 
$\mathbf a_{m}^{\ell}=\mathbf a_{i}^{j}$ and 
$\mathbf b_{m}^{\ell}=\mathbf b_{i}^{j}$ for all $\ell\in[k_{m}]$.
Thus, $\mathbf a$ and $\mathbf b$ have two different constant blocks. 
Hence, $f(\mathbf a)\neq f(\mathbf b)$.

$x>y\vee x>z\vee x=y=z$.
Consider three block tuples 
$\mathbf a = (\mathbf a_1,\dots,\mathbf a_n)$, $\mathbf b = (\mathbf b_1,\dots,\mathbf b_n)$, 
and $\mathbf c = (\mathbf c_1,\dots,\mathbf c_n)$
such that 
$\mathbf a_{i}^{j}>\mathbf b_{i}^{j}\vee \mathbf a_{i}^{j}>\mathbf c_{i}^{j}\vee \mathbf a_{i}^{j}=\mathbf b_{i}^{j}=\mathbf c_{i}^{j}$
for all $i,j$.
It follows from the definition that for every $i\in[n]$
either 
$\min_{j}\mathbf a_{i}^{j} = 
\min_{j}\mathbf b_{i}^{j} = \min_{j}\mathbf c_{i}^{j}$, or 
$\min_{j}\mathbf a_{i}^{j}>\min_{j}\mathbf b_{i}^{j}$, or 
$\min_{j}\mathbf a_{i}^{j}>\min_{j}\mathbf c_{i}^{j}$.
Since $g$ preserves the lexicographic order, 
we obtain 
$f(\mathbf a)>f(\mathbf b)\vee f(\mathbf a)>f(\mathbf c)\vee f(\mathbf a)=f(\mathbf b)=f(\mathbf c)$.
\end{proof}

\begin{cor}\label{CORPolExistence}
$(\overline{k_1},\dots,\overline{k_n})$-$\Pol(\mathbb F)$ contains 
    a block totally symmetric function
    for all $n,k_1,\dots,k_n\in \mathbb N$.
\end{cor}

\begin{lem}
Suppose $\mathbb F$ pp-defines $\mathbb A$. 
Then $\PCSP(\mathbb A|_{[N]},\mathbb A)$ can be solved by 
$\CSingl\ArcCons$  for all $N\in\mathbb N$.
\end{lem}

\begin{proof}
By Corollary \ref{CORPolExistence},
$(\overline{k_1},\dots,\overline{k_n})$-$\Pol(\mathbb F)$ 
contains a block totally symmetric function $f$ for every  $n\in\mathbb N$ and every $k_1,\dots,k_n\in\mathbb N$.
By Lemma \ref{LEMppDefinitionInPalettePol},
each $(\overline{k_1},\dots,\overline{k_n})$-$\Pol(\mathbb A)$, 
also contains it.
Then Lemma \ref{LEMCSinglArcCriterion} 
implies that 
$\PCSP(\mathbb A|_{[N]},\mathbb A)$ can be solved by 
$\CSingl\ArcCons$.
\end{proof}

\subsection{Conclusions}

 Following the argument in Section \ref{SUBSECTIONArcConsistency}
 we conclude that if $\min\in \Pol(\mathbb A)$, then $\PCSP(\mathbb A|_{[N]},\mathbb A)$ 
 can be solved by the arc-consistency algorithm, 
 and a (true) solution can be obtained from its output by applying the totally symmetric function $\min$.
Similarly, following the argument in Section \ref{SUBSECTIONPaletteOperationDefinition}
 we can use an $(\overline{\ell_1},\dots,\overline{\ell_n})$-palette totally symmetric function
 and a $(\overline{2\ell_1+1},\dots,\overline{2\ell_n+1})$-palette alternating function 
 to obtain a (true) solution 
 from the output of the algorithms $\Singl\ArcCons$ and $\Singl\AIP$, respectively.
Therefore, $\PCSP(\mathbb D|_{[N]},\mathbb D)$
and 
$\PCSP(\mathbb E|_{[N]},\mathbb E)$
can be solved by 
$\Singl\ArcCons$ and $\Singl\AIP$, respectively.

On the other hand, $\Pol(\mathbb F|_{[N]},\mathbb F)$ does not have 
palette totally symmetric functions, 
and even $\Singl(\BLP+\AIP)$ fails to solve $\PCSP(\mathbb F|_{[N]},\mathbb F)$.
Nevertheless, 
one can define a palette polymorphism that witnesses
that $\CSingl\ArcCons$ solves 
$\PCSP(\mathbb F|_{[N]},\mathbb F)$.
This result confirms that 
$\Singl$- and $\CSingl$- algorithms have the same power only for finite relational structures,
and highlights that 
the Hales-Jewett theorem
is essential for proving the characterization in terms 
of palette symmetric functions.
 
We now summarize the results of this section and present 
the algorithm for each case from 
Theorem \ref{THMTemporalClassification}.


\begin{itemize}
    \item[(1)] $\Pol(\mathbb A)$ contains a constant operation;

$\CSP(\mathbb A)$ is trivial: every instance has a solution.
    
    \item[(2)] $\min\in \Pol(\mathbb A)$ or $\max\in \Pol(\mathbb A)$;

$\Pol(\mathbb A)$ has a totally symmetric function 
of any arity;

$\PCSP(\mathbb A|_{[N]},\mathbb A)$ can be solved by 
$\ArcCons$.
    
    \item[(3)] $\mathbb D$ or $\mathbb D^{dual}$ pp-defines $\mathbb A$;

$\Pol(\mathbb D)$ contains an $(\overline{\ell_1},\dots,\overline{\ell_n})$-palette totally symmetric function for all $n,\ell_1,\dots,\ell_n\in \mathbb N$;

$\PCSP(\mathbb A|_{[N]},\mathbb A)$ can be solved by 
$\Singl\ArcCons$.

    \item[(4)] $\mathbb E$ or $\mathbb E^{dual}$ pp-defines $\mathbb A$;

$\Pol(\mathbb E)$ contains 
a $(\overline{2\ell_1+1},\dots,\overline{2\ell_n+1})$-palette alternating function
    for all $n,\ell_1,\dots,\ell_n\in \mathbb N$;
    
    $\PCSP(\mathbb A|_{[N]},\mathbb A)$ can be solved by 
$\Singl\AIP$.
    
    \item[(5)] $\mathbb F$ or $\mathbb F^{dual}$ pp-defines $\mathbb A$.

    $\Pol(\mathbb F|_{[N]},\mathbb F)$ does not contain
    an $(\overline{\ell_1},\dots,\overline{\ell_n})$-palette symmetric function
    for any $n,\ell_1,\dots,\ell_n,N\ge 4$;
    
    $(\overline{k_1},\dots,\overline{k_n})$-$\Pol(\mathbb F)$ contains 
    a block totally symmetric function
    for all $n,k_1,\dots,k_n\in \mathbb N$;

    $\PCSP(\mathbb A|_{[N]},\mathbb A)$ can be solved by 
$\CSingl\ArcCons$, but $\PCSP(\mathbb F|_{[N]},\mathbb F)$ cannot 
be solved
by 
$\Singl\ArcCons$ or even $\Singl(\BLP+\AIP)$.
\end{itemize}

Combining the above results with Theorem 
\ref{THMTemporalClassification} we obtain the following theorem from Section \ref{SUBSECTIONMainResults}.

\begin{THMMyTemporalClassificationTHM}
Let $\mathbb A$ be a temporal structure, $N\in\mathbb N$.
Then one of the following holds:
\begin{itemize}
\item $\PCSP(\mathbb A|_{[N]}, \mathbb A)$ is solvable by
$\CSingl\ArcCons\wedge \Singl\AIP$;
\item  $\CSP(\mathbb A)$ is NP-hard.
\end{itemize}
\end{THMMyTemporalClassificationTHM}

\begin{remark}
Unlike polymorphisms of $\mathbb D$ and $\mathbb E$, which come from 
compactness arguments and rely on the axiom of choice, 
the polymorphisms constructed in Lemmas \ref{LEMDPolymorphisms},
\ref{LEMEPolymorphisms},
and \ref{LEMFPolymorphisms} are concrete and can be used to obtain a (true) solution from the outputs of $\CSingl\ArcCons$ and $\Singl\AIP$.

\end{remark}

\section{Dihedral group $\mathbf D_{4}$.}
\label{SECTIONDihedralGroupProofs} 

In this section, we prove two properties of the dihedral group $\mathbf D_4$
that were stated in Section \ref{SUBSECTIONDihedralGroup}.
Recall that $\mathbb  D_4=(\mathbb Z_{2}^{3}; 
L_{3}^{1,2},L_{2,3}^{1}, L_{1,3}^{2}, E_{1-2},E_{2-3}^{0}, E_{1-3}^{0},R,\{(0,0,0)\})$.
In this section addition and multiplication are always modulo 2,
but $x_{i}^{2}$ denotes the second component of the vector $x_{i}$, not the second power.
The following claims explain the meaning of some of the relations in $\mathbb  D_4$.
We assume that $f$ is an $n$-ary operation on $\mathbb Z_{2}^{3}$.

\begin{claim}\label{ClaimSemidirect}
$f\in \Pol(\{L_3^{1,2}\})$ if and only if 
there exist  
$g_1,g_2,g_3: \mathbb Z_{2}^{2n}\to \mathbb Z_{2}$ 
and $a_1,\dots,a_n\in\mathbb Z_2$ such that
\begin{align*}
f^{1}(x_1,\dots,x_n)&=
g_{1}(x_1^{1},\dots,x_{n}^{1},
x_1^{2},\dots,x_{n}^{2}),\\
f^{2}(x_1,\dots,x_n)&=
g_{2}(x_1^{1},\dots,x_{n}^{1},
x_1^{2},\dots,x_{n}^{2}),\\
f^{3}(x_1,\dots,x_n) &= 
a_1 x_1^3+\dots+a_n x_n^3 +
g_{3}(x_1^{1},\dots,x_{n}^{1},
x_1^{2},\dots,x_{n}^{2}).
\end{align*}
\end{claim}

\begin{proof}
Let $E = \proj_{1,2}(L_3^{1,2})$. 
Then $E$ is an equivalence relation and it is preserved by $f$. 
This implies that the first two components of $f$ depend only on the first two components of the inputs, so we can choose the corresponding $g_1$ and $g_2$.
Applying $f$ to the tuples
$((x_i^1,x_i^2,x_i^3),(x_i^1,x_i^2,0),(0,0,x_i^3),(0,0,0))\in L_3^{1,2}$
for $i=1,\dots,n$
and looking at the third component,
we obtain 
\begin{equation}\label{EQUATIONDecomposition}
f^{3}(x_1,\dots,x_n) = 
f^{3}((x_1^1,x_1^2,0),\dots,(x_n^1,x_n^2,0)) + 
f^{3}((0,0,x_1^3),\dots,(0,0,x_n^3)) +
f^{3}((0,0,0),\dots,(0,0,0)).
\end{equation}
Similarly, we derive  
\begin{align*}
    f^{3}((0,0,x_1^3),\dots,(0,0,x_n^3)) =& \\
f^{3}((0,0,0),(0,0,x_2^3),\dots,&(0,0,x_n^3))+
f^{3}((0,0,x_1^3),(0,0,0),\dots,(0,0,0))+
f^{3}((0,0,0),\dots,(0,0,0)).
\end{align*}
This holds analogously when considering the $i$-th coordinate instead of the first.
Choose $a_1,\dots,a_n,b_1,\dots,b_n\in\mathbb Z_{2}$ so that 
$$f^{3}(\underbrace{(0,0,0),\dots,(0,0,0),(0,0,z)}_i,(0,0,0),\dots,(0,0,0))=a_i z+b_i \text{\quad\quad\quad for every $i\in[n]$.}$$ 
Using (\ref{EQUATIONDecomposition}),
we derive that 
$f^{3}(x_1,\dots,x_n) = 
a_1 x_1^3+\dots+a_n x_n^3 +
g_{3}(x_1^{1},\dots,x_{n}^{1},
x_1^{2},\dots,x_{n}^{2})$
for some $g_3$.
\end{proof}

\begin{claim}\label{ClaimSameonOneandTwo}
$f\in \Pol(\{E_{1-2}\})$ if and only if
there exists $g:\mathbb Z_{2}^{n}\to Z_{2}$ such that 
\begin{align*}
f^{1}(x_1,\dots,x_n)&=
g(x_1^{1},\dots,x_{n}^{1}),\\
f^{2}(x_1,\dots,x_n)&=
g(x_1^{2},\dots,x_{n}^{2}).
\end{align*}
\end{claim}

\begin{proof}
Applying $f$ to tuples
$((a_i,b_i,c_i),(d_i,a_i,e_i)\in E_{1-2}$
for $i=1,\dots,n$,
and using the fact that the result must be in $E_{1-2}$, 
we obtain
$$f^{1}((a_1,b_1,c_1),\dots,(a_n,b_n,c_n)) = 
f^{2}((d_1,a_1,e_1),\dots,(d_n,a_n,e_n)).$$
This implies that 
$f^{1}$ depends only on the first components, $f^{2}$ depends only on the second components, 
and the operations defining $f^{1}$ and $f^2$ are identical.
\end{proof}

\begin{claim}\label{ClaimSecondPowerOne}
Suppose
$f\in \Pol(\{L_{3}^{1,2},L_{2,3}^{1}\})$.
Then 
every monomial in 
the polynomial representation 
of $g_3$ (see Claim \ref{ClaimSemidirect}) contains at most one $x_{i}^{2}$-variable
(i.e, terms like $x_{1}^{2}\cdot x_{2}^{2}$ are not allowed).
\end{claim}

\begin{proof}
Assume the contrary. Choose monomials of the form 
$x_{i_1}^{1}\dots x_{i_k}^{1} x_{j_1}^{2}\dots x_{j_m}^{2}$ 
with minimal $m\ge 2$;
among them choose one with minimal $k\ge 0$.
Set $b_{i}=1$ if $i\in\{i_1,\dots,i_k\}$ 
and $b_{i}=0$ otherwise;
$c_{i}=1$ if $i\in\{j_1,\dots,j_m\}$ 
and $c_{i}=0$ otherwise;
$d_{i}=1$ if $i\in\{j_2,\dots,j_m\}$ 
and $d_{i}=0$ otherwise; 
$e_{j_1}=1$ and 
$e_{i} = 0$ for $i\neq j_1$.
Applying $f$ to the tuples
$((b_i,c_i,0),(b_i,d_i,0),(b_i,e_i,0),(b_i,0,0))\in L_{2,3}^{1}$
for $i=1,\dots,n$,
we get a tuple from $L_{2,3}^{1}$.
Examining the third components of the obtained tuple, we derive:
$$g_3(b_1,\dots,b_n,c_1,\dots,c_n)+
g_3(b_1,\dots,b_n,d_1,\dots,d_n)+
g_3(b_1,\dots,b_n,e_1,\dots,e_n)+
g_3(b_1,\dots,b_n,0,\dots,0)=0.$$
All the monomials of $g_3$ involving only the first components 
vanish, as do all linear monomials.
The only monomial contributing 1 to the above sum is 
$x_{i_1}^{1}\dots x_{i_k}^{1} x_{j_1}^{2}\dots x_{j_m}^{2}$,
which gives us a contradiction.
\end{proof}

The next claim differs from the previous one only by swapping 
1 and 2, so we omit the proof.

\begin{claim}\label{ClaimFirstPowerOne}
Suppose
$f\in \Pol(\{L_{3}^{1,2},L_{1,3}^{2}\})$.
Then 
every monomial in 
the polynomial representation 
of $g_3$ (see Claim \ref{ClaimSemidirect}) contains at most one $x_{i}^{1}$-variable
(i.e., terms like $x_{1}^{1}\cdot x_{2}^{1}$ are not allowed).
\end{claim}



Let us describe a set of operations $\mathcal D_4$, 
which we will later prove to be the clone generated by the dihedral group.
Suppose $n\in\mathbb N$, $a:[n]\to\mathbb Z_2$, 
and $c: [n]\times [n]\to \mathbb Z_{2}$
are such that 
$c_{i,j} + c_{j,i} = a_i\cdot a_j$ for every $i,j\in[n]$, 
$i\neq j$.
Then we define the following $n$-ary operation $f$ on $\mathbb Z_{2}^{3}$:
\begin{align*}
f^{1}(x_1,\dots,x_n)&=
a_{1} x_1^{1}+ \dots +a_{n} x_{n}^{1},\\
f^{2}(x_1,\dots,x_n)&=a_{1} x_1^{2}+\dots+ a_{n} x_{n}^{2},\\
f^{3}(x_1,\dots,x_n) &= 
a_1 x_1^3+\dots+a_n x_n^3 +
\sum\limits_{i,j\in[n]} c_{i,j} x_{i}^{1} x_{j}^{2}.
\end{align*}
By 
$\mathcal D_4$ we denote the set of all 
such operations $f$ for all $n$, $a$ and $c$.

\begin{claim}\label{ClaimDFourGenerate}
$ \mathcal D_{4} \subseteq \Clo(x\circ y)$. 
\end{claim}
\begin{proof}
Let $f$ be an $n$-ary operation from $\mathcal D_{4}$ defined by some 
mappings $a$ and $c$.
Let $\{i_1,\dots,i_s\} = \{i\mid a_{i} = 1\}$
with $i_1< i_2< \dots< i_s$.
Let us show that $f$ can be defined by 
$$f(x_1,\dots,x_n) = 
x_{i_1}\circ \dots\circ x_{i_s} \circ 
\prod_{i<j\colon c_{j,i} =1} (x_{i} \circ x_j \circ  x_i^{-1} \circ  x_j^{-1})
\circ 
\prod_{i\colon c_{i,i} =1} (x_{i} \circ x_i).
$$
At the first two components the operation  $\circ$ is just an addition; hence the 
last two products do not affect them, and the result matches the definition. 
We claim that 
the third component of 
$x_{i_1}\circ \dots\circ x_{i_s}$ 
is $x_{i_1}^{3}+\dots+x_{i_s}^{3} +\sum\limits_{j<k} x_{i_j}^{1}x_{i_k}^{2}$. 
This can be shown by induction on $s$, with the case $s=2$ following 
directly from the definition and the inductive step being straightforward.
Then we check that 
$(x_{i} \circ x_j \circ  x_i^{-1} \circ  x_j^{-1}) = 
(0,0,x_{i}^{1}x_{j}^{2}+ x_{i}^{2}x_{j}^{1})$ and 
$(x_{i} \circ x_i) = 
(0,0,x_{i}^{1}x_{i}^{2})$.
It remains to apply the definition of $\circ$ to verify that the third component also matches the definition.
\end{proof}

\begin{claim}\label{CLAIMDFourRelationalDescription}
$\Pol(\mathbb D_{4})\subseteq \mathcal D_{4}$.
\end{claim}

\begin{proof}
Let $f\in\Pol(\mathbb D_{4})$.
By Claim \ref{ClaimSemidirect}
we choose the corresponding $g_1,g_2,g_3$ and $a_1,\dots,a_n$.
By Claim \ref{ClaimSameonOneandTwo} we can replace $g_1$ and $g_2$ 
with a single operation $g$.
Since the relation $\{(0,0,0)\}$ is preserved by $f$, 
$g_{3}(0,0,\dots,0) = 0$.
Applying $f$ to  
tuples $((0,b_1,c_1),(0,d_1,b_1)),\dots,((0,b_n,c_n),(0,d_n,b_n))\in E_{2-3}^{0}$, 
we obtain  
$$g(b_1,\dots,b_n) = a_1 b_1+\dots+a_n b_n + g_{3}(0,0,\dots,0,d_1,\dots,d_n)$$
for all $b_1,\dots,b_n,d_1,\dots,d_n\in\mathbb Z_2$.
Similarly, using the relation 
$E_{1-3}^{0}$, we obtain 
$$g(b_1,\dots,b_n) = a_1 b_1+\dots+a_n b_n + g_{3}(d_1,\dots,d_n,0,\dots,0)$$
for all $b_1,\dots,b_n,d_1,\dots,d_n\in\mathbb Z_2$.
These two conditions imply that 
$g(z_1,\dots,z_n) = 
a_{1} z_{1}+ \dots + a_{n} z_{n}$
and $g_{3}$ returns 0 whenever the first $n$ or the last $n$ arguments equal 0.
Therefore, the polynomial representation of 
$g_3$ cannot have a monomial $x_{i}^{1}$ or $x_{i}^{2}$.
By Claims \ref{ClaimSecondPowerOne} and \ref{ClaimFirstPowerOne},
the only possible monomials 
are of the form $x_i^1\cdot x_{j}^{2}$.
Let $g_{3}(x_1^{1},\dots,x_{n}^{1},
x_1^{2},\dots,x_{n}^{2}) =\sum\limits_{i,j\in[n]} c_{i,j} x_{i}^{1} x_{j}^{2}$.
It remains to check that the relation $R$ enforces 
the property $c_{i,j} + c_{j,i} = a_i\cdot a_j$.
Substituting
the tuple 
$((1,0,0),(0,1,0))$ into the $i$-th position, 
$((0,1,0),(1,0,0))$ into the $j$-th position, 
and  
$((0,0,0),(0,0,0))$ into the remaining positions 
of $f$, we obtain the tuple $((a_i,a_j,c_{i,j}),(a_j,a_i,c_{j,i}))$, 
which must belong to $R$.
Hence, the definition of $R$ gives $a_i\cdot a_j +c_{i,j}+c_{j,i}=0$.
\end{proof}

\begin{LEMRelationsForD4LEM}
$\Pol(\mathbb D_{4}) 
=\Clo(x\circ y)= \mathcal D_{4}$.
\end{LEMRelationsForD4LEM}

\begin{proof}
We prove the inclusions 
$\Pol(\mathbb D_{4})\subseteq \mathcal D_{4}\subseteq \Clo(x\circ y)\subseteq \Pol(\mathbb D_{4})$.
The first inclusion follows from Claim \ref{CLAIMDFourRelationalDescription}, 
the second follows from Claim \ref{ClaimDFourGenerate}, 
and the last follows from the fact that $x\circ y$ preserves all the relations from 
$\mathbb D_{4}$.
\end{proof}


\begin{LEMNoPaletteSymmetricDFourLEM}
$\Clo(x\circ y)$ does not contain an idempotent  $(\underbrace{\overline{\ell+1,\ell},\overline{\ell+1,\ell},\dots,\overline{\ell+1,\ell}}_{2n})$-palette symmetric operation
for $n,\ell\ge 4$.
\end{LEMNoPaletteSymmetricDFourLEM}
\begin{proof}
Assume the contrary.
Let $f\in \Clo(x\circ y)$ be 
such an operation.
By Lemma \ref{LEMRelationsForD4}, 
$f\in \mathcal D_4$, so we can choose
$a:[n]\to \mathbb Z_{2}$ and $c:[n]\times [n]\to \mathbb Z_{2}$ 
defining $f$.
Take $i,j\in [\ell+1]$ with $i\neq j$.
Substitute 
$(1,0,0)$ for the $i$-th coordinate and into the blocks 3 and 4 (the second overlined part), 
substitute 
$(0,1,0)$ for the $j$-th coordinate and into the blocks 5 and 6 (the third overlined part), 
and $(0,0,0)$ for the remaining coordinates. 
Since $f$ is $(\underbrace{\overline{\ell+1,\ell},\overline{\ell+1,\ell},\dots,\overline{\ell+1,\ell}}_{2n})$-palette symmetric, the result must be the same for 
all $i$ and $j$. 
The first coordinate 
of the result is 
$a_i+a_{2\ell+2}+a_{2\ell+3}+\dots +a_{4\ell+2}$,
and the third one is 
$$c_{i,j} + \sum\limits_{m=1}^{2\ell+1} c_{i,4\ell+2+m} +
\sum\limits_{m=1}^{2\ell+1} c_{2\ell+1+m,j}+ \sum\limits_{m_1,m_2\in[2\ell+1]} c_{2\ell+1+m_1,4\ell+2+m_1}.$$

Denote 
$s_{i} = \sum\limits_{m=1}^{2\ell+1} c_{i,4\ell+2+m}$,
$t_{j} = \sum\limits_{m=1}^{2\ell+1} c_{2\ell+1+m,j}$.
Hence, $a_{1} = a_2 = \dots = a_{\ell+1}$,
and the value
$c_{i,j} + s_{i} + t_{j}$ does not depend on $i$ and $j$.
Since $\ell\ge 2$, 
there exist $i,j\in[\ell+1], i\neq j$, such that $s_{i} + t_{i} = s_{j} + t_{j}$.
This implies $c_{i,j} + c_{j,i} = 0$, which, 
by the definition of $\mathcal D_{4}$, 
means that $a_i\cdot a_j = 0$, and therefore 
$a_i= a_j = 0$.
Thus, we have showed that 
$a_1=\dots = a_{\ell+1} = 0$.
Repeating the argument word for word, we obtain
$a_{\ell+2}=\dots = a_{2\ell+1} = 0$.
Finally, by swapping the blocks and repeating the same argument, 
we prove that $a_1 = \dots = a_{n(2l+1)} = 0$. 
This contradicts the fact that 
$f$ is idempotent, since 
$f((1,0,0),\dots,(1,0,0)) = (1,0,0)$ implies 
$\sum_{i=1}^{n(2\ell+1)} a_{i} = 1$.
\end{proof}



\section{Palette symmetric terms for small algebras}\label{SectionSmallDomains}

In this section we prove the existence of palette symmetric (PS) term operations in a product of small Taylor algebras. To do so, we develop machinery that is considerably more general than what is needed for this particular claim. In most cases, we show that palette-symmetric terms for an algebra can be constructed from palette-symmetric terms for smaller algebras.
The situation in which this construction fails is very specific: the algebra must be Maltsev and nilpotent (see Theorem \ref{THMMainMinimalExampleFor}).
Attempts to strengthen the result by replacing the bound on the size of the domain with a purely algebraic property fail because the smaller algebras appearing in the construction are not always subalgebras; sometimes it is a retract algebra $h(\mathbf A)$ induced by a mapping $h\colon  A\to A$, sometimes 
it is a minimal Taylor reduct.

We begin with the necessary algebraic definitions. In Section \ref{SUBSECTIONSTRONGSUBUNIVERSES} we define strong subalgebras and state several of their properties from \cite{zhuk2024simplified}. Some of these properties are stronger for minimal Taylor algebras, which we introduce in Section \ref{SUBSECTIONMinimalTaylor}. In Section \ref{SUBSECTIONFREEGENERATED} we show how to prove the existence of a term operation satisfying certain identities by considering a freely generated relation.
In Sections \ref{SUBSECTIONSqueezingOperations}, \ref{SUBSECTIONGoodSetsOfTerms}, and \ref{SUBSECTIONDivideAndConquer} we explain how to construct a suitable family of terms, how to find squeezing terms within this family, and what exactly can be squeezed. In Section \ref{SUBSECTIONThreeWays} we present three methods for constructing a palette symmetric term from palette symmetric terms on smaller algebras. Finally, in Section \ref{SUBSECTIONMainProof} we derive the main results: Theorems \ref{THMMainMinimalExampleFor}
and \ref{THMExistenceWBSForSmallAlgebras}.

\subsection{Necessary definitions}
\textbf{Algebras.}
We denote algebras by bold letters $\mathbf{A},\mathbf{B}, \mathbf C,\dots$, 
their domains by $A,B,C,\dots$, 
and the basic operations
by $f^{\mathbf A},f^{\mathbf B}, g^{\mathbf C},\dots$.
We use standard universal algebraic notions of term operation, subalgebra,  factor algebra, product of algebras,
see~\cite{bergman2011universal}.
Algebras are called \emph{similar} if they have the same signature.
We write $\mathbf B\le \mathbf A$ if $\mathbf B$ is a subalgebra of $\mathbf A$.
By $\Clo(\mathbf A)$ we denote the set of all term operations in $\mathbf A$, i.e. 
the clone generated by the basic operations of $\mathbf A$.
An algebra is called 
\emph{idempotent} if all of its operations are idempotent.
In this paper we are only interested in finite idempotent Taylor algebras.
An algebra is called \emph{Taylor} if 
it satisfies a nontrivial idempotent identity. 
To distinguish identities and equations, we use $\approx$ instead of 
$=$ for identities. 
Thus, whenever we use the symbol $\approx$, we assume that all variable symbols (usually $x,y,z$) are universally quantified.
For example, an operation $m$ is called \emph{Maltsev} 
if $m(x,x,y) \approx m(y,x,x)\approx y$.
There are many equivalent definitions of Taylor algebras 
\cite{Taylor1977,siggers2010strong,KozikKrokhinValerioteWillard2015}.
We will need only the following.

\begin{thm}\label{THMMinimalTaylorClassification}
Suppose $\mathbf A$ is a finite idempotent algebra. Then the following conditions are equivalent:
\begin{enumerate}
    \item $\mathbf A$ is a Taylor algebra;
    \item $\mathbf A$ has a WNU term operation of any prime arity $p>|A|$ \cite{MarotiMcKenzie};
    \item $\mathbf A$ has a cyclic term operation of any prime arity $p>|A|$, 
    i.e. an operation $f$ satisfying
    $$f(x_1,x_2,\dots,x_p) \approx f(x_2,x_3,\dots,x_p,x_1) \cite{barto2012absorbing};$$
    \item $\mathbf A$ has an XY-symmetric term operation of any prime arity $p>|A|$, 
    i.e. an operation $f$ such that 
    $f(a_1,\dots,a_p) = f(a_{\sigma(1)},\dots,a_{\sigma(p)})$ 
    for any 
    $a_1,\dots,a_p\in A$ and any permutation $\sigma$ on $[p]$ 
    such that 
    $|\{a_1,\dots,a_p\}|=2$
    \cite{zhuk2024simplified}.
\end{enumerate}
\end{thm}

In this paper we assume that every algebra is finite, idempotent, and  Taylor.
By $0_{\mathbf A}$ we denote the equality relation on $A$, which is 
the 0-congruence on $\mathbf A$.
An algebra $(A;F_{A})$ is called \emph{polynomially complete (PC)}
if the clone generated by $F_{A}$ and all constants on $A$ is the clone of all operations on $A$
(see \cite{istinger1979characterization,lausch2000algebra}).

By $\mathcal V_{k}$ we denote the class 
of finite algebras $\mathbf A = (A;w^{\mathbf A})$ 
whose basic operation 
$w^{\mathbf A}$ is a $k$-ary idempotent WNU operation.
So the symbol $w$ always means the basic operation of an algebra from 
$\mathcal V_k$.
Since we only consider finite algebras, $\mathcal V_{k}$ is not a variety.
For a prime $p$ by $\mathbf Z_{p}$ we denote the algebra whose 
domain is $\{0,1,\dots,p-1\}$ and whose 
basic operation
$w^{\mathbf Z_{p}}$ is $a\cdot x_1+\dots+a\cdot x_{k} \pmod p$, 
where $a\cdot k = 1 \pmod p$ and  $k$ is specified by 
the membership in 
 $\mathcal V_{k}$.

For an algebra $\mathbf A\in \mathcal V_{k}$ and 
$h\colon A\to A$ such that $h\circ h = h$
by $h(\mathbf A)$ we denote the \emph{retract algebra}
whose domain is the image of $h$ and whose 
only operation $w$ is defined by
$w^{h(\mathbf A)}(x_1,\dots,x_k):=
h(w^{\mathbf A}(x_1,\dots,x_k))$.
For a function $h$, we denote its image by $\Image(h)$.
It is not hard to verify that 
$h(\mathbf A)\in\mathcal V_{k}$.


\textbf{Notations.}
A relation $R\subseteq A_{1}\times\dots\times A_{s}$ is called \emph{subdirect} if
$\proj_{i}(R) = A_{i}$ for every $i\in [s]$.
We write 
$\mathbf R\le_{sd}\mathbf A_{1}\times\dots\times \mathbf A_{m}$ if 
$R$ is a subdirect relation and $\mathbf R\le\mathbf A_{1}\times\dots\times \mathbf A_{m}$.
For an algebra $\mathbf A$ and $R\subseteq A^{n}$ by 
$\Sg_{\mathbf A}(R)$ we denote the minimal 
subuniverse of $\mathbf A^{n}$ containing $R$, that is 
\emph{the subuniverse of $\mathbf A^{n}$ generated from $R$}.
Notice that if $R = \{\alpha_1,\dots,\alpha_n\}$ then 
$\Sg_{\mathbf A}(R) = \{t^{\mathbf A}(\alpha_1,\dots,\alpha_n)\mid 
t\text{ is a term over }x_1,\dots,x_n\}$.
We say that 
\emph{a tuple $(a_1,\dots,a_n)$ generates an algebra $\mathbf A$} if
$\Sg_{\mathbf A}(\{a_1,\dots,a_n\}) = A$.
For an equivalence relation $\sigma$ on $A$ and $a\in A$ by $a/\sigma$ we denote the equivalence 
class containing $a$.
For an equivalence relation $\sigma$ on $A$ and 
$B\subseteq A$ denote 
$B/\sigma=\{b/\sigma\mid b\in B\}$.

For two binary relations $\delta_{1}\subseteq A_{1}\times A_{2}$ 
and $\delta_{2}\subseteq A_{2}\times A_{3}$ 
by $\delta_{1}\circ\delta_2$ we denote the composition, i.e. the binary 
relation 
$\{(a,b)\mid \exists c\colon (a,c)\in\delta_1\wedge (c,b)\in\delta_2\}$.
A binary subdirect relation $\delta\subseteq A\times B$ 
is called \emph{linked} if
the bipartite graph corresponding to $\delta$ is connected.


\textbf{Irreducible congruences.}
A congruence $\sigma$ on $\mathbf A$ is called \emph{irreducible} if 
there are no subalgebras 
$\mathbf S_{1},\mathbf S_{2},\dots,\mathbf S_{k}\le \mathbf A/\sigma \times \mathbf A/\sigma$ 
such that 
$0_{\mathbf A/\sigma} = S_{1}\cap S_{2}\cap \dots\cap S_{k}$ and 
$0_{\mathbf A/\sigma} \neq S_{i}$ for every $i\in[k]$. 
Similarly, we say that 
a congruence 
$\sigma$ is \emph{$\wedge$-irreducible} if it cannot be represented as 
an intersection of other congruences $\sigma_1,\dots,\sigma_n$ on $\mathbf A$.
Thus, the only difference between these two notions is that in one of them we allow any 
binary subalgebras and in another only congruences.
For a $\wedge$-irreducible congruence by $\sigma^{+}$ we denote the minimal congruence above $\sigma$.

\textbf{Bridges.}
Suppose $\sigma_{1}$ and $\sigma_{2}$ are congruences on $\mathbf D_{1}$ and $\mathbf D_{2}$, respectively.
A relation $\delta\le \mathbf D_{1}^{2}\times \mathbf D_{2}^{2}$ is called \emph{a bridge} from $\sigma_{1}$ to $\sigma_{2}$ if the following conditions hold:
\begin{enumerate}
    \item $(a_1,a_2,a_3,a_4)\in \delta\wedge 
    (a_1,b_1)\in \sigma_1\wedge (a_2,b_2)\in \sigma_1
    \wedge 
    (a_3,b_3)\in \sigma_2
    \wedge (a_4,b_4)\in \sigma_2\Longrightarrow (b_1,b_2,b_3,b_4)\in \delta$;
\item $\proj_{1,2}(\delta) \supsetneq \sigma_{1}$,
$\proj_{3,4}(\delta) \supsetneq \sigma_{2}$,
\item $(a_{1},a_2,a_{3},a_{4})\in \delta$ implies
$(a_1,a_2)\in \sigma_{1}\Leftrightarrow (a_3,a_4)\in \sigma_{2}.$
\end{enumerate}

An example of a bridge 
is the 
relation 
$\delta=\{(a_{1},a_{2},a_{3},a_{4})\mid
a_{1},a_{2},a_{3},a_{4}\in  \mathbb Z_{4}:
a_{1}-a_{2} = 2 a_{3} - 2 a_{4}\}$.
We can check that 
$\delta$ is a bridge from 
the equality relation (0-congruence) on $\mathbf Z_{4}$ 
and the $(mod\;2)$-equivalence relation on $\mathbf Z_4$.
The notion of a bridge is strongly related to other notions in Universal Algebra and Tame Congruence Theory
such as similarity and centralizers
(see \cite{willard2025zhuksbridgescentralizerssimilarity} for the detailed comparison).
For a bridge $\delta$ by $\trace(\delta)$ we denote 
the binary relation defined by 
$\trace(\delta)(x,y) = \delta(x,x,y,y)$.

We can compose a bridge $\delta_1$ from 
$\sigma_0$ to $\sigma_1$ and a bridge $\delta_{2}$ from $\sigma_1$ to 
$\sigma_2$ using the following formula:
$$\delta(x_1,x_2,z_{1},z_{2}) = \exists y_{1}\exists y_{2}\; \delta_{1}(x_{1},x_{2},y_{1},y_{2})\wedge \delta_{2}(y_{1},y_{2},z_{1},z_{2}).$$
We can prove (Lemma 28 in \cite{zhuk2024simplified}) that $\delta$ is a bridge from $\sigma_0$ to $\sigma_2$ 
whenever the congruence $\sigma_1$ is irreducible.
Moreover, $\trace(\delta) = \trace(\delta_1)\circ \trace(\delta_2)$.

%

\subsection{Two types of irreducible congruences}

An irreducible congruence $\sigma$ on $\mathbf A$ is called \emph{linear} if
there exists a bridge $\delta$ 
from $\sigma$ to $\sigma$ such that 
$\trace(\delta)\supsetneq \sigma$.
An irreducible congruence that is not linear is called \emph{a PC congruence}.
Note that a different definition is provided in \cite{zhuk2024simplified}; however, as shown in Lemma 7 of \cite{zhuk2024simplified}, it is equivalent to the one presented here.


Notice that a congruence $\sigma$ is 
an irreducible/PC/linear congruence if and only if 
$0_{\mathbf A/\sigma}$ is an irreducible/PC/linear congruence.
A congruence $\sigma$ on $\mathbf A\in \mathcal V_{k}$ is called \emph{perfect linear} if 
there exists a bridge $\delta$
from $\sigma$ to $\sigma$ such that 
$\trace(\delta) = A^{2}$. 
This definition differs from the original one, but it is equivalent by Lemma 27 of \cite{zhuk2024simplified}. 
An alternative equivalent characterization follows from commutator theory:
$\sigma$ is perfect linear iff 
$\sigma^{+}$ is abelian and 
the centralizer of $\sigma^{+}$ modulo $\sigma$ is the 1-congruence
\cite{willard2025zhuksbridgescentralizerssimilarity}.



Another important fact is that 
a bridge can only connect congruences of the same type.


\begin{lem}[Lemmas 8 and 10 in \cite{zhuk2024simplified}]\label{LEMNoBridgeBetweenDifferentTypes}

Suppose $\sigma_1$ is a linear/PC/perfect linear congruence on $\mathbf A_{1}$, 
$\sigma_2$ is an irreducible congruence on $\mathbf A_{2}$, 
there is a bridge from $\sigma_1$ to $\sigma_2$.
Then $\sigma_2$ is a linear/PC/perfect linear congruence.
Moreover, if $\sigma_1$ is a PC congruence, then 
$\mathbf A_1/\sigma_1\cong \mathbf A_2/\sigma_2$.
\end{lem}

The following claims justify the names chosen for the two types:

\begin{lem}[Lemma 11 in \cite{zhuk2024simplified}]\label{LEMLInearOnTheTopIsEasy} 
Suppose $\sigma$ is a maximal linear congruence on $\mathbf A\in\mathcal V_{k}$.
Then $\mathbf A/\sigma\cong \mathbf Z_{p}$ for some prime $p$.
\end{lem}

Notice that $0_{\mathbf Z_p}$ is a perfect linear congruence on $\mathbf Z_p$.


\begin{lem}[Lemma 12 in \cite{zhuk2024simplified}]\label{LEMPCOnTheTopIsEasy}
Suppose $\sigma$ is a maximal PC congruence on a BA and center free algebra $\mathbf A\in\mathcal V_k$. 
Then $\mathbf A/\sigma$ is a PC algebra.
\end{lem}




\textbf{$\boxtimes$-product of $\mathbf B$ and $\mathbf Z_{p}$.} 
For $x = (a,b)$
by $x^{(1)}$ and  $x^{(2)}$ we denote 
$a$ and $b$ respectively.
For an algebra $\mathbf B=(B;g^{\mathbf B})$ by 
$\mathbf B\boxtimes \mathbf Z_{p}$ we denote 
the set of algebras $\mathbf A$ such that 
$A = B\times  Z_{p}$, 
$(g^{\mathbf A}(x_1,\dots,x_n))^{(1)} =
g^{\mathbf B}(x_1^{(1)},\dots,x_n^{(1)})$
and 
$(g^{\mathbf A}(x_1,\dots,x_n))^{(2)} = 
f(x_1^{(1)},\dots,x_{n}^{(1)})
+a_1 x_1^{(2)}+\dots+a_n x_{n}^{(2)}$
for some mapping
$f\colon B^{n}\to \mathbb Z_{p}$
and 
$a_{1},\dots,a_n\in \mathbb Z_{p}$.
Notice that 
any $m$-ary term operation of $\mathbf C\in \mathbf B\boxtimes \mathbf Z_{p}$
is also determined by their behavior on the first coordinate, the corresponding coefficients 
$a_1,\dots,a_m$, and a mapping 
$f\colon B^{m}\to \mathbb Z_{p}$.

\begin{lem}[\cite{zhuk2024simplified}, Theorem 52]\label{ExistenceOfInjectiveHomomorphismTHM}
Suppose 
$\mathbf A\in \mathcal V_{k}$,
$0_{\mathbf A}$ is a perfect linear congruence.
Then there exist a prime $p$, an algebra $\mathbf C\in (\mathbf A/0_{\mathbf A}^{+}\boxtimes\mathbf Z_{p})\cap \mathcal V_{k}$ and an injective homomorphism $h\colon \mathbf A\to \mathbf C$
such that 
$h(a)^{(1)} = a/0_{\mathbf A}^{+}$ for every $a\in A$.
\end{lem}




\begin{lem}\label{LEMTrivialClassOfPerfectLinearImpliesBA} 
Suppose 
$\mathbf A\in \mathcal V_{k}$,
$0_{\mathbf A}$ is a perfect linear congruence.
Then there exists a prime $p$ such that each equivalence class of $0_{\mathbf A}^{+}$ is either of size 1, or of size $p$,
and 
the union of all the large equivalence classes binary absorbs $\mathbf A$
(see the definition of absorption in Section \ref{SUBSECTIONSTRONGSUBUNIVERSES}).
\end{lem}


\begin{proof}
By Lemma \ref{ExistenceOfInjectiveHomomorphismTHM}
there exists an algebra $\mathbf C\in (\mathbf A/0_{\mathbf A}^{+}\boxtimes\mathbf Z_{p})\cap \mathcal V_{k}$ 
and an injective homomorphism $h\colon \mathbf A\to \mathbf C$.
Then to every equivalence class of $0_{\mathbf A}^{+}$ 
we assign a subalgebra of $\mathbf Z_{p}$, which implies that 
there are two types of equivalence classes: of size $p$ and of size 1.
We claim that a binary absorbing term can be defined by 
$g(x,y) = w(x,x,\dots,x,y)$. 
Let $D$ be a large equivalence class of $0_{\mathbf A}^{+}$
and $a\in A$.
We need to prove that 
$w(D,\dots,D,a)$
and $w(a,\dots,a,D)$ belong to  large equivalence classes.
To see this, choose two different elements $d_1,d_2\in D$
and notice that 
$w^{\mathbf A}(d_2,\dots,d_2,a)\neq w^{\mathbf A}(d_1,d_2,\dots,d_2,a)$
and $w^{\mathbf A}(a,\dots,a,d_1)\neq w^{\mathbf A}(a,\dots,a,d_2)$.
\end{proof}

\begin{lem}\label{LEMXYSymmetricIsPSForSmall} 
Suppose $\mathbf C\in (\mathbf Z_2\boxtimes\mathbf Z_{2})\cap \mathcal V_{k}$.
Then 
any XY-symmetric term operation is also symmetric term operation.
\end{lem}

\begin{proof}
Let $t$ be an XY-symmetric term operation.
Since $t$ is symmetric on elements $(0,0)$ and $(0,1)$, we derive that 
$a_1=\dots = a_n$ in the definition of $\boxtimes$-product.
Since $t$ is symmetric on elements $(0,0)$ and $(1,0)$, we derive 
that both $g^{(1)}$ and $f$ in this definition are symmetric, which completes the proof.
\end{proof}


\subsection{Strong subuniverses}\label{SUBSECTIONSTRONGSUBUNIVERSES}

\textbf{(Binary) absorbing subuniverse.}
We say $B$ is \emph{an absorbing subuniverse} of an algebra 
$\mathbf A$ if 
there exists $t\in \Clo(\mathbf{A})$ such that
$t(B,B,\dots,B,A,B,\dots,B) \subseteq B$ for any position of $A$. Also in this case we say that 
\emph{$B$ absorbs $\mathbf A$ with a term operation $t$}.
We write $B\absorbseq_{n} \mathbf A$, where $B$ absorbs $\mathbf A$ and the term operation can be chosen $n$-ary.
Since we can always add a dummy variable,
$B\absorbseq_{n} \mathbf A\Rightarrow B\absorbseq_{n+1} \mathbf A$.

We 
say that $B$ binary absorbs $\mathbf A$ if 
$B\absorbseq_2 \mathbf A$, and 
$B$ ternary absorbs $\mathbf A$ if 
$B\absorbseq_3 \mathbf A$.
For brevity, we will sometimes write \emph{BA} instead of 
binary absorbing.

\textbf{Central subuniverse.}
A subuniverse $C$ of $\mathbf A$ is called \emph{central} 
if it is an absorbing subuniverse 
and 
for every $a\in A\setminus C$ 
we have 
$(a,a)\notin \Sg_{\mathbf A}((\{a\}\times C)\cup (C\times \{a\}))$.
Central subuniverses are strongly connected with ternary absorption.

\begin{lem}[\cite{zhuk2021strong}, Corollary 6.11.1]\label{LEMCenterImpliesTernaryAbsorption}
Suppose $B$ is a central subuniverse of $\mathbf A$, 
then $B\absorbseq_{3}\mathbf A$. 
\end{lem}

We say that an algebra $\mathbf A$ is \emph{BA and center free} 
if $\mathbf A$ has no nontrivial binary absorbing subuniverse and no
nontrivial central subuniverse.

\textbf{All types of subuniverses.} 
Suppose $\varnothing\neq \mathbf C\lneq \mathbf B\le \mathbf A$. We write 
\begin{itemize}
\item $C<_{\TBA}^{\mathbf A} \mathbf B$
if $C$ is a BA subuniverse of $\mathbf B$.
\item $C<_{\TC}^{\mathbf A} \mathbf B$
if $C$ is a central subuniverse of $\mathbf B$.
\item $C<_{\TD(\sigma)}^{\mathbf A} \mathbf B$
if there exists an irreducible  congruence $\sigma$ such that 
\begin{enumerate}
    \item $B^2\subseteq \sigma^{+}$;
    \item $C=B\cap E$ for some block $E$ of $\sigma$;
    \item $\mathbf B/\sigma$ is BA and center free.
\end{enumerate}
\item $C<_{\TL(\sigma)}^{\mathbf A} \mathbf B$ if $C<_{\TD(\sigma)}^{\mathbf A} B$  and 
$\sigma$ is linear.
\item $C<_{\TPC(\sigma)}^{\mathbf A} \mathbf B$ if $C<_{\TD(\sigma)}^{\mathbf A} \mathbf B$ and $\sigma$ is a PC congruence.
\item $C<_{\TS}^{\mathbf A} \mathbf B$
if there exists a BA and central (simultaneously) subuniverse $D$ in $\mathbf B$
such that $D\le \mathbf C$.
\end{itemize}

Notice that the above definition of types differs from the 
original one in \cite{zhuk2024simplified}, since  
we write $\sigma^{+}$ instead of $\sigma^{*}$;
that is, we use the minimal congruence above rather than 
the minimal compatible binary  relation above. 
Nevertheless, this does not affect the definition.
Indeed, if $\sigma^{*}\not\supseteq B^2$ 
while $\sigma^{+}\supseteq B^2$, then, by Lemma 51 in \cite{zhuk2024simplified}, 
$\sigma^{*}\cap B^2$ defines a linked binary relation 
on $\mathbf B/\sigma$ and yields either a BA or central subuniverse, which is forbidden
(see Theorem 3.11.1 in \cite{ZebsNotes} or Lemma 70 in \cite{zhuk2024simplified}).

If we do not want to specify the 
congruence 
we write $C<_{T}^{\mathbf A} \mathbf B$
instead of $C<_{T(\sigma)}^{\mathbf A} \mathbf B$. 
If $C<_{T}^{\mathbf A} \mathbf B$ then we say that $C$ is \emph{a subuniverse of $\mathbf B$ of type $T$}.
We write $C\le_{T(\sigma)}^{\mathbf A} \mathbf B$ to mean 
that either $C<_{T(\sigma)}^{\mathbf A} \mathbf B$ or $C=B$.
Often, instead of 
$C<_{\TBA}^{\mathbf A} \mathbf B$,
$C<_{\TC}^{\mathbf A} \mathbf B$,
and $C<_{\TS}^{\mathbf A} \mathbf B$
we write 
$C<_{\TBA} \mathbf B$,
$C<_{\TC} \mathbf B$,
and $C<_{\TS} \mathbf B$, 
respectively, since $\mathbf A$ is irrelevant to the definition.
Also, we write
$C<_{\TBA,\TC} \mathbf B$ to mean that both 
$C<_{\TBA} \mathbf B$ and $C<_{\TC} \mathbf B$ hold.




We write 
$C\lll^{\mathbf A} \mathbf B$ if there exist $B_{0},B_{1},\dots,B_{n}\subseteq B$
and $T_{1},\dots,T_{n}\in\{\TBA,\TC, \TS,\TD\}$ 
such that 
$\mathbf C=\mathbf B_{n}<_{T_{n}}^{\mathbf A} \mathbf B_{n-1}<_{T_{n-1}}^{\mathbf A}<
\dots <_{T_{2}}^{\mathbf A}<\mathbf B_1<_{T_{1}}^{\mathbf A} \mathbf B_{0} = \mathbf B$.
Note that $n$ may be 0 and 
the relation $\lll^{\mathbf A}$ is reflexive.
We do not allow empty subuniverses and 
the condition
$\varnothing\lll \mathbf A$ never holds.
Nevertheless, it is sometimes convenient to allow the empty set.
In this case we place a dot above the symbol and write 
$B\dot\lll \mathbf A$ to mean that either  
$B\lll \mathbf A$ or $B = \varnothing$.
With the same meaning we use dots in the notations 
$C\dot<_{T}^{\mathbf A}\mathbf B$ or
$C\dot\le_{T}^{\mathbf A}\mathbf B$.




\begin{lem}[Lemma 6.4 in \cite{zhuk2021strong}]\label{LEMCenterIntersection} 
Suppose 
$\mathbf C_{1}\le_{\TC}\mathbf A$ 
and 
$\mathbf C_{2}\le_{\TC}\mathbf A$.
Then 
$C_{1}\cap C_{2}\dot\le_{\TC}\mathbf A$.
\end{lem}

\begin{lem}[Lemma 6.5 in \cite{zhuk2021strong}]\label{LEMCenterIntersectSubalgebra} 
Suppose 
$\mathbf C\le_{\TC}\mathbf A$ 
and 
$\mathbf B\le\mathbf A$.
Then 
$C\cap B\dot\le_{\TC}\mathbf A$.
\end{lem}

\begin{lem}\label{LEMBAMainProperty}
Suppose 
$B<_{\TBA} \mathbf A$, 
$R\le_{sd} \mathbf A^{n}$. 
Then $R\cap B^{n}\neq\varnothing$
\end{lem}

\begin{proof}
To get a tuple from $R\cap B^{n}$ we apply a binary absorbing operation to 
all the tuples from $R$. 
\end{proof}

\begin{lem}[Proposition 2.4 in \cite{barto2012absorbing}]\label{LEMAbsOfAbs} 
Suppose $\mathbf C\absorbseq \mathbf B \absorbseq  \mathbf A$. 
Then $\mathbf C\absorbseq  \mathbf A$.
\end{lem}

\begin{lem}\label{LEMCenterOfCenter}  
Suppose $\mathbf C\le_{\TC} \mathbf B \le_{\TC} \mathbf A$. 
Then $\mathbf C\le_{\TC} \mathbf A$.
\end{lem}

\begin{proof}
By Lemma \ref{LEMAbsOfAbs} 
$\mathbf C\absorbseq \mathbf A$.
If $a\in A\setminus B$ then, 
since $B<_{\TC} \mathbf A$, we have
$$
(a,a)
\notin
\Sg_{\mathbf A}(
(\{a\}\times B)\cup (B\times \{a\}))
\supseteq 
\Sg_{\mathbf A}
(
(\{a\}\times C)\cup (C\times \{a\}))
$$
If $a\in B\setminus C$ then, 
since $C<_{\TC} \mathbf B$, we have 
$
(a,a)\notin
\Sg_{\mathbf A}
(
(\{a\}\times C)\cup (C\times \{a\}))$
\end{proof}

\begin{lem}[Corollary  6.9.2 in \cite{zhuk2021strong}]\label{LEMCenterImplies} 
Suppose $R\le \mathbf A_{1}\times\dots\times \mathbf A_{n}$,
$\proj_{1}(R) = A_{1}$, 
$C_{i}\le_{\TC} \mathbf A_{i}$ for every $i\in[n]$.
Then 
$\proj_{1}(R\cap (C_{1}\times\dots\times C_{N}))\dot\le_{\TC} \mathbf A_{1}$.
\end{lem}

\begin{lem}\label{LEMUbiquity} 
Suppose $B\lll \mathbf A$ and $|B|>1$.
Then there exists $C<_{T}^{\mathbf A} \mathbf B$, where $T\in \{\TBA, \TC, \TL, \TPC\}$.
\end{lem}

\begin{lem}[Lemma 20(i) in \cite{zhuk2024simplified}]\label{LEMIntersectALL}

Suppose $B\lll \mathbf A$, $D\lll \mathbf A$. 
Then $B\cap D\dot\lll \mathbf A$.
\end{lem}

\begin{lem}[Corollary 22 in \cite{zhuk2024simplified}]\label{LEMMainStableIntersection} 
 
Suppose 
\begin{enumerate}
    \item $R\le_{sd} \mathbf A_{1} \times \dots\times 
    \mathbf A_{n}$, $n\ge 2$;
    \item $C_{i}<_{T_{i}(\sigma_{i})}^{\mathbf A_{i}} \mathbf B_{i}\lll \mathbf A_{i}$, where 
    $T_{i}\in\{\TBA, \TC,\TS,\TL,\TPC\}$ for $i=1,2,\dots,n$;
    \item $R\cap (C_{1}\times\dots\times C_{n}) = \varnothing$;
    \item 
    $R\cap (C_{1}\times\dots\times C_{j-1}\times 
    B_{j}\times C_{j+1}\times\dots\times C_{n})  \neq  \varnothing$
    for every $j\in[n]$.
\end{enumerate}
Then one of the following conditions holds:
\begin{enumerate}
    \item[(ba)] $T_1=\dots=T_n=\TBA$;
    \item[(l)] 
        $T_1=\dots= T_n=\TL$ and 
        for every $k,\ell\in [n]$
        there exists a bridge $\delta$ from 
    $\sigma_{k}$ and $\sigma_{\ell}$ such that $\trace(\delta)=
    \sigma_{k}\circ \proj_{k,\ell}(R)\circ \sigma_{\ell}$;
    \item[(c)] $n=2$ and $T_1=T_2=\TC$;
    \item[(pc)] $n=2$, $T_1=T_2=\TPC$, $\mathbf A_{1}/\sigma_{1} \cong \mathbf A_{2}/\sigma_{2}$,
    and the relation $\{(a/\sigma_1,b/\sigma_{2})\mid (a,b)\in R\}$ is bijective. 
\end{enumerate}         
\end{lem}

\begin{lem}[Corollary 18(r1) in \cite{zhuk2024simplified}]\label{LEMPropagateToRelations}
 
Suppose $R\le_{sd} \mathbf A_{1}\times\dots\times \mathbf A_{n}$,
$B_{i}\lll \mathbf A_{i}$ for every $i\in[n]$. Then 
$\proj_{1}(R\cap (B_1\times\dots\times B_{n}))\dot\lll \mathbf A_{1}$.
\end{lem}

\subsection{Minimal Taylor algebras}\label{SUBSECTIONMinimalTaylor}





An algebra $\mathbf B$ is a \emph{reduct} of $\mathbf A$ if $A=B$ and $\Clo(\mathbf B) \subseteq \Clo(\mathbf A)$.
An algebra $\mathbf A$ is called a \emph{minimal Taylor algebra} if it is Taylor but no proper reduct of $\mathbf A$ is. 


\begin{lem}[Propositions 5.2, 5.4 in \cite{MinimalTaylorUnifying}]\label{LEMMainTaylorLemma}
Every Taylor algebra has a minimal Taylor reduct.
Any subalgebra, finite power, or quotient of a minimal Taylor algebra is again a minimal Taylor algebra.
\end{lem}


\begin{lem}[Theorem 5.10 in \cite{MinimalTaylorUnifying} , Theorem 6.15 in \cite{zhuk2021strong}]\label{LEMTernaryAbsorptionInTM}
Let $\mathbf A$ be a minimal Taylor algebra and $B \subseteq A$. 
Then the following conditions are equivalent.
\begin{enumerate}
    \item[(a)] $B \absorbseq_3 \mathbf A$.
    \item[(b)] $B \le_{\TC} \mathbf A$.
    \item[(c)] 
    The relation $R(x,y) = B(x) \vee B(y)$ is a subuniverse of $\mathbf A^2$.
\end{enumerate}
\end{lem}

For an algebra $\mathbf A$ 
by $\mathcal C(\mathbf A)$ we denote the set of all proper subuniverses 
$B\lneq \mathbf A$ such that 
the relation $R(x,y) = B(x) \vee B(y)$ is a subuniverse of $\mathbf A^2$.
If $\mathbf B$ is a reduct of $\mathbf A$, then 
$\mathcal C(\mathbf B)\supseteq \mathcal C(\mathbf A)$;
this property does not hold for central or ternary absorbing subuniverses, which is exactly why we use $\mathcal C(\mathbf A)$ instead.








\begin{lem}\label{LemProductOfSmallAlgebrasHasCyclic} 
Suppose $\mathbf A_1,\dots,\mathbf A_s$ are Taylor algebras 
such that $|A_{i}|<p$ for a prime $p$.
Then $\mathbf A_{1}\times\dots\times \mathbf A_{s}$ has a cyclic term operation of 
arity $p$.
\end{lem}

\begin{proof}
The proof is by induction on $s$. For $s=1$ it follows from
Theorem \ref{THMMinimalTaylorClassification}.
Let $w_1$ be a $p$-ary cyclic term in 
$\mathbf A_1,\dots,\mathbf A_{s-1}$,
 $w_2$ be a $p$-ary cyclic term in $\mathbf A_{s}$. 
Then the term 
$$w(x_1,\dots,x_{p})\approx w_{1}(w_2(x_1,\dots,x_{p}),w_2(x_2,\dots,x_{p},x_{1}),\dots,w_2(x_p,x_{1}\dots,x_{p-1}))$$
is a cyclic term for $\mathbf A_{1},\dots,\mathbf A_{s}$.
\end{proof}

\begin{lem}\label{LemFindTMinProduct} 
Suppose $\mathbf A_1,\dots,\mathbf A_s\in\mathcal V_{p}$, 
such that $p>|A_{i}|$ for every $i\in[s]$.
Then there exists a $p$-ary term 
$c$ such that 
$(A_i;c^{\mathbf A_{i}})$ is a 
minimal Taylor algebra and $c^{\mathbf A_{i}}$ is cyclic for every $i\in[s]$.
\end{lem}

\begin{proof}
Using Lemma \ref{LEMMainTaylorLemma}
choose a minimal Taylor reduct $\mathbf B$ of $\mathbf A_{1}\times \dots\times \mathbf A_{s}$.
Notice that $\mathbf B\cong \mathbf A_{1}'\times\dots\times \mathbf A_{s}'$, 
where each $\mathbf A_{i}'$ is a minimal Taylor reduct of $\mathbf A_{i}$.
By Lemma \ref{LemProductOfSmallAlgebrasHasCyclic}, $\mathbf B$ has a cyclic term of arity $p$.
Then the term defining $c$ from the basic operation $w$ is exactly the term we need.
\end{proof}




\subsection{Free generated relations and symmetric reductions}\label{SUBSECTIONFREEGENERATED}

\textbf{The free generated relations.} 
For algebras 
$\mathbf A_1,\dots,\mathbf A_{s}\in\mathcal V_{k}$
by 
$R^{n}_{\mathbf A_1,\dots,\mathbf A_{s}}$
we denote the following relation.
Coordinates of the relation 
are indexed by 
$(\mathbf A_{i},\alpha)$, where $\alpha\in A_{i}^{n}$ is 
a tuple generating $\mathbf A_{i}$.
We denote the set of all indices by $I^{n}_{\mathbf A_1,\dots,\mathbf A_{s}}$.
For $i\in[n]$
by $\gamma_{i}$ we denote the tuple whose elements are indexed by 
$I^{n}_{\mathbf A_1,\dots,\mathbf A_{s}}$ and 
whose $(\mathbf A_{i},\alpha)$-th element is equal to 
$\alpha(i)$ for every $(\mathbf A_{i},\alpha)\in I$.
Then $R^{n}_{\mathbf A_1,\dots,\mathbf A_{s}}$
is the subalgebra generated by 
$\gamma_1,\dots,\gamma_n$, 
and the tuples $\gamma_1,\dots,\gamma_n$ 
are called \emph{the generators} of $R^{n}_{\mathbf A_1,\dots,\mathbf A_{s}}$.
By $\mathcal R^{n}_{\mathbf A_1,\dots,\mathbf A_{s}}$ we denote the set of 
all relations $R$ whose coordinates are indexed by $I^{n}_{\mathbf A_1,\dots,\mathbf A_{s}}$, and 
the projection of $R$ onto the $(\mathbf A_{i},\alpha)$-th coordinate is 
$A_{i}$ for every $(\mathbf A_{i},\alpha)\in I^{n}_{\mathbf A_1,\dots,\mathbf A_{s}}$.

Similarly, for $k_1,\dots,k_{n}\in\mathbb N$ and algebras 
$\mathbf A_1,\dots,\mathbf A_{s}\in\mathcal V_{k}$
by 
$R^{\overline{k_1},\dots,\overline{k_n}}_{\mathbf A_1,\dots,\mathbf A_{s}}$
we denote the following relation.
Coordinates of the relation 
are indexed by 
$(\mathbf A_{i},\alpha)$, where $\alpha\in A_{i}^{k_1+\dots+k_n}$ is 
a $(\overline{k_1},\dots,\overline{k_n})$-palette tuple generating $\mathbf A_{i}$.
The set of all indices we denote by $I^{\overline{k_1},\dots,\overline{k_n}}_{\mathbf A_1,\dots,\mathbf A_{s}}$. 
For $i\in[k_1+\dots+k_n]$
by $\gamma_{i}$ we denote the tuple whose elements are indexed by 
$I^{\overline{k_1},\dots,\overline{k_n}}_{\mathbf A_1,\dots,\mathbf A_{s}}$ and 
whose $(\mathbf A_j;\alpha)$-th element is equal to 
$\alpha(i)$ for every $(\mathbf A_j;\alpha)\in I^{\overline{k_1},\dots,\overline{k_n}}_{\mathbf A_1,\dots,\mathbf A_{s}} $.
Then $R^{\overline{k_1},\dots,\overline{k_n}}_{\mathbf A_1,\dots,\mathbf A_{n}}$
is the subuniverse generated by 
$\gamma_1,\dots,\gamma_{k_1+\dots+k_n}$, 
and the tuples $\gamma_1,\dots,\gamma_{k_1+\dots+k_n}$ 
are called \emph{the generators} of $R^{\overline{k_1},\dots,\overline{k_n}}_{\mathbf A_1,\dots,\mathbf A_{n}}$.
By $\mathcal R^{\overline{k_1},\dots,\overline{k_n}}_{\mathbf A_1,\dots,\mathbf A_{s}}$ we denote the set of 
all relations $R$ 
whose coordinates are indexed by $I^{\overline{k_1},\dots,\overline{k_n}}_{\mathbf A_1,\dots,\mathbf A_{s}}$, and 
the projection of $R$ onto the $(\mathbf A_{i},\alpha)$-th coordinate is 
$A_{i}$ for every $(\mathbf A_{i},\alpha)\in I^{\overline{k_1},\dots,\overline{k_n}}_{\mathbf A_1,\dots,\mathbf A_{s}}$.
Note that we will use the same notations for $(\overline{k_1},\dots,\overline{k_n},\ell)$-palette tuples. 
For example, 
the coordinates of the free generated relation 
$R^{\overline{k_1},\dots,\overline{k_n},\ell}_{\mathbf A_1,\dots,\mathbf A_{s}}$
are indexed by 
$(\overline{k_1},\dots,\overline{k_n},\ell)$-palette tuples generating $\mathbf A_{i}$.

\textbf{Reductions}. 
To avoid repetitions, 
we define reductions for 
relations from $\mathcal R^{\overline{k_1},\dots,\overline{k_n}}_{\mathbf A_1,\dots,\mathbf A_{s}}$, 
assuming 
that 
the same definitions can be applied to 
$\mathcal R^{n}_{\mathbf A_1,\dots,\mathbf A_{s}}$.
In our proof we reduce a relation $R\in\mathcal R^{\overline{k_1},\dots,\overline{k_n}}_{\mathbf A_1,\dots,\mathbf A_{s}}$
by reducing its coordinates.
\emph{A reduction} $D^{(\top)}$ for $R\in\mathcal R^{\overline{k_1},\dots,\overline{k_n}}_{\mathbf A_1,\dots,\mathbf A_{s}}$ 
is a mapping that assigns a subuniverse 
$D_{(\mathbf A_{i},\alpha)}^{(\top)}\le \mathbf A_{i}$ to every $(\mathbf A_{i},\alpha)\in I^{\overline{k_1},\dots,\overline{k_n}}_{\mathbf A_1,\dots,\mathbf A_{s}}$.
We write 
$D^{(\bot)}\lll D^{(\top)}$ and $D^{(\bot)}\le_{T} D^{(\top)}$ 
whenever
$D_{i}^{(\bot)}\lll D_{i}^{(\top)}$ for every $i\in I^{\overline{k_1},\dots,\overline{k_n}}_{\mathbf A_1,\dots,\mathbf A_{s}}$ and 
$D_{i}^{(\bot)}\le_{T} D_{i}^{(\top)}$ for every $i\in I^{\overline{k_1},\dots,\overline{k_n}}_{\mathbf A_1,\dots,\mathbf A_{s}}$, respectively.
Notice that any reduction $D^{(\top)}$ can be viewed as a relation 
from $\mathcal R^{\overline{k_1},\dots,\overline{k_n}}_{\mathbf A_1,\dots,\mathbf A_{s}}$.
Then for any $R\in \mathcal R^{\overline{k_1},\dots,\overline{k_n}}_{\mathbf A_1,\dots,\mathbf A_{s}}$  
and a reduction $D^{(\bot)}$ 
by $R^{(\bot)}$ we denote 
$R\cap D^{(\bot)}$.
A reduction $D^{(\bot)}$ is called \emph{1-consistent} for $R\in \mathcal R^{\overline{k_1},\dots,\overline{k_n}}_{\mathbf A_1,\dots,\mathbf A_{s}}$
if 
$\proj_{i}(R^{(\bot)}) = D_{i}^{(\bot)}$ for every $i\in I^{\overline{k_1},\dots,\overline{k_n}}_{\mathbf A_1,\dots,\mathbf A_{s}}$.

\textbf{Permutations and symmetries.}
By $\Perm^{k_1,\dots,k_n}$ we denote the set of all 
permutations on 
$[k_1+\dots+k_n]$ such that 
the first $k_1$ elements remain in the first $k_1$ positions, 
the next $k_2$ elements remain in the next $k_2$ positions, and so on.
For a tuple $\alpha\in A^{k_1+\dots+k_n}$ denote 
$\Perm^{k_1,\dots,k_n}(\alpha) = \{\alpha^{\sigma}\mid \sigma\in \Perm^{k_1,\dots,k_n}\}$.
For an index
$i = (\mathbf A_{j},\alpha)$ we denote by $\Perm^{k_1,\dots,k_n}(i)$  the set of 
indices $(\mathbf A_{j},\beta)$ with $\beta\in \Perm^{k_1,\dots,k_n}(\alpha)$.
Let  $\gamma$ be a tuple whose coordinates 
are indexed by elements from $I^{\overline{k_1},\dots,\overline{k_n}}_{\mathbf A_1,\dots,\mathbf A_{s}}$, 
and let $\sigma\colon [k_1+\dots+k_n]\to [k_1+\dots+k_n]$ 
be a permutation.
By $\gamma^{\permd{\sigma}}$ we denote the
tuple $\gamma'$ defined by  
$\gamma'((\mathbf A_{i},\alpha)) = 
\gamma((\mathbf A_{i},\alpha^{\sigma}))$
for every $(\mathbf A_{i},\alpha)\in I^{\overline{k_1},\dots,\overline{k_n}}_{\mathbf A_1,\dots,\mathbf A_{s}}$. 
Similarly, for a relation $R\in\mathcal R^{\overline{k_1},\dots,\overline{k_n}}_{\mathbf A_1,\dots,\mathbf A_{s}}$
put $R^{\permd{\sigma}}=\{\gamma^{\permd{\sigma}}\mid \gamma\in R\}$. 
A relation $R$ is called \emph{$\sigma$-symmetric} if 
$R^{\permd{\sigma}} = R$.
A relation $R\in \mathcal R^{\overline{k_1},\dots,\overline{k_n}}_{\mathbf A_1,\dots,\mathbf A_{s}}$ is called 
\emph{symmetric} if it is $\sigma$-symmetric for 
every $\sigma \in \Perm^{k_1,\dots,k_n}$.
Similarly, a 
reduction $D^{(\top)}$ is called \emph{symmetric} if 
$D_{i}^{(\top)}=
D_{j}^{(\top)}$
for any $j\in\Perm^{k_1,\dots,k_n}(i)$.

We say that a finite idempotent algebra $\mathbf A$ has \emph{bounded width} if 
there is no linear congruence on any subalgebra of $\mathbf A$ (see \cite{BartoKozikLocalConsistency} for alternative definitions). In the theory of strong subalgebras 
this means that 
only types $\TBA, \TC,\TS,\TPC$ may appear.
We say that 
a term $t$ is  
\emph{local $(\overline{k_1},\dots,\overline{k_n})$-palette symmetric} in $\mathbf A$
if we require symmetry only 
on $(\overline{k_1},\dots,\overline{k_n})$-palette tuples generating 
$\mathbf A$.

\begin{lem}\label{LEMNoLinearCongruencesImpliesPS}
Suppose $\mathbf A_{1},\dots,\mathbf A_{s}$ are algebras such that 
each $\mathbf A_{i}$ is a bounded width algebra or
a BA and center free PC algebra,
$k_1,\dots,k_n\in\mathbb N$.
Then there exists a 
local $(\overline{k_1},\dots,\overline{k_n})$-palette symmetric term operation
for $\mathbf A_{1},\dots,\mathbf A_{s}$.
\end{lem}

\begin{proof}
First, notice that 
if $\mathbf A$ is a BA and center free PC algebra, 
then $B\lll A$ implies that $|B|=1$ or $B = A$, 
which guarantees that 
the linear case cannot come from PC algebras.
Let us consider an inclusion minimal symmetric 1-consistent reduction 
$D^{(\top)}$ for $R^{\overline{k_1},\dots,\overline{k_n}}_{\mathbf A_1,\dots,\mathbf A_{s}}$ 
such that 
$D_{(\mathbf A_{i},\alpha)}^{(\top)}\lll \mathbf A_{i}$ for every  
$(\mathbf A_{i},\alpha)\in I^{\overline{k_1},\dots,\overline{k_n}}_{\mathbf A_1,\dots,\mathbf A_{s}}$.
Such a reduction always exists as the trivial one satisfies the required properties. 
Consider two cases:

Case 1. $|D_{(\mathbf A_{i},\alpha)}^{(\top)}|=1$
for every  
$(\mathbf A_{i},\alpha)\in I^{\overline{k_1},\dots,\overline{k_n}}_{\mathbf A_1,\dots,\mathbf A_{s}}$. Then $D^{(\top)}$ viewed as a relation contains only one symmetric tuple.
Then there exists a term giving 
this tuple 
on the generators of $R^{\overline{k_1},\dots,\overline{k_n}}_{\mathbf A_1,\dots,\mathbf A_{s}}$.
Since the tuple is symmetric, the term operation is 
local $(\overline{k_1},\dots,\overline{k_n})$-PS 
for $\mathbf A_1,\dots,\mathbf A_s$.

Case 2. $|D_{(\mathbf A_{i},\alpha)}^{(\top)}|>1$
for some  
$(\mathbf A_{i},\alpha)\in I^{\overline{k_1},\dots,\overline{k_n}}_{\mathbf A_1,\dots,\mathbf A_{s}}$.
By Lemma \ref{LEMUbiquity} there exists 
$B<_{T} D_{(\mathbf A_{i},\alpha)}^{(\top)}$
for some 
$T\in\{\TBA,\TC,\TPC\}$.
Let us define a reduction $D^{(\triangle)}$ for
$R^{\overline{k_1},\dots,\overline{k_n}}_{\mathbf A_1,\dots,\mathbf A_s}$
by
$D^{(\triangle)}_{(\mathbf A_i, \beta)} = B$
for any $\beta \in \Perm^{k_1,\dots,k_n}(\alpha)$,
and by the corresponding $A_{j}$ on the remaining coordinates.
Since $\alpha$ is a palette tuple, 
$R^{\overline{k_1},\dots,\overline{k_n}}_{\mathbf A_1,\dots,\mathbf A_s}\cap D^{(\triangle)}\neq\varnothing$. 
If $R^{\overline{k_1},\dots,\overline{k_n}}_{\mathbf A_1,\dots,\mathbf A_s}\cap 
D^{(\top)}\cap D^{(\triangle)}=\varnothing$,
then, by Lemma \ref{LEMBAMainProperty}, $T\neq \TBA$, and 
by Lemma \ref{LEMMainStableIntersection} there should be 2 coordinates giving a contradiction.
This cannot happen as the reduction $D^{(\top)}$ is 1-consistent.

Then the reduction
$D^{(\bot)}$ defined by 
$D^{(\bot)}_{j}=\proj_{j}(R^{\overline{k_1},\dots,\overline{k_n}}_{\mathbf A_1,\dots,\mathbf A_{s}}\cap 
D^{(\top)}\cap D^{(\triangle)})$
is smaller than $D^{(\top)}$ and symmetric.
By Lemma \ref{LEMPropagateToRelations}, $D^{(\bot)} \lll D^{(\top)}$, 
which contradicts our assumption about the minimality. 
\end{proof}



\begin{cor}\label{CORNoLinearCongruencesImpliesPS}
Every finite bounded width algebra $\mathbf A$ has 
a $(\overline{k_1},\dots,\overline{k_n})$-palette symmetric term operation
for every $k_1,\dots,k_n\in\mathbb N$.
\end{cor}

\begin{proof}
Let $\mathbf A_{1},\dots,\mathbf A_{s}$ be the set of all subalgebras of $\mathbf A$
including $\mathbf A$.
Then Lemma \ref{LEMNoLinearCongruencesImpliesPS} 
gives a local $(\overline{k_1},\dots,\overline{k_n})$-palette symmetric term operation for $\mathbf A_{1},\dots,\mathbf A_{s}$, which is 
a $(\overline{k_1},\dots,\overline{k_n})$-palette symmetric term operation 
for $\mathbf A$.
\end{proof}
\begin{lem}\label{LEMPSInsideCenters}
Suppose 
$\mathbf A_1,\dots,\mathbf A_s\in\mathcal V_{k}$,
then for every $n,k_1,\dots,k_n,\ell\in\mathbb N$ there exists a term $t$ of arity $k_1+\dots+k_n+\ell$
such that for any $i\in[s]$ and any $(\overline{k_1},\dots,\overline{k_n},\ell)$-palette tuple
$(\mathbf a_1,\dots,\mathbf a_n,\mathbf b)$ generating 
$\mathbf A_{i}$: 
\begin{enumerate}
\item 
there exists a minimal central subuniverse
$B$ of $\mathbf A_{i}$  
such that $t^{\mathbf A}(\mathbf a_1^{\sigma_{1}},\dots,\mathbf a_n^{\sigma_{n}},\mathbf b^{\delta})\in B$
for any permutations $\sigma_1,\dots,\sigma_{n},\delta$ on $[k_1],\dots,[k_n]$, and $[\ell]$, respectively;
\item if $C<_{\TC} \mathbf A_{i}$ and 
$\mathbf b\in C^{\ell}$ then 
$t^{\mathbf A_{i}}(\mathbf a_1,\dots,\mathbf a_n,\mathbf b)\in C$.
\end{enumerate}
\end{lem}

\begin{proof}
Let us consider the free-generated relation
$R^{\overline{k_1},\dots,\overline{k_n},\ell}_{\mathbf A_1,\dots,\mathbf A_s}$.
For every $i\in [s]$ and every 
 tuple
$(\mathbf a_1,\dots,\mathbf a_n,\mathbf b)$
put $D^{(\top)}_{(\mathbf A_{i},(\mathbf a_1,\dots,\mathbf a_n,\mathbf b))} = B$,
where $B$ is the minimal central subuniverse of $\mathbf A_{i}$ containing all the elements of $\mathbf b$.
It exists because the intersection of two central subuniverses is central
(Lemma \ref{LEMCenterIntersection}).
Since $R^{\overline{k_1},\dots,\overline{k_n},\ell}_{\mathbf A_1,\dots,\mathbf A_s}\cap D^{(\top)}$ contains the last $\ell$ generators of 
$R^{\overline{k_1},\dots,\overline{k_n},\ell}_{\mathbf A_1,\dots,\mathbf A_s}$, 
it is not empty.
By Lemma \ref{LEMCenterImplies},
$\proj_{(\mathbf A_{i},\alpha)}(R^{\overline{k_1},\dots,\overline{k_n},\ell}_{\mathbf A_1,\dots,\mathbf A_{s}}\cap 
D^{(\top)})\le_{\TC}\mathbf A_{i}$
for every $(\mathbf A_{i},\alpha)\in I^{\overline{k_1},\dots,\overline{k_n},\ell}_{\mathbf A_1,\dots,\mathbf A_{s}}$.
Then, by the minimality of the central subuniverses 
the reduction $D^{(\top)}$ is 1-consistent.


Let us consider 
an inclusion minimal 1-consistent symmetric reduction 
$D^{(\bot)}\le _{\TC} D^{(\top)}$.
Such a reduction exists as the reduction 
$D^{(\top)}$ satisfies this property.
If 
$D^{(\bot)}_{(\mathbf A_i, \alpha)}$ is a minimal central subuniverse for every $(\mathbf A_i, \alpha)\in I^{\overline{k_1},\dots,\overline{k_n},\ell}_{\mathbf A_1,\dots,\mathbf A_s}$, 
then we are done.
Otherwise, 
fix $(\mathbf A_i, \alpha)\in I^{\overline{k_1},\dots,\overline{k_n},\ell}_{\mathbf A_1,\dots,\mathbf A_s}$ and choose 
$B<_{\TC} D^{(\bot)}_{(\mathbf A_i, \alpha)}$.
Let us define a reduction $D^{(\triangle)}$ for
$R^{\overline{k_1},\dots,\overline{k_n},\ell}_{\mathbf A_1,\dots,\mathbf A_s}$
by
$D^{(\triangle)}_{(\mathbf A_i, \beta)} = B$
for any $\beta \in \Perm^{k_1,\dots,k_n,\ell}(\alpha)$,
and by the corresponding $A_{j}$ on the remaining coordinates.
Since $\alpha$ is a palette tuple, 
$R^{\overline{k_1},\dots,\overline{k_n},\ell}_{\mathbf A_1,\dots,\mathbf A_s}\cap D^{(\triangle)}\neq\varnothing$. 
If $R^{\overline{k_1},\dots,\overline{k_n},\ell}_{\mathbf A_1,\dots,\mathbf A_s}\cap 
D^{(\bot)}\cap D^{(\triangle)}=\varnothing$,
then by Lemma \ref{LEMMainStableIntersection} there should be 2 coordinates giving a contradiction.
This cannot happen as the reduction $D^{(\bot)}$ is 1-consistent.

Then the reduction 
$D^{(\star)}$ defined by 
$D^{(\star)}_{j}=\proj_{j}(R^{\overline{k_1},\dots,\overline{k_n},\ell}_{\mathbf A_1,\dots,\mathbf A_{s}}\cap 
D^{(\bot)}\cap D^{(\triangle)})$
is smaller than $D^{(\bot)}$ and symmetric.
By Lemma \ref{LEMCenterImplies}, 
$D^{(\star)}\le_{\TC}D^{(\bot)}$, 
by Lemma \ref{LEMCenterOfCenter} 
$D^{(\star)}\le_{\TC}D^{(\top)}$, which contradicts our assumption about the minimality of $D^{(\bot)}$.
\end{proof}







\subsection{Squeezing operations}\label{SUBSECTIONSqueezingOperations}

For two terms $t_1$ and $t_2$ over the variables
$x_1,\dots,x_{n+1}$ 
we denote the term 
$t_1(x_1,\dots,x_n, t_2(x_1,\dots,x_n,x_{n+1}))$
by $t_1\circ t_2$.
Similarly, 
by $t^{m}$ we denote 
$\underbrace{t\circ t\circ\dots\circ t}_{m}$.
We say that a set $\mathcal F$ of terms over the variables 
$x_1,\dots,x_{n+1}$ is \emph{closed under right composition} 
if 
$f_1\circ f_2\in \mathcal F$ for all 
$f_1,f_2\in \mathcal F$.



Let $\mathcal F$ be a set of terms 
over the variables 
$x_1,\dots,x_{n+1}$.
By $\squeezer(\mathcal F,\mathbf A_{1},\dots,\mathbf A_{s})$ we denote 
the set of all terms $f\in \mathcal F$ 
such that 
\begin{enumerate}
\item 
for any term $t\in \mathcal F$, $j\in[s]$, and $a_{1},\dots,a_n\in A_{j}$
    $$|\{f^{\mathbf A_j}(a_1,\dots,a_n,b)\mid b\in A_{j}\}|
    \le |\{t^{\mathbf A_j}(a_1,\dots,a_n,b)\mid b\in A_{j}\}|;$$ 
\item for all $j\in[s]$ we have $f^{\mathbf A_j}\circ f^{\mathbf A_j} = f^{\mathbf A_j}$.
\end{enumerate}

Thus, $\squeezer(\mathcal F,\mathbf A_1,\dots,\mathbf A_{s})$
consists of terms from $\mathcal F$ such that the corresponding images are the smallest ones among all such terms.

We say that a family of terms $\mathcal F$
is \emph{good for algebras $\mathbf A_1,\dots,\mathbf A_{s}\in\mathcal V_{k}$}
if 
for any $f\in \mathcal F$, any $i\in [s]$, any tuples 
$(a_1,\dots,a_n)$ and 
$(b_1,\dots,b_m)$ with $m\le n$ generating $\mathbf A_{i}$,
there exists 
$t\in \mathcal F$
such that 
$t(b_1,\dots,b_{m},y_1,\dots,y_{n-m},x)\approx 
f(a_1,\dots,a_{n},x)$.





\begin{lem}\label{LEMSquizProperties}
Suppose 
$\mathbf A_{1},\dots,\mathbf A_s\in\mathcal V_{k}$,
$\mathcal F$ is a nonempty set of terms over the variables $x_1,\dots,x_{n+1}$
closed under right composition.
Then \begin{enumerate}
    \item[(1)] $\squeezer(\mathcal F,\mathbf A_{1},\dots,\mathbf A_{s'})\supseteq 
    \squeezer(\mathcal F,\mathbf A_{1},\dots,\mathbf A_s)$ if $s'\le s$;
    \item[(2)] $\squeezer(\mathcal F,\mathbf A_{1},\dots,\mathbf A_s)=
    \squeezer(\mathcal F,\mathbf A_1,\dots,\mathbf A_s,\mathbf A_i/\sigma,\mathbf B )$
    for any congruence $\sigma$  on $\mathbf A_i$
    and subalgebra $\mathbf B\le \mathbf A_i$;
    \item[(3)] 
    $(t_{1}\circ t_2\circ t_3)^{N!}\in \squeezer(\mathcal F,\mathbf A_{1},\dots,\mathbf A_s)$ 
    for any $t_1, t_3\in\mathcal F$
    and $t_2\in \squeezer(\mathcal F,\mathbf A_{1},\dots,\mathbf A_s)$,
     where $N = \max\{|A_{1}|,\dots,|A_s|\}$;
    \item[(4)] $\squeezer(\mathcal F,\mathbf A_{1},\dots,\mathbf A_s)\neq\varnothing$.
\end{enumerate}
\end{lem}

\begin{proof}
Properties (1) and (2) follow immediately from the definition.
To prove property (3) notice that 
$(t_{1}\circ t_2\circ t_3)$ is from $\mathcal F$
and 
satisfies property 1 of squeezers. 
Property 2 is guaranteed by the $N!$-th power.

Let us prove that 
$\squeezer(\mathcal F,\mathbf A_{1},\dots,\mathbf A_s)$ is not empty
(property (4)).
Since we can compose terms, 
we can derive a term $t$ 
such that 
the set 
    $\{t^{\mathbf A_j}(a_1,\dots,a_n,b)\mid b\in A_{j}\}$
    is of minimal size for 
    all $j\in[s]$ and $a_1,\dots,a_n\in A_j$. 
It remains to obtain property 2 of squeezers by taking the 
$N!$-th power of the term, where $N = \max\{|A_{1}|,\dots,|A_s|\}$.
\end{proof}

\begin{lem}\label{LEMSqueezForGenerating}
Suppose 
$\mathbf A\in\mathcal V_{k}$,
$\mathcal F$ is a set of terms over the variables 
$x_1,\dots,x_{n+1}$ closed under right composition,
$\mathcal F$ is good for $\mathbf A$,
$\Sg_{\mathbf A}(\{a_1,\dots,a_n\}) 
= \Sg_{\mathbf A}(\{b_1,\dots,b_n\}) = A$.
Then for all $f\in \squeezer(\mathcal F,\mathbf A)$
and $t\in \mathcal F$ 
we have 
$|\{f(a_1,\dots,a_n,c)\mid c\in A\}|\le 
|\{t(b_1,\dots,b_n,c)\mid c\in A\}|$.
\end{lem}

\begin{proof}
Since $\mathcal F$ is good for $\mathbf A$, 
there exists $t'\in\mathcal F$ such that 
$t'(a_1,\dots,a_n,x)\approx t(b_1,\dots,b_n,x)$.
Then the required condition follows from 
the definition of a squeezer.
\end{proof}

\begin{lem}\label{LEMfgh}
Suppose 
$f,g_1,g_2,h$ are unary mappings on a finite set $A$
such that 
$|\Image(h)|=|\Image(f)|=|\Image(f\circ g_1\circ h)|=
|\Image(f\circ g_2\circ h)| = k$.
Then 
$(f\circ g_1\circ h)^{|A|!} = (f\circ g_2\circ h)^{|A|!}$. 
\end{lem}

\begin{proof}
Let us denote the left-hand and the right-hand sides of the equality by 
$f_1$ and $f_2$ respectively.
As $|A|!$ is divisible by the size of any cycle in $f\circ g_1\circ h$ and $f\circ g_2\circ h$, we have $f_1\circ f_1 = f_1$ and $f_2 \circ f_2 = f_2$.
Then the equality follows from 
$\Image(f_1) = \Image(f_2) = \Image(f)$
and $f_1(a) = f_1(b)\Leftrightarrow
h(a) = h(b)\Leftrightarrow
f_2(a) = f_2(b)$ for all $a,b\in A$.
\end{proof}

\begin{lem}\label{LEMGoodSqueezer}
Suppose 
$\mathbf A_{1},\dots,\mathbf A_s\in\mathcal V_{k}$,
$\mathcal F$ is a set of terms over the variables 
$x_1,\dots,x_{n+1}$ closed under right composition,
$\mathcal F$ is good for $\mathbf A_{1}$, 
$m\le n$,  
$g\in \squeezer(\mathcal F,\mathbf A_{1},\dots,\mathbf A_s)$,
$\Sg_{\mathbf A_{1}}(\{a_1,\dots,a_n\}) 
= \Sg_{\mathbf A_{1}}(\{b_1,\dots,b_m\}) = A_{1}$.
Then there exists a term 
$f\in \squeezer(\mathcal F,\mathbf A_{1},\dots,\mathbf A_s)$
such that 
$$f^{\mathbf A_{1}}(b_1,\dots,b_m,y_1,\dots,y_{n-m},x) \approx 
g^{\mathbf A_{1}}(a_1,\dots,a_n,x).$$
\end{lem}

\begin{proof}
Since $\mathcal F$ is good for $\mathbf A_{1}$, 
there exists 
$g'\in \mathcal F$ such that 
$g'(b_1,\dots,b_{m},y_1,\dots,y_{n-m},x) \approx g(a_1,\dots,a_n,x)$.
Choose any term $f'\in\squeezer(\mathcal F, \mathbf A_{1},\dots,\mathbf A_s).$ Put
$f = (g'\circ f'\circ g')^{S!}$, where 
$S$ is the maximal size of an algebra in $\mathbf A_1,\dots,\mathbf A_s$.
By Lemma \ref{LEMSquizProperties}(3), 
$f\in \squeezer(\mathcal F, \mathbf A_{1},\dots,\mathbf A_s)$.
Choose some
$c_1,\dots,c_{n-m}\in A_{j}$.
Denote 
\begin{align*}
g_{0}'(x)&:= g'(b_1,\dots,b_m,c_1,\dots,c_{n-m},x) = 
g(a_1,\dots,a_n,x), \\
f_{0}'(x)&:= f'(b_1,\dots,b_m,c_1,\dots,c_{n-m},x).
\end{align*} 
By Lemma \ref{LEMSqueezForGenerating}, we have 
$|\Image(g_{0}')| = |\Image(f_{0}')|=|\Image(g_{0}'\circ f_{0}'\circ g_{0}')|$.
By property 2 of squeezers, we obtain
$(g_{0}'\circ g_{0}'\circ g_{0}')^{S!} = g_{0}'$.
Then Lemma \ref{LEMfgh} implies 
$(g_{0}'\circ f_{0}'\circ g_{0}')^{S!} = (g_{0}'\circ g_{0}'\circ g_{0}')^{S!}= g_{0}'$, which is exactly the equality we need. 
\end{proof}

\begin{lem}\label{LEMMagicListOfTerms}
Suppose 
$\mathbf A_{1},\dots,\mathbf A_s\in\mathcal V_{k}$,
$\mathcal F$ is a set of terms over the variables 
$x_1,\dots,x_{n+1}$ closed under right composition,
$\mathcal F$ is good for $\mathbf A_{1},\dots,\mathbf A_{s}$. 
Then there exists a finite set $\mathcal F_{0}\subseteq \squeezer(\mathcal F,\mathbf A_{1},\dots,\mathbf A_{s})$
such that 
for every $j\in [s]$, 
every
$(\overline{k_1},\dots,\overline{k_n})$-palette tuple 
$(\mathbf a_1,\dots,\mathbf a_n)\in A_{j}^{k_1+\dots+k_n}$
generating $\mathbf A_{j}$, 
and every mapping $h:A_{j}\to A_{j}$
the following conditions are equivalent:
\begin{enumerate}
\item[(1)] $\exists f\in \mathcal F_0 \;\exists i_{1}\in[k_1]\dots\exists i_n\in[k_n]\colon f^{\mathbf A_{j}}(\mathbf a_{1}^{i_1},\dots,\mathbf a_{n}^{i_n},x)\approx h(x)$;
\item[(2)] $\exists f\in \mathcal F_0\; 
\forall i_{1}\in[k_1]\dots\forall i_n\in[k_n]\colon  f^{\mathbf A_{j}}(\mathbf a_{1}^{i_1},\dots,\mathbf a_{n}^{i_n},x)\approx h(x)$.
\end{enumerate}

\end{lem}

\begin{proof}
For every $j\in[s]$ we build a set $\mathcal F_{j}$ as follows.
For every $m\le n$,
every $g\in\squeezer(\mathcal F, \mathbf A_{1},\dots,\mathbf A_s)$, and every $a_1,\dots,a_n,b_1,\dots,b_m\in A_{j}$ 
such that 
$\Sg_{\mathbf A_{j}}(\{a_1,\dots,a_n\}) 
= \Sg_{\mathbf A_{j}}(\{b_1,\dots,b_m\})= A_{j}$
we add to $\mathcal F_{j}$ a term $f\in\squeezer(\mathcal F, \mathbf A_{1},\dots,\mathbf A_s)$ 
such that 
for all $c_1,\dots,c_{n-m}\in A_{j}$
$$f^{\mathbf A_{j}}(b_1,\dots,b_m,c_1,\dots,c_{n-m},x) \approx 
g^{\mathbf A_{j}}(a_1,\dots,a_n,x),$$
which exists by Lemma \ref{LEMGoodSqueezer}. Let $S_{n}$ be the set of all permutations on $[n]$.  Put
$$\mathcal F_0 = \{f\mid \exists f'\in \mathcal F_1\cup\dots\cup \mathcal F_s,  \sigma\in S_n, 
f(x_1,\dots,x_n,z)\approx f'(x_{\sigma(1)},\dots,x_{\sigma(n)},z)\}.$$
Let us show that $\mathcal F_0$ satisfies the required properties.
Assume that condition (1) is satisfied, that is, for some $f\in \mathcal F_0$ and $i_1,\dots,i_n$ we have 
$f^{\mathbf A_{j}}(\mathbf a_{1}^{i_1},\dots,\mathbf a_{n}^{i_n},x)\approx h(x)$.
Since $(\mathbf a_1,\dots,\mathbf a_n)$ is a $(\overline{k_1},\dots,\overline{k_n})$-palette tuple, 
there should be constant blocks $\mathbf a_{j_1}, \dots, \mathbf a_{j_m}$ 
such that 
elements of $\mathbf a_{j_1}, \dots, \mathbf a_{j_m}$ generate 
$A_{j}$.
Let each $\mathbf a_{j_\ell}$ consist of an element $b_{\ell}$ for every $\ell\in[m]$. 
 Then the required term 
satisfying condition (2) 
comes from the term chosen for $\mathbf a_{1}^{i_1},\dots,\mathbf a_{n}^{i_n}$, 
$b_1,\dots,b_{m}$ and a permutation $\sigma$ that maps 
$\ell$ to $j_\ell$  for every $\ell\in[m]$. 
\end{proof}

We introduce an order on finite sets of finite algebras 
by comparing lexicographically the sizes of algebras. 
Formally, 
we say that 
$\{\mathbf A_1,\dots,\mathbf A_s\}< \{\mathbf B_1,\dots,\mathbf B_m\}$ if 
there exists $n_0$
such that 
\begin{enumerate}
    \item for every $n>n_0$ 
    both $\{\mathbf B_1,\dots,\mathbf B_m\}$
    and $\{\mathbf A_1,\dots,\mathbf A_s\}$
    have the same number of algebras of size $n$;
\item     the set $\{\mathbf B_1,\dots,\mathbf B_m\}$
    has more algebras of size $n_0$ than
    $\{\mathbf A_1,\dots,\mathbf A_s\}$.
\end{enumerate}

\subsection{The existence of good sets of terms}\label{SUBSECTIONGoodSetsOfTerms}

In this section we show how to build good sets of terms that will help us to squeeze the domains later. The idea is to choose terms that preserve a congruence or unary relations in the last coordinate. For the linear and central case it is enough to consider all such terms 
to get a good family satisfying all the necessary properties. 
However, for the PC case we have to use a more involved definition, and this is precisely the reason why we introduce good sets of terms.

\begin{lem}\label{LEMPreservingSetsIsGood}
Suppose $\mathbf A_1,\dots,\mathbf A_{s}\in\mathcal V_{k}$, 
$\mathcal B_{i}$ is a set of subsets of $\mathbf A_{i}$ for every $i$, 
$\mathcal T$ is the set of all terms $t$ 
over the variables $x_1,\dots,x_{n+1}$ such that 
$t^{\mathbf A_{i}}(a_1,\dots,a_{n+1}) \in B$
for any $i\in[s]$, $B\in\mathcal B_{i}$, $a_1,\dots,a_n\in A_{i}$, $a_{n+1}\in B$.
Then 
\begin{itemize}
    \item[(n)] $\mathcal T$ is not empty;
    \item[(c)] $\mathcal T$ is closed under right composition;
    \item[(g)] $\mathcal T$ is good for $\mathbf A_1,\dots,\mathbf A_{s}$;
    \item[(s)] for every $i\in[s]$ one of the following conditions holds: \begin{enumerate}
        \item there exists 
    $t\in\mathcal T$ and a tuple 
    $(a_1,\dots,a_n)$ generating $\mathbf A_{i}$ such that 
    $|\{t^{\mathbf A_{i}}(a_1,\dots,a_n,b)\mid b\in A_{i}\}|<|A_{i}|$;
    \item there is no tuple $(a_1,\dots,a_n)$ generating $\mathbf A_{i}$;
    \item 
    $\{t^{\mathbf A_{i}}(a_1,\dots,a_n,b)\mid b\in A_{i}\}=A_{i}$ for every $t\in \mathcal T$ and $a_1,\dots,a_n\in A_{i}$.
    \end{enumerate}    
\end{itemize}
\end{lem}
\begin{proof}
$\mathcal T$ is nonempty because the projection onto the variable $x_{n+1}$ belongs to $\mathcal T$.
Property (c) follows immediately from the definition. 
To prove property (g) for $i\in[s]$, tuples 
$(a_1,\dots,a_n)$ and $(b_1,\dots,b_m)$ generating $\mathbf A_{i}$, 
and $f\in \mathcal T$, 
we can define a required term as follows. 
For every $j\in[n]$
choose a term $t_{j}$ such that 
$t_{j}^{\mathbf A_{i}}(b_1,\dots,b_{m}) = a_{j}$
Then the term $f(t_1(x_1,\dots,x_{m}),\dots,t_{n}(x_1,\dots,x_{m}), x_{n+1})$
with dummy variables $x_{m+1},\dots,x_{n}$ belongs to $\mathcal T$ and satisfies the required condition.

Let us prove property (s).
If conditions 2 or 3 are satisfied then we are done. 
Otherwise, consider a tuple 
$(b_1,\dots,b_n)$ generating $\mathbf A_{i}$, a 
tuple 
$(a_1,\dots,a_n)$, and a term $t\in\mathcal T$ such that 
$\{t^{\mathbf A_{i}}(a_1,\dots,a_n,b)\mid b\in A_{i}\}\neq A_{i}$.
Again, for every $j\in[n]$
choose a term $t_{j}$ such that 
$t_{j}^{\mathbf A_{i}}(b_1,\dots,b_{n}) = a_{j}$.
Then the term $t'(x_1,\dots,x_{n+1}) \approx 
t(t_1(x_1,\dots,x_{n}),\dots,t_n(x_1,\dots,x_{n}),x_{n+1})$ 
and a tuple $(b_1,\dots,b_n)$ witness condition 1.
\end{proof}


\begin{cor}\label{CORLinearTerms}
Suppose $\mathbf A_1,\dots,\mathbf A_{s}\in\mathcal V_{k}$, 
$\sigma$ is a maximal linear congruence on $\mathbf A_{1}$,
$\mathcal T$ is the set of all terms $t$ 
over the variables $x_1,\dots,x_{n+1}$ such that 
$t^{\mathbf A_{1}/\sigma}(x_1,\dots,\dots,x_{n+1}) \approx x_{n+1}$. 
Then $\mathcal T$ satisfies
properties (n),(c),(g), and (s) from Lemma \ref{LEMPreservingSetsIsGood}.
\end{cor}
\begin{proof}
 It is sufficient to set $\mathcal B_{1}$ to be the set of all  equivalence classes of $\sigma$, $\mathcal B_{i} = \varnothing$ for 
 $i\ge 2$,
 and apply Lemma \ref{LEMPreservingSetsIsGood}.
\end{proof}


\begin{cor}\label{CORCentralTerms}
Suppose $\mathbf A_1,\dots,\mathbf A_{s}\in\mathcal V_{k}$, 
$\mathcal T$ is the set of all 
terms $t$ over the variables 
$x_1,\dots,x_{n+1}$ such that 
$t^{\mathbf A_{i}}(a_1,\dots,a_{n+1})\in B$ 
for any $i\in[s]$, $B<_{\TC}\mathbf A_{i}$, $a_1,\dots,a_{n}\in A_{i}$, 
$a_{n+1}\in B$.
Then $\mathcal T$ satisfies
properties (n),(c),(g), and (s) from Lemma \ref{LEMPreservingSetsIsGood}.
\end{cor}

\begin{proof}
 It is sufficient to set $\mathcal B_{i}$ to be the set of all proper central subuniverses of $\mathbf A_{i}$ for every $i\in[s]$ and apply Lemma \ref{LEMPreservingSetsIsGood}.
\end{proof}

A term $t$ is called \emph{strongly stable} for an algebra $\mathbf A$ 
if 
$\{t^{\mathbf A}(a_1,\dots,a_n)\}\lll \mathbf A$ for any tuple $(a_1,\dots,a_n)$
generating $\mathbf A$.


\begin{lem}\label{LEMPCTerms}
Suppose $\mathbf A_1,\dots,\mathbf A_{s}\in\mathcal V_{k}$ are BA and center free algebras, 
$\mathcal T$ is the set of all 
terms $t$ over the variables 
$x_1,\dots,x_{n+1}$ such that 
\begin{itemize}
    \item $t$ is strongly stable in $\mathbf A_{1},\dots,\mathbf A_{s}$;
    \item for every $i\in[s]$, every maximal PC congruence $\sigma$ on $\mathbf A_{i}$, 
    and every tuple $(a_1,\dots,a_n)$ generating $\mathbf A_{i}$ 
    we have $t^{\mathbf A_{i}/\sigma}(a_1/\sigma,\dots,a_n/\sigma,x)\approx x$.
\end{itemize}
Then 
\begin{itemize}
    \item[(n)] $\mathcal T$ is not empty;
    \item[(c)] $\mathcal T$ is closed under right composition;
    \item[(g)] $\mathcal T$ is good for $\mathbf A_1,\dots,\mathbf A_{s}$;
    \item[(s)] for every $i\in[s]$ one of the following conditions holds: \begin{enumerate}
        \item there exists 
    $t\in\mathcal T$ and a tuple 
    $(a_1,\dots,a_n)$ generating $\mathbf A_{i}$ such that 
    $|\{t^{\mathbf A_{i}}(a_1,\dots,a_n,b)\mid b\in A_{i}\}|<|A_{i}|$;
    \item there is no tuple $(a_1,\dots,a_n)$ generating $\mathbf A_{i}$;
    \item $\mathbf A_{i}\cong \mathbf B_1\times\dots\times \mathbf B_{\ell}$ 
    for some PC algebras $\mathbf B_1,\dots,\mathbf B_{\ell}$.
    \end{enumerate}    
\end{itemize}
\end{lem}
\begin{proof}

(n). Consider the free generated relation $R^{n+1}_{\mathbf A_1,\dots,\mathbf A_{s}}$.
Let $\sigma_{i}$ be the intersection of all the maximal PC
congruences on $\mathbf A_{i}$ for each $i\in[s]$.
Let $D^{(\top)}$ be the reduction defined 
by $D^{(\top)}_{(\mathbf A_{i},(a_1,\dots,a_{n+1}))} = a_{n+1}/\sigma_{i}$,
where
$a_{n+1}/\sigma_{i}$ is the equivalence class of $\sigma_{i}$ containing $a_{n+1}$.
By Lemma \ref{LEMIntersectALL}, 
$D^{(\top)}_{(\mathbf A_{i},(a_1,\dots,a_{n+1}))}\lll \mathbf A_{i}$.
Note that 
$R^{n+1}_{\mathbf A_1,\dots,\mathbf A_{s}}\cap D^{(\top)}$ is not empty because the intersection contains the $(n+1)$-th generator.
Using Lemmas \ref{LEMUbiquity} and \ref{LEMPropagateToRelations}, we can find a reduction 
$D^{(\bot)}\lll D^{(\top)}$ 
such that 
$|R^{n+1}_{\mathbf A_1,\dots,\mathbf A_{s}}\cap D^{(\bot)}|=1$.
The term defining the only tuple in this intersection belongs to $\mathcal T$ and witnesses that it is not empty.

(c). Checking this property is straightforward.

(g). Let $f\in\mathcal T$, $i\in[s]$, 
$(a_1,\dots,a_n)$ and $(b_1,\dots,b_{m})$ generate $\mathbf A_{i}$. 
Choose terms $t_1,\dots,t_j$ such that 
$t_j^{\mathbf A_{i}}(b_1,\dots,b_{m}) = a_{j}$ for every $j\in [n]$.
Let $t'(x_1,\dots,x_{n}) \approx 
f(t_1(x_1,\dots,x_{m}),\dots,t_n(x_1,\dots,x_{m}),x_{n+1})$, 
where variables $x_{m+1},\dots,x_{n}$ are dummy.
To define the required term consider the free generated relation $R^{n+1}_{\mathbf A_1,\dots,\mathbf A_{s}}$.
Let $D^{(\top)}$ be the reduction defined 
by $D^{(\top)}_{(\mathbf A_{j},(c_1,\dots,c_{n+1}))} = c_{n+1}/\sigma_{j}$
for any $j\in[s]$ and $c_1,\dots,c_{n+1}\in A_{j}$.
Since the $(n+1)$-th generator of $R^{n+1}_{\mathbf A_1,\dots,\mathbf A_{s}}$ is in $D^{(\top)}$,
$R^{n+1}_{\mathbf A_1,\dots,\mathbf A_{s}}\cap D^{(\top)}\neq \varnothing$.
Let $D^{(\triangle)}$
be the reduction defined 
by $D^{(\triangle)}_{(\mathbf A_{i},(b_1,\dots,b_{n+1}))}
= \{f(a_1,\dots,a_n,b_{n+1})\}$ for any $b_{m+1},\dots,b_{n+1}\in A_i$.
Since $f$ is strongly stable in $\mathbf A_{i}$, 
$D^{(\triangle)}_{(\mathbf A_{i},(b_1,\dots,b_{n+1}))}\lll \mathbf A_{i}$. 
Since $t'$ applied to $R^{n+1}_{\mathbf A_1,\dots,\mathbf A_{s}}$
gives a tuple from $D^{(\triangle)}$, 
the intersection $R^{n+1}_{\mathbf A_1,\dots,\mathbf A_{s}}\cap D^{(\triangle)}$ is not empty. 
Assume that 
$R^{n+1}_{\mathbf A_1,\dots,\mathbf A_{s}}\cap D^{(\top)}\cap D^{(\triangle)}=\varnothing$.
By Lemma \ref{LEMMainStableIntersection},
there exists 
an equivalence class $C$ of a PC congruence on $\mathbf A_{i}$ and $(c_1,\dots,c_{n+1})$ such that 
$c_{n+1}\in C$ and 
$R^{n+1}_{\mathbf A_1,\dots,\mathbf A_{s}}\cap D^{(\top)}$ has no tuples 
whose $(\mathbf A_{i},(c_1,\dots,c_{n+1}))$-coordinate is
from $C$. This cannot happen as 
$D^{(\top)}_{(\mathbf A_{i},(c_1,\dots,c_{n+1}))}\subseteq C$.
Thus, $R^{n+1}_{\mathbf A_1,\dots,\mathbf A_{s}}\cap D^{(\top)}\cap D^{(\triangle)}\neq\varnothing$.
It remains to use Lemmas \ref{LEMUbiquity} and \ref{LEMPropagateToRelations}, to find a reduction 
$D^{(\bot)}\lll D^{(\top)}\cap D^{(\triangle)}$ 
such that 
$|R^{n+1}_{\mathbf A_1,\dots,\mathbf A_{s}}\cap D^{(\bot)}|=1$.
The term defining the only tuple in this intersection satisfies the required conditions.

(s). If conditions 2 or 3 are satisfied then we are done. 
Otherwise, choose a tuple $(a_1,\dots,a_n)$ generating $\mathbf A_{i}$ and consider the free generated relation 
$R^{n+1}_{\mathbf A_1,\dots,\mathbf A_{s}}$. 
As before, we define a reduction $D^{(\top)}$ 
by $D^{(\top)}_{(\mathbf A_{j},(c_1,\dots,c_{n+1}))} = c_{n+1}/\sigma_{j}$
for any $j\in[s]$ and $c_1,\dots,c_{n+1}\in A_{j}$.
Since $\mathbf A_{i}$ is not a product of PC algebras, 
$\sigma_{i}\neq 0_{\mathbf A_{i}}$.
Choose an equivalence class $C$ of $\sigma_{i}$
containing more than one element. 
Using Lemma \ref{LEMUbiquity} choose an element $e$ such 
that $\{e\}\lll^{\mathbf A_{i}} C$.
Let us define a reduction $D^{(\triangle)}$ 
by $D^{(\triangle)}_{(\mathbf A_{i},(a_1,\dots,a_{n},d))}=\{e\}$
for any $d\in C$.
Notice that 
$R^{n+1}_{\mathbf A_1,\dots,\mathbf A_{s}}\cap D^{(\triangle)}$ is not empty 
because 
the elements $a_1,\dots,a_n$ generate the whole $\mathbf A_{i}$ including the element $e$.
Let us show that $R^{n+1}_{\mathbf A_1,\dots,\mathbf A_{s}}\cap D^{(\top)}\cap D^{(\triangle)}$ is also not empty.
This follows from Lemma \ref{LEMMainStableIntersection}
and the fact that there is no maximal PC congruence that cuts the class $C$.
Again, it remains to use Lemmas \ref{LEMUbiquity} and \ref{LEMPropagateToRelations} to find a reduction 
$D^{(\bot)}\lll D^{(\top)}\cap D^{(\triangle)}$ 
such that 
$|R^{n+1}_{\mathbf A_1,\dots,\mathbf A_{s}}\cap D^{(\bot)}|=1$.
The term defining the only tuple in this intersection satisfies the required conditions.
\end{proof}

\subsection{Divide and conquer}\label{SUBSECTIONDivideAndConquer}

First, we will show that the families of terms we defined earlier actually squeeze the domains 
unless the algebras satisfy very strong property. 
Second, we prove a technical claim that allows us to combine freely a term operation on one PC 
algebra and a strongly stable term on another algebra.

\begin{lem}\label{LEMCenterImpliesSqueezing}
Suppose $\mathbf A_1,\dots,\mathbf A_s\in\mathcal V_{k}$,
$j\in[s]$, 
$\{c\}$ is not a central subuniverse of $\mathbf A_j$
for some $c\in A_{j}$,
$(c_1,\dots,c_n)$ is a tuple generating $\mathbf A_{j}$.
Then there exists 
a term $t$ such that 
\begin{enumerate}
    \item $t^{\mathbf A_i}(a_1,\dots,a_{n},a_{n+1})\in C$
for all $i\in[s]$, $C<_{\TC}\mathbf A_{i}$, $a_1,\dots,a_n\in A_{i}$, and $a_{n+1}\in C$;
    \item $\{t^{\mathbf A_j}(c_1,\dots,c_{n},b)\mid b\in A_{j}\}\neq A_{j}$.
\end{enumerate}
\end{lem}

\begin{proof}
First, let us add all subalgebras $\mathbf A_{s+1},\dots, \mathbf A_{s'}$ of the algebras 
$\mathbf A_1,\dots,\mathbf A_s$ to the list.
By Lemma \ref{LEMCenterIntersectSubalgebra},
$C\le_{\TC} \mathbf A$ and $B\le \mathbf A$ imply
$C\cap B\le_{\TC} \mathbf B$.
Then, having all subalgebras in the list, 
it is sufficient to prove condition 1 only 
for tuples $(a_1,\dots,a_n,a_{n+1})$ generating $\mathbf A_{i}$ for 
each $i\in[s']$.

Let us consider 
the free generated relation 
$R^{n+1}_{\mathbf A_1,\dots,\mathbf A_{s'}}$.
Let $\gamma_1,\dots,\gamma_{n+1}$ be the generators of this relation. 
For every 
$(\mathbf A_{i},\alpha)\in I_{\mathbf A_1,\dots,\mathbf A_{s'}}^{n+1}$ 
by 
$D^{(\top)}_{(\mathbf A_{i},(a_1,\dots,a_{n+1}))}$
we denote the minimal central subuniverse of $\mathbf A_{i}$ containing $a_{n+1}$, which can be equal to  $A_{i}$.
It exists because the intersection of two central subuniverses is central
(Lemma \ref{LEMCenterIntersection}).
Let 
$C$ be the minimal central subuniverse of $\mathbf A_{j}$ 
containing $c$.
By Lemma \ref{LEMCenterImplies},
$\proj_{(\mathbf A_{i},(c_1,\dots,c_n,c))}(R^{n+1}_{\mathbf A_1,\dots,\mathbf A_{s'}}\cap 
D^{(\top)})\dot\le_{\TC}\mathbf A_{i}$.
Since $R^{n+1}_{\mathbf A_1,\dots,\mathbf A_{s'}}\cap 
D^{(\top)}$
contains the tuple $\gamma_{n+1}$, 
it is not empty.
Then, by the minimality of the central subuniverses 
we have
$\proj_{(\mathbf A_{i},(c_1,\dots,c_n,c))}(R^{n+1}_{\mathbf A_1,\dots,\mathbf A_{s'}}\cap 
D^{(\top)})=C$.
By Lemma \ref{LEMUbiquity} we have one of the two cases:

Case 1. There exists $B<_{\TC}^{A_{j}} C$.
Choose a tuple  
$\gamma\in R^{n+1}_{\mathbf A_1,\dots,\mathbf A_{s'}}\cap 
D^{(\top)}$ such that 
$\gamma(\mathbf A_{j},(c_1,\dots,c_n,c))\in B$
and the corresponding term $t$ generating this tuple
from $\gamma_1,\dots,\gamma_{n+1}$.
Since $\gamma\in D^{(\top)}$, the term $t$ satisfies the required condition 1.
It remains to show that 
$h(x)\approx t^{\mathbf A_{j}}(c_1,\dots,c_n,x)$ is not a bijection. 
By Lemma \ref{LEMCenterOfCenter}, $B<_{\TC} \mathbf A_{j}$, 
then condition 1 implies that 
$h(B)\subseteq B$.
Since 
$h(c)\in B$ and $c\notin B$, 
the operation $h$ cannot be bijective.

Case 2. 
There exists $B<_{T}^{A_{j}} C$ for some 
$T\in\{\TBA,\TPC,\TL\}$. 
Let $D^{(\bot)}$ be defined by 
$D^{(\bot)}_{(\mathbf A_{j},(a_1,\dots,a_n,d))} = B$
for every $d\in C$ 
and $D^{(\bot)}_{(\mathbf A_i,\alpha)}=
A_{i}$
 for all other coordinates 
$(\mathbf A_i,\alpha)\in I^{n+1}_{\mathbf A_1,\dots,\mathbf A_{s'}}$.
Since the tuple 
$(c_1,\dots,c_{n})$ generates $\mathbf A_{j}$, 
$R^{n+1}_{\mathbf A_1,\dots,\mathbf A_{s'}}\cap D^{(\bot)}\neq\varnothing$.
Combining this with 
$R^{n+1}_{\mathbf A_1,\dots,\mathbf A_{s'}}\cap D^{(\top)}\neq\varnothing$
and Lemma \ref{LEMMainStableIntersection},
we derive that 
$R^{n+1}_{\mathbf A_1,\dots,\mathbf A_{s'}}\cap D^{(\top)}\cap D^{(\bot)}\neq\varnothing$.
Then any term generating a tuple 
from this nonempty intersection
satisfies the required conditions 1 and 2.
\end{proof}

\begin{lem}\label{LEMLinearGivesSquizing}
Suppose $p$ is a prime, $\mathbf A\in\mathcal V_{k}$, 
$0_{\mathbf A}$ is $\wedge$-irreducible, there is no bridge 
from $0_{\mathbf A}$ to $0_{\mathbf Z_{p}}$,
$(a_1,\dots,a_n)$ is a tuple generating $\mathbf A$.
Then there exists 
a term $t$ such that 
\begin{enumerate}
    \item $t^{\mathbf Z_{p}}(x_1,\dots,x_{n+1})\approx x_{n+1}$;
    \item $\{t^{\mathbf A}(a_1,\dots,a_{n},b)\mid b\in A\}\neq A$.
\end{enumerate}
\end{lem}

\begin{proof}
Let us consider 
the free generated relation 
$R^{n+1}_{\mathbf Z_{p},\mathbf A}$.
Let $\gamma_1,\dots,\gamma_{n+1}$ be the generators for this relation. 
Choose two different elements $c$ and $c'$
such that
$(c,c')\in 0_{\mathbf A}^{+}$.
Let us define a reduction $D^{(\top)}$ as follows. 
For  
$(\mathbf Z_{p},\alpha)\in I^{n+1}_{\mathbf Z_{p},\mathbf A}$ 
let
$D^{(\top)}_{(\mathbf Z_{p},\alpha)} = \{\alpha(n+1)\}$, 
and $D^{(\top)}_{(\mathbf A,\alpha)}=A$ for $(\mathbf A,\alpha)\in I^{n+1}_{\mathbf Z_{p},\mathbf A}$.
We define a reduction $D^{(\bot)}$ by  
$D_{(\mathbf A,\alpha)}^{(\bot)} = 
\proj_{(\mathbf A,\alpha)}(R^{n+1}_{\mathbf Z_{p},\mathbf A}\cap 
D^{(\top)})$
for $(\mathbf A,\alpha)\in I^{n+1}_{\mathbf Z_{p},\mathbf A}$,
and $D_{(\mathbf Z_{p},\alpha)}^{(\bot)} = D_{(\mathbf Z_{p},\alpha)}^{(\top)}$ for 
$(\mathbf Z_{p},\alpha)\in I^{n+1}_{\mathbf Z_{p},\mathbf A}$.
By Lemma \ref{LEMPropagateToRelations},  
$D^{(\bot)}\lll D^{(\top)}$, and 
$R^{n+1}_{\mathbf Z_{p},\mathbf A}\cap D^{(\bot)}$ is not empty as it contains 
$\gamma_{n+1}$.

By Lemma \ref{LEMUbiquity} choose $\{b\}\lll^{A} D^{(\bot)}_{(\mathbf A,(a_1,\dots,a_n,c))}$.
If there exists a tuple 
in $R^{n+1}_{\mathbf Z_{p},\mathbf A}\cap 
D^{(\top)}$ such that 
its coordinates 
$(\mathbf A,(a_1,\dots,a_n,c))$
and $(\mathbf A,(a_1,\dots,a_n,c'))$
are equal to $b$, 
then the term corresponding to this tuple is a required term.
Otherwise, 
choose a minimal $C$ 
such that 
$\{b\}\lll^{A} C\lll^{A}A$
and $R^{n+1}_{\mathbf Z_{p},\mathbf A}\cap 
D^{(\top)}$ contains a tuple $\gamma$ 
whose $(\mathbf A,(a_1,\dots,a_n,c))$-th coordinate equals $b$ 
and $(\mathbf A,(a_1,\dots,a_n,c'))$-th coordinate 
belongs to $C$.
Choose $B$ such that 
$\{b\}\lll^{A} B<_{T(\delta)}^{\mathbf A} C$.
Since 
$b\in\Sg_{\mathbf A}(\{a_1,\dots,a_{n}\})$, 
$R^{n+1}_{\mathbf Z_{p},\mathbf A}$ contains a tuple whose 
coordinates $(\mathbf A,(a_1,\dots,a_n,c))$
and $(\mathbf A,(a_1,\dots,a_n,c'))$ are equal to $b$.
Then by Lemma \ref{LEMMainStableIntersection},
$T = \TL$  and there is a bridge from 
$\delta$ to $0_{\mathbf Z_{p}}$.
Since $(c,c')\in 0_{\mathbf A}^{+}$, we have 
$$(\gamma(\mathbf A,(a_1,\dots,a_n,c)),\gamma(\mathbf A,(a_1,\dots,a_n,c')))\in 0_{\mathbf A}^{+}\setminus \delta.$$
Hence, $\delta\not\supseteq 0_{\mathbf A}^{+}$ and 
by the $\wedge$-irreducibility of 
$0_{\mathbf A}$ we get 
$\delta = 0_{\mathbf A}$.
The existence of a bridge 
between $0_{\mathbf A}$
and $0_{\mathbf Z_{p}}$ contradicts our assumptions and  completes the proof.

\end{proof}

\begin{cor}\label{CORLinearGivesSquizing}
Suppose $\mathbf A_1,\mathbf A_2\in\mathcal V_{k}$, 
$\sigma$ is a maximal congruence on $\mathbf A_{1}$ of linear type,
$\omega$ is a $\wedge$-irreducible congruence on $\mathbf A_2$, there is no bridge 
from $\omega$ to $\sigma$, 
$(a_1,\dots,a_n)$ is a tuple generating $\mathbf A_2$.
Then there exists 
a term $t$ such that 
\begin{enumerate}
    \item $t^{\mathbf A_1/\sigma}(x_1,\dots,x_{n+1})\approx x_{n+1}$;
    \item $\{t^{\mathbf A_2}(a_1,\dots,a_{n},b\mid b\in A_{2}\}\neq A_{2}$.
\end{enumerate}
\end{cor}

\begin{proof}
By Lemma \ref{LEMLInearOnTheTopIsEasy},
$\mathbf A_{1}/\sigma\cong \mathbf Z_{p}$ for some prime $p$.
It remains to 
apply Lemma \ref{LEMLinearGivesSquizing}
to the algebras $\mathbf A_1/\sigma$ and $\mathbf A_{2}/\omega$.
\end{proof}



\begin{lem}\label{LEMNoMaltsevImpliesSqueezing}
Suppose $\mathbf A_1,\dots,\mathbf A_s\in\mathcal V_{q}$,
$\sigma$ is a maximal congruence on $\mathbf A_1$ of linear type, 
there does not exist a term $m$ such that 
$m^{\mathbf A_{i}}$ is a Maltsev operation for every $i\in[s]$. 
Then there exists 
a term $t$ such that 
\begin{enumerate}
    \item $t^{\mathbf A_i/\sigma}(x,y)\approx y$;
    \item $\{t^{\mathbf A_j}(a,b\mid b\in A_{j}\}\neq A_{j}$ for some
    $a\in A_{j}$.
\end{enumerate}
\end{lem}
\begin{proof}
By Lemma \ref{LEMLInearOnTheTopIsEasy},
$\mathbf A_{1}/\sigma\cong \mathbf Z_{p}$ for some prime $p$.
Let us consider a term $g_1$ 
such that 
$g_1^{\mathbf A_1/\sigma}(x_1,\dots,x_{p+1})\approx x_1+\dots+x_{p+1}$.
For every 
$i\in\mathbb N$
we define 
$g_{i+1}$ as a composition of $g_1$ with $p+1$ copies of $g_{i}$ with different variables. So, $g_{i}$ is of arity $(p+1)^{i}$ for every $i$.
Let $N = \max\limits_{i\in[s]}(|A_{i}|)$.
Then $g_{N!}^{\mathbf A_i}(x,\dots, x,g_{N!}^{\mathbf A_i}(x,\dots, x,y)) = g_{N!}^{\mathbf A_i}(x,\dots, x,y)$
for every $i$ and 
$g_{N!}^{\mathbf A_1/\sigma}(x,x,\dots,x,y) = 
g_{N!}^{\mathbf A_1/\sigma}(y,x,\dots,x,x) = y$.
Put 
$t_1(x,y) = g_{N!}(x,\dots,x,y)$,
$t_2(x,y) = g_{N!}(y,x,\dots,x)$,
and 
$m(x,y,z) = g_{N!}(x,y,\dots,y,z)$. 
Then $m^{\mathbf A_{i}}$ is a Maltsev operation
for every $i\in[s]$, or one of the $t_1$ or $t_2$ satisfies the required conditions 1 and 2.
\end{proof}

\begin{lem}\label{LEMChooseCombineStronglyStable}
Suppose 
$\mathbf B, \mathbf A_{1},\dots,\mathbf A_s\in\mathcal V_{k}$, 
$\mathbf B$ is a BA and center free PC algebra, and 
\begin{enumerate}
    \item[(1)] $\mathbf B$ is not isomorphic to $\mathbf A_{i}/\sigma$ 
    for any $i\in[s]$ and any congruence $\sigma$;
    \item[(2)] $t_{1}$ is an $n$-ary term that is strongly stable in $\mathbf A_{i}$ for every $i\in[s]$;
    \item[(3)] $t_{2}$ is an $m$-ary term.
\end{enumerate}
Then there exists an $(m+n)$-ary term $t$ 
such that 
\begin{itemize}
\item[(a)] $t^{\mathbf A_{i}}(b_1,\dots,b_{m},a_1,\dots,a_n) = t_1^{\mathbf A_{i}}(a_1,\dots,a_{n})$
for every $i\in [s]$ and any tuples $(b_1,\dots,b_m)$ and $(a_1,\dots,a_n)$ generating $\mathbf A_{i}$,
    \item[(b)] $t^{\mathbf B}(b_1,\dots,b_{m},a_1,\dots,a_n) = t_2^{\mathbf B}(b_1,\dots,b_{m})$
for any tuples $(b_1,\dots,b_m)$ and $(a_1,\dots,a_n)$ generating $\mathbf B$.
\end{itemize}
\end{lem}
\begin{proof}
Let 
$R^{m+n}_{\mathbf B,\mathbf A_1,\dots,\mathbf A_{s}}$
be the free generated relation. 
Let us define a reduction 
$D^{(\top)}$ by 
$D^{(\top)}_{(\mathbf B,(b_1,\dots,b_m,a_1,\dots,a_n))} = \{t_{2}^{\mathbf B}(b_1,\dots,b_m)\}$ 
when $(a_1,\dots,a_n)$ and $(b_1,\dots,b_m)$ generate $\mathbf B$, 
and $D^{(\top)}_{(\mathbf C,\alpha)} = C$ otherwise.
Since $\mathbf B$ is a BA and center free PC algebra, 
$D_{(\mathbf B,(b_1,\dots,b_m,a_1,\dots,a_n))}^{(\top)}<_{\TPC} \mathbf B$.
Similarly, 
define a reduction 
$D^{(\bot)}$ by 
$D^{(\bot)}_{(\mathbf A_{i},(b_1,\dots,b_m,a_1,\dots,a_n))} = \{t_{1}^{\mathbf A_{i}}(a_1,\dots,a_n)\}$
when $(a_1,\dots,a_n)$ and $(b_1,\dots,b_m)$ generate $\mathbf A_{i}$, 
and $D^{(\bot)}_{(\mathbf C,\alpha)} = C$ otherwise.
The terms $t_1$ and $t_2$ witness that 
$R^{m+n}_{\mathbf B,\mathbf A_1,\dots,\mathbf A_{s}}\cap D^{(\top)}$ and 
$R^{m+n}_{\mathbf B,\mathbf A_1,\dots,\mathbf A_{s}}\cap D^{(\bot)}$
are not empty. 
If 
$R^{m+n}_{\mathbf B,\mathbf A_1,\dots,\mathbf A_{s}}\cap D^{(\top)}\cap D^{(\bot)}=\varnothing$
then Lemma \ref{LEMMainStableIntersection} implies that 
there exists a congruence $\sigma$ on some $\mathbf A_{i}$ 
such that $\mathbf A_{i}/\sigma\cong \mathbf B$, which contradicts condition (1).
Thus, $R^{m+n}_{\mathbf B,\mathbf A_1,\dots,\mathbf A_{s}}\cap D^{(\top)}\cap D^{(\bot)}\neq\varnothing$. 
It remains to choose any tuple from this nonempty intersection, then a term $t$
generating this tuple satisfies the required conditions.
\end{proof}

\begin{lem}\label{LEMGoodPropertyForStable}
Suppose  
$t_1$ is an $n$-ary strongly stable term in algebras $\mathbf A_{1},\dots,\mathbf A_s\in\mathcal V_{k}$,
then 
for any $i\in [s]$, $m\le n$, and tuples $(b_1,\dots,b_m)$ and $(a_1,\dots,a_n)$ generating $\mathbf A_i$
there exists a term $t_2$ such that 
\begin{itemize}
    \item[(s)] $t_2$ is strongly stable in $\mathbf A_{1},\dots,\mathbf A_s$,
    \item[(g)] 
    $t_2^{\mathbf A_{i}}(b_1,\dots,b_m,c_1,\dots,c_{n-m}) = t_{1}^{\mathbf A_{i}}(a_1,\dots,a_n)$
    for any $c_1,\dots,c_{n-m}\in A_{i}$.
\end{itemize}
\end{lem}
\begin{proof}
Since $a_1,\dots,a_n\in\Sg_{\mathbf A_{i}}(\{b_1,\dots,b_m\})$, 
there are $m$-ary terms $g_1,\dots,g_n$ such that 
$g_{j}^{\mathbf A_{i}}(b_1,\dots,b_{m})=a_{j}$ for every $j\in[n]$.
Define 
$t_{3}(x_1,\dots,x_m) = t_1(g_1(x_1,\dots,x_m),\dots,g_{n}(x_1,\dots,x_m))$.
Consider the free generated relation
$R^{n}_{\mathbf A_1,\dots,\mathbf A_{s}}$.
Let us define a reduction 
$D^{(\top)}$ by 
$D^{(\top)}_{(\mathbf A_{i},(b_1,\dots,b_m,c_1,\dots,c_{n-m}))} = \{t_{3}^{\mathbf B}(b_1,\dots,b_m)\}$ 
for all $c_1,\dots,c_{n-m}\in A_{i}$, 
and $D^{(\top)}_{(\mathbf C,\alpha)} = C$ otherwise.
Notice that 
$\{t_{3}^{\mathbf A_{i}}(b_1,\dots,b_m)\}=
\{t_{1}^{\mathbf A_{i}}(a_1,\dots,a_n)\}\lll \mathbf A_{i}$.
The term $t_3$ witnesses that 
$R^{n}_{\mathbf A_1,\dots,\mathbf A_{s}}\cap D^{(\top)}\neq \varnothing$.
Using Lemmas \ref{LEMUbiquity} and \ref{LEMPropagateToRelations}, 
we can find a reduction 
$D^{(\bot)}\lll D^{(\top)}$ such that 
$|R^{n}_{\mathbf A_1,\dots,\mathbf A_{s}}\cap D^{(\bot)}|=1$.
It remains to choose a term $t_2$ generating 
the only tuple in $R^{n}_{\mathbf A_1,\dots,\mathbf A_{s}}\cap D^{(\bot)}$.
\end{proof}
\subsection{Three ways to reduce the domain}\label{SUBSECTIONThreeWays}

In this section we show how squeezers can be used to reduce domains
and define a PS operation using 
PS operations on smaller domains.
First, we will show how to do this for a very special PC case even without squeezers. 
 Then, using squeezers we do this for the linear 
 case and remaining PC cases. Finally, we present the construction for the central case, which is the most complicated one.

\begin{thm}\label{THMReductionForPCCongruence}
Suppose 
$\mathbf A_{1},\dots,\mathbf A_s\in\mathcal V_{p}$,
$|\mathbf A_{1}|\ge |\mathbf A_{j}|$ for all $j\in [s]$,
$\mathbf A_{1}$ is a BA and center free PC algebra, and 
\begin{itemize}
    \item[(i)] for every 
    $\{\mathbf B_{1},\dots,\mathbf B_u\}<\{\mathbf A_{1},\dots,\mathbf A_{s}\}$ 
    and every $m,n\in\mathbb N$ there exists  a
    $(\underbrace{\overline{k^{n}},\dots,\overline{k^{n}}}_{m})$-PS term 
    for 
    $\mathbf B_{1},\dots,\mathbf B_{u}$.
\end{itemize}
Then $\mathbf A_{1},\dots,\mathbf A_s$
admit a 
    $(\underbrace{\overline{k},\dots,\overline{k}}_{n})$-PS term 
    for every $n\in \mathbb N$.
\end{thm}
\begin{proof}
First, assume that 
$\mathbf A_{i}$ is not isomorphic to $\mathbf A_{1}$ for any $i\ge 2$ as 
otherwise we could exclude $\mathbf A_{i}$ from the list and use 
assumption (i).
Let 
$\mathbf A_{s+1},\dots,\mathbf A_{s'}$ be the set of all subalgebras 
of $\mathbf A_{1},\dots,\mathbf A_{s}$.

Let $f_1,\dots,f_m$ be a set of $n$-ary terms that are strongly stable in
$\mathbf A_2,\dots,\mathbf A_{s'}$
such that 
for any other strongly stable term $f$ there exists 
$j\in[m]$ such that $f_{j}^{\mathbf A_{i}} = f^{\mathbf A_{i}}$ for all $i\in[s']$. Thus, these terms should cover all possible (strongly stable) term operations on each $\mathbf A_{i}$.

Recall that 
local palette symmetric terms in $\mathbf A$
 differ from ordinary palette symmetric terms in that the symmetry is required only on tuples that generate 
$\mathbf A$.
By Lemma \ref{LEMNoLinearCongruencesImpliesPS},
there exists 
a local $(\underbrace{\overline{k},\dots,\overline{k}}_{n})$-PS term $t_1$
for $\mathbf B$.
Using Lemma \ref{LEMChooseCombineStronglyStable}
we build terms $g_1,\dots,g_m$ of arity 
$k\cdot n+n$ such that for every 
$j\in [m]$ 
\begin{itemize}
\item[(a)] $g_{j}^{\mathbf A_{i}}(b_1,\dots,b_{k\cdot n},a_1,\dots,a_n) = f_{j}^{\mathbf A_{i}}(a_1,\dots,a_{n})$
for every $i\in[s']$ and any tuples $(b_1,\dots,b_{k\cdot n})$ and $(a_1,\dots,a_n)$ generating $\mathbf A_{i}$,
    \item[(b)] $g_{j}^{\mathbf A_1}(b_1,\dots,b_{k\cdot n},a_1,\dots,a_n) = t_1^{\mathbf A_1}(b_1,\dots,b_{k\cdot n})$
for any tuples $(b_1,\dots,b_{k\cdot n})$ and $(a_1,\dots,a_n)$ generating $\mathbf A_{1}$.
\end{itemize}

By condition (i), 
we can find 
a $(\underbrace{\overline{k^{n}},\dots,\overline{k^{n}}}_{m})$-PS term $t_2$
for $\mathbf A_{2},\dots,\mathbf A_{s'}$.
Then the required $(\underbrace{\overline{k},\dots,\overline{k}}_{n})$-PS term can be defined by 
\begin{align*}
t(\mathbf x_1, \dots,\mathbf x_n):=
t_2((g_{1}(\mathbf x_1, 
\dots,\mathbf x_n,\mathbf x_1^{j_1},\dots,\mathbf x_n^{j_n}))_{j_1,\dots,j_n\in[k]},\dots,&\\
\dots, (g_{m}(\mathbf x_1, \dots,\mathbf x_n,\mathbf x_1^{j_1},\dots,&\mathbf x_n^{j_n}))_{j_1,\dots,j_n\in[k]}).
\end{align*}
Here, the first $k^{n}$ arguments of $t_2$ are filled with 
$g_{1}$, where 
$\mathbf x_1^{j_1},\dots,\mathbf x_n^{j_n}$ are substituted into the last $n$ coordinates for $j_1,\dots,j_n\in[k]$. 
The next $k^{n}$ arguments are filled with $g_{2}$ and so on.
Let us show that $t$ is in fact a $(\overline{k},\dots,\overline{k})$-PS term for $\mathbf A_1,\dots,\mathbf A_s$.
Since we added all the subalgebras to the list, it is sufficient to prove that  
$t$ is local $(\overline{k},\dots,\overline{k})$-PS term for 
each $i\in[s']$.
Let 
$(\mathbf a_1,\dots,\mathbf a_n)$ be a $(\overline{k},\dots,\overline{k})$-palette tuple generating $A_{i}$.
Notice 
$(\mathbf a_1^{j_1},\dots,\mathbf a_n^{j_n})$ also generates $\mathbf A_{i}$
for all $j_1,\dots,j_n\in[k]$.
First, we consider the case when $i=1$.
By condition (b), 
$g_{i}^{\mathbf A_{1}}(\mathbf a_1,\dots,\mathbf a_n,\mathbf a_1^{j_1},\dots,\mathbf a_{n}^{j_{n}}) = t_{1}^{\mathbf A_{1}}(\mathbf a_1,\dots,\mathbf a_n)$. 
Hence, the result does not depend on $j_1,\dots,j_n$ and we substitute only one value in $t_2$. 
Since $t_1$ is local $(\overline k,\dots,\overline k)$-PS in $\mathbf A_1$, 
permutations in $\mathbf a_1,\dots,\mathbf a_n$ do not change this one value.
Thus, we proved the required condition for such tuples.

Consider the case when $i\ge 2$.
Property (a) implies that 
each $g_{j}^{\mathbf A_{i}}$ on such tuples only depends on the last 
$n$ coordinates, hence permutation of elements 
inside blocks of $(\mathbf a_1,\dots,\mathbf a_n)$ only imply permutations 
in the corresponding blocks of size $k^{n}$ of $t_2$
but the content of each block does not
depend on the permutations.
Since $t_2$ is $(\overline{k^{n}},\dots,\overline{k^{n}})$-PS in $\mathbf A_{i}$, 
it remains to prove that the tuple we substitute in $t_2^{\mathbf A_{i}}$ is 
$(\overline{k^{n}},\dots,\overline{k^{n}})$-palette.

If a value $c\in A_{i}$ appears in the tuple we substitute in 
$t_{2}^{\mathbf A_{i}}$, then there should be 
$j\in[m]$ and 
a tuple $(c_1,\dots,c_n)$ generating $\mathbf A_{i}$ 
such that 
$f_{j}^{\mathbf A_{i}}(c_1,\dots,c_n) = c$.
Since the tuple $(\mathbf a_1,\dots,\mathbf a_n)$
is palette, 
there exist constant blocks $\mathbf a_{\ell_1},\dots,\mathbf a_{\ell_v}$ 
containing all the elements $c_1,\dots,c_n$. 
By Lemma \ref{LEMGoodPropertyForStable} 
there exists $j'\in [m]$
such that 
$f_{j'}(\mathbf a_{1}^{j_1},\dots,\mathbf a_{n}^{j_{n}})= f_{j}(b_1,\dots,b_{n})$ for all $j_1,\dots,j_{n}\in[k]$. 
Hence, the $j'$-th block of the tuple we substitute in $t_2$ consists of 
the element $c$.
Thus, this tuple is palette.
\end{proof}




\begin{thm}\label{THMNotMinimalExampleCongruence}
Suppose 
$\mathbf A_{1},\dots,\mathbf A_s\in\mathcal V_{p}$,
$\sigma$ is a congruence on $\mathbf A_{1}$, 
$\mathcal F$ is a set of terms over the variables 
$x_1,\dots,x_{n+1}$ closed under right composition,
$\mathcal F$ is good for algebras $\mathbf A_1,\dots,\mathbf A_{s}$
and their subalgebras,
$\varnothing \neq J\subseteq [s]$,
$I\subseteq [s]$, 
\begin{enumerate}
\item [(p)] $f^{\mathbf A_{1}/\sigma}(a_1/\sigma,\dots,a_n/\sigma,x)\approx x$ for any $f\in \mathcal F$ 
and any tuple $(a_1,\dots, a_n)$ generating $\mathbf A_{1}$;
    \item[(s)] for every $j\in J$ there exist $t\in \mathcal F$ and a tuple $(a_1,\dots,a_n)$ generating $\mathbf A_{j}$ such that $|\{t^{\mathbf A_{j}}(a_1,\dots,a_n,b)\mid b\in A_{j}\}|<|A_{j}|$;
    \item[(k)] $\{t^{\mathbf A_{i}}(a_1,\dots,a_n,b)\mid b\in A_{i}\}=A_{i}$ for every $i\in I$, $t\in \mathcal F$, and $a_1,\dots,a_n\in A_{i}$;    
    \item[($\ell$)] 
    $|A_{i}|\le\max\limits_{j\in J} |A_{j}|$ for any $i\notin J\cup I$; 
    \item[(i)] for every 
    $\{\mathbf B_{1},\dots,\mathbf B_u\}<\{\mathbf A_{1},\dots,\mathbf A_{s}\}$ 
    and every $m,\ell\in\mathbb N$ there exists  a
    $(\underbrace{\overline{k^{\ell}},\dots,\overline{k^{\ell}}}_{m})$-PS term 
    for 
    $\mathbf B_{1},\dots,\mathbf B_{u}$.
\end{enumerate}
Then $\mathbf A_{1},\dots,\mathbf A_s$
admit a 
    $(\underbrace{\overline{k},\dots,\overline{k}}_{n})$-PS term.
\end{thm}
Before proving the theorem, let us comment on
some of the above conditions. 
Condition (s) describes the nonempty set $J$ of algebras we can squeeze with a tuple generating the whole algebra, 
condition (k) describes the set $I$ of algebras that we cannot squeeze at all. 
Thus, $I\cap J=\varnothing$ but it is possible that $I\cup J\neq [s]$.
Condition ($\ell$) says that any algebra outside of $I$ and $J$ is not larger than one of the algebras in $J$.

\begin{proof}
Let 
$\mathbf A_{s+1},\dots,\mathbf A_{s'}$ be the set of all subalgebras 
of $\mathbf A_{1},\dots,\mathbf A_{s}$.
Using Lemma \ref{LEMMagicListOfTerms}
we choose a finite set 
$\{f_1,\dots,f_{m}\}\subseteq\squeezer(\mathcal F,\mathbf A_{1},\dots,\mathbf A_{s'})$.
Let $\mathcal U$ be the set of all 
algebras $\mathbf B\in \mathcal V_{p}$ such that 
$B\subsetneq A_{i}$ for some $i\in J$.
Then $\mathcal U\cup \{\mathbf A_{i}\mid i\in [s]\setminus J\}<\{\mathbf A_1,\dots,\mathbf A_{s}\}$, and, by condition (i), 
$\mathcal U\cup \{\mathbf A_{i}\mid i\in [s]\setminus J\}$
admits a 
    $(\underbrace{\overline{k},\dots,\overline{k}}_{m})$-PS term $t_0$.
Let $\mathcal U_{0}$ be the set of all proper subalgebras of $\mathbf A_{1}$.
Since 
 $\mathcal U_{0}\cup \{\mathbf A_2,\dots,\mathbf A_{s}\}<\{\mathbf A_1,\dots,\mathbf A_s\}$,  by assumption (i) 
there exists  a 
$(\underbrace{\overline{k^{n}},\dots,\overline{k^{n}}}_{m})$-PS term $t_1$ for 
 $\mathcal U_{0}$ and $\mathbf A_2,\dots,\mathbf A_s$. 
Then 
a 
    $(\underbrace{\overline{k},\dots,\overline{k}}_{n})$-PS term $t$ 
    for 
    $\mathbf A_{1},\dots,\mathbf A_s$ can be obtained as follows.
    
First, for every $i\in[m]$ let the term  
$g_{i}(\mathbf x_1,\dots,\mathbf x_n,y_1,\dots,y_n)$
be obtained from 
$t_0(\mathbf x_1,\dots,\mathbf x_n)$ by replacing 
each $w(z_1,\dots,z_p)$ by 
$$f_{i}(y_1,\dots,y_n,w(f_{i}(y_1,\dots,y_n,z_1),\dots,f_{i}(y_1,\dots,y_n,z_p))).$$
Thus, 
we add $f_{i}(y_1,\dots,y_n,z)$ to every variable $z$ to reduce the corresponding domain.
Then put 
\begin{align*}
t(\mathbf x_1, \dots,\mathbf x_n):=
t_1((g_{1}(\mathbf x_1, 
\dots,\mathbf x_n,\mathbf x_1^{j_1},\dots,\mathbf x_n^{j_n}))_{j_1,\dots,j_n\in[k]},\dots,&\\
\dots, (g_{m}(\mathbf x_1, \dots,\mathbf x_n,\mathbf x_1^{j_1},\dots,&\mathbf x_n^{j_n}))_{j_1,\dots,j_n\in[k]}).
\end{align*}
Here, the first $k^{n}$ arguments of $t_1$ are filled with 
$g_{1}$ where we replace the variables
$y_1,\dots,y_n$ by 
$\mathbf x_1^{j_1},\dots,\mathbf x_n^{j_n}$ for $j_1,\dots,j_n\in[k]$.
The next $k^{n}$ arguments are filled with $g_{2}$ and so on.
Let us show that $t$ is in fact a $(\overline{k},\dots,\overline{k})$-PS term for 
$\mathbf A_1,\dots,\mathbf A_s$.

First, let us prove this for $\mathbf A_{i}$, where $i\in I$.
Let 
$(\mathbf a_1,\dots,\mathbf a_n)$ be a $(\overline{k},\dots,\overline{k})$-palette tuple from $A_{i}^{nk}$.
Condition $(k)$ together with property 2 of squeezers imply that 
$f_{\ell}(a_1,\dots,a_n,x)\approx x$ for any $\ell\in[m]$ and $a_1,\dots,a_n\in A_{i}$.
Hence
$g_{\ell}^{\mathbf A_i}(\mathbf a_1, \dots,\mathbf a_n,\mathbf a_1^{j_1},\dots,\mathbf a_n^{j_n}) =t_{0}^{\mathbf A_{i}}(\mathbf a_1, \dots,\mathbf a_n)$
for all $\ell\in[m]$ and $j_1,\dots,j_{n}\in[k]$.
Since 
$t_0$ is PS in $\mathbf A_{i}$,
permutations in $\mathbf a_1,\dots,\mathbf a_n$ do not change the result.
Then the idempotency of $t_1^{\mathbf A_{i}}$ implies that $t$ is 
$(\overline{k},\dots,\overline{k})$-PS in $\mathbf A_{i}$.

Second, let us prove that $\mathbf A_{i}$ is $(\overline{k},\dots,\overline{k})$-PS for $i\in [s]\setminus I$.
Again, let 
$(\mathbf a_1,\dots,\mathbf a_n)$ be a $(\overline{k_1},\dots,\overline{k_n})$-palette tuple from $A_{i}^{nk}$.
Let us fix $j_1,\dots,j_n\in[k]$ and define 
$h(x) \approx f_{\ell}(\mathbf a_1^{j_1},\dots,\mathbf a_n^{j_n},x)$.
Let
$\mathbf B= \Sg_{\mathbf A_{i}}(\{\mathbf a_1^{j_1},\dots,\mathbf a_n^{j_n}\})$ and
$\mathbf B'=h(\mathbf B)$.
Let us show that  
$\mathbf B'\in \mathcal U\cup \{\mathbf A_{j}\mid j\in [s]\setminus J\}$.
In fact, if $i\in J$ and 
$(\mathbf a_1^{j_1},\dots,\mathbf a_n^{j_n})$ generates $\mathbf A_{i}$ 
then 
it follows from 
condition (s) and Lemma \ref{LEMSqueezForGenerating}
that $\mathbf B'\lneq \mathbf B=\mathbf A_{i}$ and $\mathbf B'\in \mathcal U$.
If $i\in J$ and $(\mathbf a_1^{j_1},\dots,\mathbf a_n^{j_n})$ does not generate $\mathbf A_{i}$, then $\mathbf B,\mathbf B'\in\mathcal U$. 
If $i\notin I\cup J$,
then either $h(x)\approx x$ and $\mathbf B'=\mathbf B= \mathbf A_{i}$, 
or by condition ($\ell$) $|B'|< |A_{i}|\le \max\limits_{j\in J} |A_{j}|$
and $\mathbf B'\in\mathcal U$.
Thus, the term $t_0$ is PS in $\mathbf B'$.
Recall that the choice of $j_1,\dots,j_n\in[k]$ 
determines the function $h$, and the function $h$ determines the algebra $\mathbf B'$.
Since $t_{0}$ is
$(\overline{k},\dots,\overline{k})$-PS in each $\mathbf B'$, 
permutations in  blocks $\mathbf a_1,\dots,\mathbf a_n$
only imply 
permutations in the corresponding blocks of size $k^{n}$ of $t_1$
but the content of each block does not depend on the permutations.

Let us show that the tuple 
we substitute in $t_1^{\mathbf A_{i}}$ is 
$(\overline{k^{n}},\dots,\overline{k^{n}})$-palette. 
This follows from the fact that any value $c$ that ever appears in this tuple 
comes from some $g_{e}^{\mathbf A_{i}}(\mathbf a_1, 
\dots,\mathbf a_n,\mathbf a_1^{j_1},\dots,\mathbf a_n^{j_n})=c$.
Since conditions (1) and (2) in Lemma \ref{LEMMagicListOfTerms} are equivalent, 
there exists $d\in[m]$ such that 
$g_{d}^{\mathbf A_{i}}(\mathbf a_1, 
\dots,\mathbf a_n,\mathbf a_1^{i_1},\dots,\mathbf a_n^{i_n})=
c$
for all $i_1,\dots,i_n\in[k]$.
This implies that 
the $d$-th block of the tuple we substitute into $t_1^{\mathbf A_{i}}$
contains only value $c$, which means that it is palette.

Thus, we showed that the tuple we substitute in 
$t_1^{\mathbf A_{i}}$ is 
$(\overline{k^{n}},\dots,\overline{k^{n}})$-palette 
and permutations in 
$\mathbf a_1,\dots,\mathbf a_n$ only imply permutations 
in the corresponding blocks of size $k^{n}$.
Since
$t_1$ is $(\overline{k^{n}},\dots,\overline{k^{n}})$-PS in $\mathbf A_{i}$ for $i\ge 2$, 
the term $t$ is
$(\overline{k},\dots,\overline{k})$-PS in $\mathbf A_i$ for $i\ge 2$.



Assume that $i=1$ and 
$(\mathbf a_1, 
\dots,\mathbf a_n)$ generates $\mathbf A_{1}$.
Let us show that 
the values $g_{d}^{\mathbf A_{1}}(\mathbf a_1, 
\dots,\mathbf a_n,\mathbf a_1^{j_1},\dots,\mathbf a_n^{j_n})$
and 
$g_{e}^{\mathbf A_{1}}(\mathbf a_1, 
\dots,\mathbf a_n,\mathbf a_1^{i_1},\dots,\mathbf a_n^{i_n})$
are equivalent modulo $\sigma$ for all $d,e\in[m]$ and
$j_1,\dots,j_{n},i_1,\dots,i_n\in[k]$.
By condition (p) 
and since $\sigma$ is preserved by $w^{\mathbf A_{1}}$,
the value $g_{e}^{\mathbf A_{1}}(\mathbf a_1,\dots,\mathbf a_n,\mathbf a_1^{i_1},\dots,\mathbf a_n^{i_n})$
is equivalent to $t_0^{\mathbf A_{1}}(\mathbf a_1,\dots,\mathbf a_n)$ for 
all $e\in[m]$ and $i_1,\dots,i_{n}\in[k]$. Therefore, we always substitute into $t_{1}^{\mathbf A_1}$
only values from one equivalence class of $\sigma$.
Since $t_1$ is $(\overline{k^{n}},\dots,\overline{k^{n}})$-PS
for all proper subalgebras of $\mathbf A_1$, 
 permutations inside blocks
$\mathbf a_1,\dots,\mathbf a_n$ do not change the result
and $t$ is $(\overline{k},\dots,\overline{k})$-PS in $\mathbf A_{1}$.

Finally, if $i=1$ and 
$(\mathbf a_1, 
\dots,\mathbf a_n)$ does not generate $\mathbf A_{1}$, 
then it follows from the fact that $t_1$ is $(\overline{k^{n}},\dots,\overline{k^{n}})$-PS
in the subalgebra generated by 
$(\mathbf a_1, 
\dots,\mathbf a_n)$.
\end{proof}


Recall that 
by $\mathcal C(\mathbf A)$ we denote the set of all proper subuniverses $B\le \mathbf A$ 
such that $(B\times A)\cup (A\times B)$ is a subuniverse of 
$\mathbf A^{2}$ (see Subsection \ref{SUBSECTIONMinimalTaylor}).

\begin{lem}\label{LEMCenterGoesInside}
Suppose $\mathbf A\in \mathcal V_{p}$, $h\colon A\to A$, $h\circ h = h$, 
$B\in\mathcal C(\mathbf A)$,
$h(B)\subseteq B$. 
Then $h(B)\in \mathcal C(h(\mathbf A))\cup \{h(A)\}$.
\end{lem}

\begin{proof}

It is sufficient to check that 
$w^{h(\mathbf A)}(x_1,\dots,x_k):= 
h(w^{\mathbf A}(x_1,\dots,x_k))$ preserves the relation 
$(h(B)\times h(A))\cup 
(h(A)\times h(B))$, which is straightforward.
\end{proof}



We say that a $(\overline{k_1},\dots,\overline{k_n},\ell)$-PS term operation $f$ in $\mathbf A$
\emph{satisfies property (c)} if 
for any $B\in \centers(\mathbf A)$ and 
any $(\overline{k_1},\dots,\overline{k_n},\ell)$-palette tuple 
$(\mathbf a_1,\dots,\mathbf a_n,\mathbf b)$ such that 
$\mathbf b\in B^{\ell}$ 
we have $f(\mathbf a_1,\dots,\mathbf a_n,\mathbf b)\in B$.

For an algebra $\mathbf A$, a tuple $(a_1,\dots,a_{n})\in A^{n}$ is called \emph{self-contained} if
$\{a_1,\dots,a_n\}$ is a subuniverse of an algebra $\mathbf A$.
We say that
a function 
is \emph{weak 
$(\underbrace{\overline{k^{n}},\dots,\overline{k^{n}}}_{m},\ell)$-PS satisfying property (c)}, 
if we require both palette symmetry and property (c) only on 
self-contained tuples.

\begin{lem}\label{LEMWeaKPSImpliesPS}
Suppose $k,n,\ell \in \mathbb N$, algebras 
$\mathbf A_1,\dots,\mathbf A_s$ admit
a weak 
$(\underbrace{\overline{k^{n}},\dots,\overline{k^{n}}}_{m},\ell)$-PS
term satisfying property (c) for all $m \in \mathbb N$.
Then $\mathbf A_1,\dots,\mathbf A_s$ admit
$(\underbrace{\overline{k},\dots,\overline{k}}_{n},\ell)$-PS term satisfying property (c).
\end{lem}

\begin{proof}
Choose terms $t_1,\dots,t_{m}$ of arity $n$ so that
for any $i\in [s]$ and any $n$-ary term operation $g$ in $\mathbf A_{i}$ there exists $j\in [m]$ 
such that $t_{j}^{\mathbf A_{i}}=g$.
Let $f$ be  
a weak 
$(\underbrace{\overline{k^{n}},\dots,\overline{k^{n}}}_{m},\ell)$-PS
term satisfying property (c) 
for $\mathbf A_1,\dots,\mathbf A_s$.
Then a $(\underbrace{\overline{k},\dots,\overline{k}}_{n},\ell)$-PS term
$t$ can be defined by 
\begin{align*}
t(\mathbf x_1, \dots,\mathbf x_n,\mathbf x_{n+1}):=
f((t_{1}(\mathbf x_1^{j_1},\dots,\mathbf x_n^{j_n}))_{j_1,\dots,j_n\in[k]},\dots,
(t_{m}(\mathbf x_1^{j_1},\dots,\mathbf x_n^{j_n}))_{j_1,\dots,j_n\in[k]},\mathbf x_{n+1}).
\end{align*}
Notice that permutations 
in $\mathbf x_1,\dots,\mathbf x_{n}$ 
only imply permutations inside blocks of size $k^{n}$.
Thus, we only need to check that the tuple we substitute in $f$ is 
$(\underbrace{\overline{k^{n}},\dots,\overline{k^{n}}}_{m},\ell)$-palette
and self-contained.
Fix $i\in [s]$
and a $(\underbrace{\overline{k},\dots,\overline{k}}_{n},\ell)$-palette tuple 
$(\mathbf a_1,\dots,\mathbf a_{n+1})\in A_{i}^{kn+\ell}$.
Let $b$ be an element of the subalgebra generated by elements of this tuple 
in $\mathbf A_{i}$.
Let $\mathbf a_{j_1},\dots,\mathbf a_{j_r}$ be all the constant blocks.
Let $g$ be an $n$-ary term operation in $\mathbf A_{i}$ such that 
only variables on positions $j_1,\dots,j_{r}$ are not dummy,
and $g(\mathbf a_1^{1},\dots,\mathbf a_{n}^{1}) = b$.
By our assumptions, $g = t_{j}^{\mathbf A_{i}}$ for some $j\in[m]$.
Then, the $j$-th block of the tuple we substitute consists only of 
element $b$, which means that the corresponding tuple is palette and self-contained.
\end{proof}


\begin{remark}\label{REMRemoveEllBlock}
For any 
$(\overline{k_1},\dots,\overline{k_n},k_i)$-PS function $f$  
we can derive 
a $(\overline{k_1},\dots,\overline{k_n})$-PS function $g$  
by 
$$g(\mathbf x_1,\dots,\mathbf x_n):=
f(\mathbf x_1,\dots,\mathbf x_n,\mathbf x_{i}).$$
\end{remark}

\begin{thm}\label{THMNotMinimalExampleTernaryAbs}

Suppose 
$\mathbf A_{1},\dots,\mathbf A_s\in\mathcal V_{p}$ are minimal Taylor,
$\mathcal F$ is a set of terms over the variables 
$x_1,\dots,x_{n+1}$ closed under right composition,
$\mathcal F$ is good for algebras $\mathbf A_1,\dots,\mathbf A_{s}$
and their subalgebras,  $\varnothing\neq I\subseteq [s]$, 
\begin{enumerate}
\item [(p)] $f^{\mathbf A_{i}}(a_1,\dots,a_n,b)\in B$ for any $f\in \mathcal F$,
$i\in[s]$, $B\in\mathcal C(\mathbf A_{i})$, $a_1,\dots,a_n\in A_{i}$, and $b\in B$; 
\item[(s)] for every $i\in I$ there exist $t\in \mathcal F$ and a tuple $(a_1,\dots,a_n)$ generating  $\mathbf A_{i}$ such that $|\{t^{\mathbf A_{i}}(a_1,\dots,a_n,b)\mid b\in A_{i}\}|<|A_{i}|$;
    \item[(k)] $\{t^{\mathbf A_{i}}(a_1,\dots,a_n,b)\mid b\in A_{i}\}=A_{i}$ for every $i\in [s]\setminus I$, $t\in \mathcal F$, and $a_1,\dots,a_n\in A_{i}$;    
\item[(n)] $\mathcal C(\mathbf A_{i})$ is not empty for some $i\in[s]$;
    \item[(i)] for every 
    $\{\mathbf B_{1},\dots,\mathbf B_u\}<\{\mathbf A_{1},\dots,\mathbf A_{s}\}$ 
    and every $m,\ell,n_0\in\mathbb N$ there exists a 
 $(\underbrace{\overline{k^{n_0}},\dots,\overline{k^{n_0}}}_{m},\ell)$-PS term $t$
    for 
    $\mathbf B_{1},\dots,\mathbf B_{u}$ satisfying property (c).     
\end{enumerate}
Then for every $\ell\in\mathbb N$ the algebras $\mathbf A_{1},\dots,\mathbf A_s$
admit a weak
    $(\underbrace{\overline{k},\dots,\overline{k}}_{n},\ell)$-PS term 
satisfying property (c).
\end{thm}
\begin{proof}
Let 
$\mathbf A_{s+1},\dots,\mathbf A_{s'}$ be the set of all subalgebras 
of $\mathbf A_{1},\dots,\mathbf A_{s}$.
Using Lemma \ref{LEMMagicListOfTerms}
we choose a finite set 
$\{f_1,\dots,f_{m}\}\subseteq\squeezer(\mathcal F,\mathbf A_{1},\dots,\mathbf A_{s'})$.
Let us fix $\ell\in\mathbb N$.
Put $M = (k!)^{n}\cdot \ell!$.
Let $\mathcal U$ be the set of all 
algebras $\mathbf B\in \mathcal V_{p}$ such that 
$B\subsetneq A_{i}$ for some $i\in I$.
Then $\mathcal U\cup \{\mathbf A_{i}\mid i\in [s]\setminus I\}<\{\mathbf A_1,\dots,\mathbf A_{s}\}$, and, by condition (i), 
$\mathcal U\cup \{\mathbf A_{i}\mid i\in [s]\setminus I\}$
admits a 
   $(\underbrace{\overline{k},\dots,\overline{k}}_{n},M)$-PS term $t_0$
satisfying property (c).
Let $J$ be the set of all $j\in [s]$ such that 
$\mathcal C(\mathbf A_{j})$ is not empty.
By condition (n), $J\neq \varnothing$.
Let $\mathcal U_0$ be the set of all 
algebras $\mathbf B\in \mathcal V_{p}$ such that 
$B\subsetneq A_{i}$ for some $i\in J$.
Then $\mathcal U_0\cup \{\mathbf A_{i}\mid i\in [s]\setminus J\}<\{\mathbf A_1,\dots,\mathbf A_{s}\}$ and, by condition (i) and Remark \ref{REMRemoveEllBlock},
$\mathcal U_0\cup \{\mathbf A_{i}\mid i\in [s]\setminus J\}$
admits a 
$(\underbrace{\overline{k^{n}},\dots,\overline{k^{n}}}_{m})$-PS term $t_1$.
Finally, by Lemma \ref{LEMPSInsideCenters}
there exists a $(\underbrace{\overline{k},\dots,\overline{k}}_{n},\ell)$-block term $t_2$ 
for $\mathbf A_1,\dots,\mathbf A_{s'}$
satisfying properties 1 and 2 in Lemma \ref{LEMPSInsideCenters}.
Note that the only reason for requiring the algebras $\mathbf A_1,\dots,\mathbf A_{s}$ to be minimal Taylor is to ensure that every subalgebra $\mathbf B\in\mathcal C(\mathbf A_{i})$ is 
central, and conversely that every proper central subalgebra of $\mathbf A_{i}$ belongs to $\mathcal C(\mathbf A_{i})$ (by Lemma \ref{LEMTernaryAbsorptionInTM}). 
Thus, properties 1 and 2 in Lemma \ref{LEMPSInsideCenters} can be reformulated for 
subalgebras from $\mathcal C(\mathbf A_{i})$.
Then 
a weak
    $(\underbrace{\overline{k},\dots,\overline{k}}_{n},\ell)$-PS term $t$ 
    for 
    $\mathbf A_{1},\dots,\mathbf A_s$ satisfying property (c) can be obtained as follows.

First, for every $i\in[m]$ let the term  
$g_{i}(\mathbf x_1,\dots,\mathbf x_n,\mathbf u, y_1,\dots,y_n)$
be obtained from 
$t_{0}(\mathbf x_1,\dots,\mathbf x_n,\mathbf u)$ by replacing 
each $w(z_1,\dots,z_p)$ by 
$$f_{i}(y_1,\dots,y_n,w(f_{i}(y_1,\dots,y_n,z_1),\dots,f_{i}(y_1,\dots,y_n,z_p))).$$
Thus, 
we add $f_{i}(y_1,\dots,y_n,z)$ to every variable $z$ to reduce the corresponding domain $\mathbf A_{j}$ for $j\in I$.
Then put 
\begin{align*}
t_3(\mathbf x_1, \dots,\mathbf x_n,\mathbf u):=
t_1((g_{1}(\mathbf x_1, 
\dots,\mathbf x_n,\mathbf u,\mathbf x_1^{j_1},\dots,\mathbf x_n^{j_n}))_{j_1,\dots,j_n\in[k]},\dots,&\\
\dots, (g_{m}(\mathbf x_1, \dots,\mathbf x_n,\mathbf u,\mathbf x_1^{j_1},\dots,&\mathbf x_n^{j_n}))_{j_1,\dots,j_n\in[k]}).
\end{align*}
Here, the first $k^{n}$ arguments of $t_1$ are filled with 
$g_{1}$ where we replace the variables
$y_1,\dots,y_n$ by 
$\mathbf x_1^{j_1},\dots,\mathbf x_n^{j_n}$ for $j_1,\dots,j_n\in[k]$.
The next $k^{n}$ arguments are filled with $g_{2}$ and so on.

Finally, to define $t$ we replace $\mathbf u$ in $t_{3}$ by 
the term $t_2$ applied to all possible permutations of the corresponding arguments.
Formally, 
let $\{(\sigma_{i,1},\dots,\sigma_{i,n+1})\mid i\in [M]\}$ be the set 
of all tuples of permutations, where 
$\sigma_{i,1},\dots,\sigma_{i,n}$ are permutations on $[k]$ 
and 
$\sigma_{i,n+1}$ is a permutation on $[\ell]$.
Then 
$$t(\mathbf x_1,\dots,\mathbf x_{n+1}) := 
t_{3}(\mathbf x_1,\dots,\mathbf x_{n},
t_2(\mathbf x_{1}^{\sigma_{1,1}},\dots,\mathbf x_{n+1}^{\sigma_{1,n+1}}),\dots,t_2(\mathbf x_{1}^{\sigma_{M,1}},\dots,\mathbf x_{n+1}^{\sigma_{M,n+1}})).$$
Let us show that $t$ is in fact a weak $(\overline{k},\dots,\overline{k},\ell)$-PS term for 
$\mathbf A_1,\dots,\mathbf A_s$ satisfying property (c).

First, let us show property (c).
Suppose $(\mathbf a_1,\dots,\mathbf a_{n},\mathbf a_{n+1})$
is a self-contained $(\overline{k},\dots,\overline{k},\ell)$-palette tuple on $\mathbf A_{i}$  
such that 
$\mathbf a_{n+1}\in B^{\ell}$ and $B\in\mathcal C(\mathbf A_{i})$.
By properties of $t_2$, 
we have 
$t_2(\mathbf a_{1}^{\sigma_{i',1}},\dots,\mathbf a_{n+1}^{\sigma_{i',n+1}})\in B$ for every $i'\in[M]$.
By property (p)
and Lemma \ref{LEMCenterGoesInside}, 
$h(B)\in\mathcal C(h(\mathbf A_{i}))\cup\{h(A_{i})\}$ 
for any $c_1,\dots,c_n\in A_{i}$,
where $h(x)\approx f(c_1,\dots,c_n,x)$.
By Lemma \ref{LEMSqueezForGenerating} and condition (s), $|h(\mathbf A_{i})|<|A_{i}|$ if $i\in I$, and by condition (k) we have $h(\mathbf A_{i})= A_{i}$ if $i\notin I$.
Therefore, $h(\mathbf A_{i}) \in \mathcal U\cup \{\mathbf A_{i}\mid i\in [s]\setminus I\}$.
Since the original tuple is self-contained, 
the tuple we substitute in $t_{0}$ is palette, 
which is exactly the reason why we only prove the existence of a weak(!) palette term.
Since $t_0$ satisfies property (c)  
on $\mathcal U\cup \{\mathbf A_{i}\mid i\in [s]\setminus I\}$, 
each $g_{q}^{\mathbf A_{i}}$ returns an element
of $B$. Since $B$ is a subalgebra of $\mathbf A_{i}$, 
$t_1^{\mathbf A_{i}}$ also returns an element of $B$. 
Hence $t^{\mathbf A_{i}}(\mathbf a_1,\dots,\mathbf a_{n},\mathbf a_{n+1})\in B$.

Let us show that $t$ is 
$(\overline{k},\dots,\overline{k},\ell)$-PS term for 
each $\mathbf A_{i}$. 
We start with the case when $i\in [s]\setminus I$.
In this case $t_{0}$ is $(\underbrace{\overline{k},\dots,\overline{k}}_{n},M)$-PS in $\mathbf A_{i}$.
By property (k) and property 2 of squeezers,  $f_{q}(c_1,\dots,c_n,x)\approx x$ for any $q\in[m]$ and $c_1,\dots,c_n\in A_{i}$, 
hence
the value $g_{q}^{\mathbf A_i}(\mathbf a_1, \dots,\mathbf a_n,\mathbf u,y_1,\dots,y_n)$ only 
depends on $\mathbf a_1, \dots,\mathbf a_n,\mathbf u$, where $\mathbf u$ depends only on 
 $\mathbf a_1,\dots,\mathbf a_{n+1}$. 
Since $t_{0}$ is $(\underbrace{\overline{k},\dots,\overline{k}}_{n},M)$-PS in $\mathbf A_{i}$, 
permutations in $\mathbf a_1,\dots,\mathbf a_n,\mathbf a_{n+1}$ do not change this value.
Thus, all the values we substitute in $t_1^{\mathbf A_{i}}$ are the same, 
and the idempotency of $t_1^{\mathbf A_{i}}$ implies that $t$ is weak
$(\overline{k},\dots,\overline{k},\ell)$-PS in $\mathbf A_{i}$.

Let us prove that $t$ is  $(\overline{k},\dots,\overline{k},\ell)$-PS in $\mathbf A_{i}$ for $i\in I$.
Since the tuple $(\mathbf a_1,\dots,\mathbf a_n,\mathbf a_{n+1})$ is $(\overline{k},\dots,\overline{k},\ell)$-palette,  
$(\mathbf a_1^{j_1},\dots,\mathbf a_n^{j_n})$ generates the same subalgebra of $\mathbf A_{i}$ as $(\mathbf a_1,\dots,\mathbf a_n,\mathbf a_{n+1})$
for any $j_1,\dots,j_n\in[k]$.
Let us fix $q\in[m]$ and $j_1,\dots,j_n\in[k]$ and define 
$h(x) \approx f_{q}(\mathbf a_1^{j_1},\dots,\mathbf a_n^{j_n},x)$.
Let 
$\mathbf B=
\Sg_{h(\mathbf A_{i})}(\{h(\mathbf a_1^{j_1}),\dots,h(\mathbf a_n^{j_n})\})$. 
Since $i\in I$,  
property (s) and Lemma \ref{LEMSqueezForGenerating} imply that 
$B\subsetneq A_{i}$.
Since $\mathbf B\in \mathcal U\cup \{\mathbf A_{i}\mid i\in [s]\setminus I\}$,
the term $t_0$ is a 
$(\overline{k},\dots,\overline{k},M)$-PS term on $\mathbf B$
satisfying property (c).
Thus, the choice of $q\in[m]$ and $j_1,\dots,j_n\in[k]$ 
determines the function $h$ and the function $h$ determines the algebra $\mathbf B$.
Since $t_0$ is a
$(\overline{k},\dots,\overline{k},M)$-PS on each $\mathbf B$, 
permutations in  blocks $\mathbf a_1,\dots,\mathbf a_n, \mathbf a_{n+1}$
only imply 
permutations in the corresponding blocks of size $k^{n}$ of $t_1$
but the content of each block does not depend on the permutations.

Let us show that the tuple 
we substitute in $t_1^{\mathbf A_{i}}$ is 
$(\overline{k^{n}},\dots,\overline{k^{n}})$-palette. 
Any value $c$ that ever appears in this tuple 
comes from some $g_{q}^{\mathbf A_{i}}(\mathbf a_1, 
\dots,\mathbf a_{n},\mathbf u,\mathbf a_1^{j_1},\dots,\mathbf a_n^{j_n})=c$.
Since conditions (1) and (2) in Lemma \ref{LEMMagicListOfTerms} are equivalent, 
there exists $d\in[m]$ such that 
$g_{d}^{\mathbf A_{i}}(\mathbf a_1, 
\dots,\mathbf a_{n},\mathbf u,\mathbf a_1^{i_1},\dots,\mathbf a_n^{i_n})=
c$
for all $i_1,\dots,i_n\in[k]$.
This implies that 
the $d$-th block of the tuple we substitute into $t_1^{\mathbf A_{i}}$
contains only value $c$, which means that it is palette.

Thus, the tuple we substitute in 
$t_1^{\mathbf A_{i}}$ is 
$(\overline{k^{n}},\dots,\overline{k^{n}})$-palette 
and permutations in 
$\mathbf a_1,\dots,\mathbf a_n,\mathbf a_{n+1}$ only imply permutations 
in the corresponding blocks of size $k^{n}$.

If $i\notin J$, then 
$t_1$ is $(\overline{k^{n}},\dots,\overline{k^{n}})$-PS in $\mathbf A_{i}$.
Therefore, $t$ is weak
$(\overline{k},\dots,\overline{k},\ell)$-PS in $\mathbf A_i$.

Assume that $i\in J$
and 
the tuple 
$(\mathbf a_{1},\dots,\mathbf a_{n+1})$ does not generate $\mathbf A_{i}$.
Since $t_1$ is $(\overline{k^{n}},\dots,\overline{k^{n}})$-PS in any proper subalgebra of $\mathbf A_{i}$, 
the term $t$ is weak
$(\overline{k},\dots,\overline{k},\ell)$-PS in $\mathbf A_i$.
Finally, assume that $i\in J$
and the tuple 
$(\mathbf a_{1},\dots,\mathbf a_{n+1})$ generates $\mathbf A_{i}$.
Then there exists $B\in\mathcal C(\mathbf A_{i})$ 
such that 
$t_2(\mathbf a_{1}^{\sigma_{1,i'}},\dots,\mathbf a_{n+1}^{\sigma_{n+1,i'}})\in B$
for any $i'\in [M]$.
Since each $f_{q}$ right-preserves $\mathcal C(\mathbf A_{i})$ in $\mathbf A_{i}$ 
and $t_{0}$ satisfies property (c), 
all the elements we substitute into $t_{1}$ are from $B$.
Since $t_1$ is $(\overline{k^{n}},\dots,\overline{k^{n}})$-PS on 
$\mathbf B$, $t$ is weak $(\overline{k},\dots,\overline{k},\ell)$-PS in $\mathbf A_i$.
\end{proof}

\subsection{Main result}\label{SUBSECTIONMainProof}

Here we combine the results of the previous section and show that 
we can recursively build 
PS term operations using  
PS term operations on smaller algebras unless our algebras are Maltsev, 
nilpotent, and large enough.

\begin{thm}\label{THMMainMinimalExampleFor}
Suppose $\mathbf A_1,\dots,\mathbf A_s\in\mathcal V_{p}$ are minimal Taylor, 
where $p$ is a prime and $|A_{i}|<p$ for every $i\in [s]$, and
\begin{enumerate}
    \item[(i)] for every 
    $\{\mathbf B_{1},\dots,\mathbf B_u\}<\{\mathbf A_{1},\dots,\mathbf A_{s}\}$ 
    and every $m,n\in\mathbb N$ there exists  a 
    $(\underbrace{\overline{p^{n}},\dots,\overline{p^{n}}}_{m})$-PS term 
    for 
    $\mathbf B_{1},\dots,\mathbf B_{u}$;
    \item[(c)] 
    $\mathbf A_1,\dots,\mathbf A_s$ do not admit 
    a 
    $(\underbrace{\overline{p^{n}},\dots,\overline{p^{n}}}_{m})$-PS term
    for some $m, n\in\mathbb N$.
\end{enumerate}
Then the following conditions hold:
\begin{enumerate}
    \item[(1)] there exists a term $m$ such that 
    $m^{\mathbf A_{i}}$ is a Maltsev operation for every $i\in [s]$;
    \item[(2)] any $\wedge$-irreducible congruence 
    of $\mathbf A_{i}$ is perfect linear for any $i\in[s]$;
    \item[(3)] for any $i,j\in[s]$ and any $\wedge$-irreducible congruences 
    on $\mathbf A_{i}$ and $\mathbf A_{j}$, respectively,
    there exists a bridge between them;
    \item[(4)]  $|A_{i}|\ge 8$ for some $i$.
\end{enumerate}
\end{thm}

\begin{proof}
First, choose a minimal set of algebras $\{\mathbf B_{1},\dots,\mathbf B_u\}\le \{\mathbf A_{1},\dots,\mathbf A_{s}\}$  
such that for some $n,\ell,m\in\mathbb N$
 there does not exist a 
 $(\underbrace{\overline{p^{n}},\dots,\overline{p^{n}}}_{m},\ell)$-PS term in $\mathbf B_1,\dots,\mathbf B_u$ satisfying property (c).     
To make it clear, we assume that either $\{\mathbf B_{1},\dots,\mathbf B_u\}<\{\mathbf A_{1},\dots,\mathbf A_{s}\}$ in our partial order, 
or $\{\mathbf B_{1},\dots,\mathbf B_u\} =\{\mathbf A_{1},\dots,\mathbf A_{s}\}$.
Such a set of algebras exists because 
there is no such term for $\mathbf A_{1},\dots,\mathbf A_{s}$, for some $m,n\in\mathbb N$,
and $\ell = p^{n}$.

 By Lemma \ref{LemFindTMinProduct}, 
 choose a $p$-ary cyclic term $t$ 
 such that $\mathbf B_{i}':=(B_{i}; t^{\mathbf B_{i}})$ is a minimal Taylor algebra for every $i\in[u]$.
 Since $\mathbf B_{i}'$ contains fewer term operations, the property (c) only becomes stronger 
 for $\mathbf B_{1}',\dots,\mathbf B_{u}'$.
 In fact, since
 $\mathcal C(\mathbf B_{i}')\supseteq \mathcal C(\mathbf B_{i})$ 
 for every $i\in[u]$, 
 there does not exist a 
$(\underbrace{\overline{p^{n}},\dots,\overline{p^{n}}}_{m},\ell)$-PS term in $\mathbf B_1',\dots,\mathbf B_u'$ satisfying property (c).
Notice that for every $j\in[u]$ 
there should be a tuple 
$(a_1,\dots,a_{n})$ generating $\mathbf B_{j}'$ as otherwise 
we could replace $\mathbf B_{j}'$ by all its subalgebras and get a smaller 
set of algebras without PS term operations.
Consider three cases:

Case 1. $\{c\}<_{\TC} \mathbf B_{j}'$ for every $j\in[u]$ and $c\in B_{j}'$.
By Lemma \ref{LEMCenterIntersectSubalgebra}, 
$\{c\}<_{\TC} \mathbf S$ for any subalgebra $\mathbf S$ of $\mathbf B_{j}'$ 
and $c\in S$.
Then Lemma \ref{LEMPSInsideCenters} applied to 
$\mathbf B_{1}',\dots,\mathbf B_{u}'$ and their subalgebras 
gives $(\underbrace{\overline{p^{n}},\dots,\overline{p^{n}}}_{m},\ell)$-PS term in $\mathbf B_1',\dots,\mathbf B_u'$ satisfying property (c), 
which contradicts our assumptions.

Case 2. $\mathcal C(B_{i}')$ is not empty for some $i\in [u]$ 
but 
$\{c\}<_{\TC} \mathbf B_{j}'$ for some $j\in[u]$ and $c\in B_{j}$.
Consider a family of terms $\mathcal F$ defined in Corollary \ref{CORCentralTerms} for $\mathbf B_{1}',\dots,\mathbf B_{u}'$.
By Lemma \ref{LEMCenterImpliesSqueezing} 
there exists $t\in \mathcal F$ 
such that 
$\{t^{\mathbf A_i}(c_1,\dots,c_{n},b)\mid b\in A_{j}\}\neq A_{j}$
for a tuple  $(c_1,\dots,c_{n})$ generating $\mathbf A_{j}$.
Then we can check that 
the algebras $\mathbf B_1', \dots,\mathbf B_{u}'$ 
together with the family of terms $\mathcal F$ 
satisfy all the conditions of Theorem \ref{THMNotMinimalExampleTernaryAbs}.
Then 
there exists a weak  $(\underbrace{\overline{p^{n}},\dots,\overline{p^{n}}}_{m},\ell)$-PS term in $\mathbf B_1',\dots,\mathbf B_u'$ satisfying 
property (c).
Lemma \ref{LEMWeaKPSImpliesPS} implies that 
there exists a  $(\underbrace{\overline{p^{n}},\dots,\overline{p^{n}}}_{m},\ell)$-PS term in $\mathbf B_1',\dots,\mathbf B_u'$ satisfying 
property (c), which contradicts 
our assumption that 
there does not exist one for $\mathbf B_1',\dots,\mathbf B_u'$.

Case 3. $\mathcal C(B_{i}')$ is empty for all $i\in [u]$.
Then property (c) is void, and there does not exist
$(\underbrace{\overline{p^{n}},\dots,\overline{p^{n}}}_{m})$-PS term in $\mathbf B_1',\dots,\mathbf B_u'$ 
as adding $\ell$ dummy variables to it would give 
a 
$(\underbrace{\overline{p^{n}},\dots,\overline{p^{n}}}_{m},\ell)$-PS term
satisfying property (c).
By condition (i) $\{\mathbf B_{1}',\dots,\mathbf B_{u}'\}$ has to coincide 
with $\{\mathbf A_{1},\dots,\mathbf A_{s}\}$.
Notice that 
if $0_{\mathbf A_{i}}$ is not $\wedge$-irreducible for some $i\in[s]$ 
then we can replace 
$\mathbf A_{i}$ by 
$\mathbf A_{i}/\omega_1,\dots, \mathbf A_{i}/\omega_m$, 
where $\omega_1\cap\dots\cap \omega_m = 0_{\mathbf A_{i}}$, 
and use the fact that we have a PS term for all 
smaller sets of algebras.
Thus, we assume that $0_{\mathbf A_{i}}$ is $\wedge$-irreducible for every $i\in[s]$.
Since $\mathcal C(\mathbf A_{j})$ is empty and $\mathbf A_{j}$ is a minimal Taylor algebra, it has no central or binary absorbing subuniverses.
By Lemma \ref{LEMUbiquity} we have one of the two subcases:

Subcase 3A. There exists a maximal linear congruence 
$\sigma$ on some $\mathbf A_{i}$.
By Lemma \ref{LEMLInearOnTheTopIsEasy}
$\mathbf A_{i}/\sigma \cong \mathbf Z_{q}$ for some prime $q$.
Consider a family of terms $\mathcal F$ defined in Corollary \ref{CORLinearTerms} for the chosen $\sigma$.
We want to get a contradiction with Theorem \ref{THMNotMinimalExampleCongruence} for the case when $I\cup J = [s]$
and, therefore, a void condition ($\ell$).
It is sufficient to check that $J\neq \varnothing$.
By property (s) in Lemma \ref{LEMPreservingSetsIsGood}, this is equivalent to 
the existence of $j\in [s]$, $t\in \mathcal F$, and $a_1,\dots,a_n\in A_{j}$ such that $|\{t^{\mathbf A_{j}}(a_1,\dots,a_n,b)\mid b\in A_{j}\}|<|A_{j}|$.
Then Lemma \ref{LEMNoMaltsevImpliesSqueezing} 
implies that 
there exists a term $m$ such that $m^{\mathbf A_j}$ is a Maltsev operation for all $j\in[s]$.
Corollary \ref{CORLinearGivesSquizing} implies that there exists
a bridge between $\sigma$ and any $\wedge$-irreducible congruence on any $\mathbf A_{j}$.
Note that $\sigma$, and therefore
any $\wedge$-irreducible congruence on any $\mathbf A_{j}$, is perfect linear.
To complete this case it is sufficient to show that $|A_{j}|\ge 8$ for some $j\in[s]$. Assume the contrary.
Choose some $j\in[s]$.
By Lemma \ref{LEMTrivialClassOfPerfectLinearImpliesBA} 
each equivalence class of $0_{\mathbf A_{j}}^{+}$ contains $q$ elements. 
Therefore $A_{j}/0_{\mathbf A_{j}}^{+}$ contains at most 3 elements.
If it contains 2 elements then $q=2$ and 
$\mathbf A_{j}/0_{\mathbf A_{j}}^{+} \cong \mathbf Z_2$.
If it contains 3 elements and $0_{\mathbf A_{j}}^{+}$ is $\wedge$-irreducible, then 
$\mathbf A_{j}/0_{\mathbf A_{j}}^{+} \cong \mathbf Z_3$, 
which implies that $A_{j}$ has at least $3\cdot 3 =9$ elements. 
If it contains 3 elements but $0_{\mathbf A_{j}}^{+}$ is not $\wedge$-irreducible, 
then there are congruences $\sigma_1$ and $\sigma_2$ on $\mathbf A_{j}$ 
such that
$0_{\mathbf A_{i}}^{+}= \sigma_1 \cap \sigma_2$ and
$\sigma_1\circ \sigma_2\neq \sigma_2\circ \sigma_1$, 
which contradicts the existence of a Maltsev term \cite{sankappanavar1981course}.
Thus, we showed that each $\mathbf A_{j}$ is either isomorphic to 
$\mathbf Z_q$ or isomorphic to an algebra from 
$\mathbf Z_{q}\boxtimes \mathbf Z_{q}$.
Let us consider an XY-symmetric term of arity $k^{n}$, which exists by Theorem \ref{THMMinimalTaylorClassification}.
By Lemma \ref{LEMXYSymmetricIsPSForSmall} 
this term is also symmetric on each $\mathbf A_{i}$.
By adding dummy variables we get a 
$(\overline{k^{n}},\dots,\overline{k^{n}})$-PS term for 
$\mathbf A_{1},\dots,\mathbf A_s$, which contradicts our assumptions.

Subcase 3B. Every maximal congruence on each 
$\mathbf A_{i}$ is a PC congruence.
Consider a family of terms $\mathcal T$ defined in Lemma \ref{LEMPCTerms}.
Since $0_{\mathbf A_{i}}$ is $\wedge$-irreducible 
for every $i$, 
property (s) in Lemma \ref{LEMPCTerms}
implies that 
for every $i\in[s]$ either there exist
    $t\in\mathcal T$ and a tuple 
    $(a_1,\dots,a_n)$ generating $\mathbf A_{i}$ such that 
    $|\{t^{\mathbf A_{i}}(a_1,\dots,a_n,b)\mid b\in A_{i}\}|<|A_{i}|$, 
    or $\mathbf A_{i}$ is a PC algebra.
Let $\mathbf A_{v}$ be a maximal algebra in the set $\{\mathbf A_{1},\dots,\mathbf A_{s}\}$. 
If $\mathbf A_{v}$ is a PC algebra, then we get a contradiction using Theorem \ref{THMReductionForPCCongruence}.
Otherwise, choose some maximal congruence $\sigma$ on $\mathbf A_{v}$.
Then we can get a contradiction using Theorem \ref{THMNotMinimalExampleCongruence} applied to 
the family $\mathcal T$ and the congruence $\sigma$, because 
$v\in J$ and condition ($\ell$) is satisfied.
\end{proof}

The main result of this section is now a simple corollary.

\begin{THMExistenceWBSForSmallAlgebrasTHM}
Suppose $\mathbf A_{1},\dots,\mathbf A_{s}$ are similar algebras
admitting WNU term operations, 
$|A_{i}|<8$ for every $i\in[s]$. 
Then 
$\mathbf A_{1}\times\dots\times \mathbf A_{s}$ has 
a $(\underbrace{p^{m},\dots,p^{m}}_{n})$-palette symmetric term operation 
for every prime $p>8$ and $m,n\in\mathbb N$.
\end{THMExistenceWBSForSmallAlgebrasTHM}

\begin{proof}
Assume the converse. Choose a minimal set of similar Taylor algebras 
$\{\mathbf B_1,\dots,\mathbf B_u\}\le \{\mathbf A_1,\dots,\mathbf A_{s}\}$ such that 
they do not admit the corresponding palette symmetric term operation. 
By Lemma \ref{LemFindTMinProduct} we can choose a cyclic term $c$ 
such that 
$(B_{i};c^{\mathbf B_{i}})$ is a minimal Taylor algebra for every $i\in[u]$.
Then Theorem \ref{THMMainMinimalExampleFor} implies that 
$|B_{i}|\ge 8$ for some $i$, which contradicts our assumptions.
\end{proof}

\section*{Acknowledgments}

I am very grateful to Michael Kompatscher for showing that the dihedral group 
$\mathbf D_4$
 does not admit palette symmetric terms, which prevented me from attempting to prove their existence for another five years. I thank Antoine Mottet for pointing out that a gap in my proof could be filled using the famous 
 Hales-Jewett theorem. 
 I am also grateful to Standa \v Zivn\'y and Lorenzo Ciardo for explaining to me the minion characterization of singleton algorithms. 
 Finally, I thank Jakub Rydval for motivating me to find palette polymorphisms for temporal relational structures;
 this simple exercise led us to discover a structure demonstrating the limitations of the Hales-Jewett argument and to better understand the relationship between singleton and constraint-singleton algorithms.
 
\bibliographystyle{plain}
\bibliography{refs}

\end{document}